\documentclass[pdflatex,sn-mathphys-num]{sn-jnl}% Math and Physical Sciences Numbered Reference Style
\usepackage{graphicx}%
\usepackage{multirow}%
\usepackage{amsmath,amssymb,amsfonts}%
\usepackage{amsthm}%
\usepackage{mathrsfs}
\usepackage{anyfontsize}
\usepackage[title]{appendix}%
\usepackage{textcomp}%
\usepackage{manyfoot}%
\usepackage{booktabs}%
\usepackage{url}
\usepackage{algorithm}%
\usepackage{algorithmicx}%
\usepackage{algpseudocode}%
\usepackage{listings}%
\usepackage{hyperref}
\hypersetup{bookmarksdepth=3}

\usepackage{xurl}

\usepackage{siunitx}
\usepackage{upquote}% upright quotes in code listings (minted removed: it needs shell escape, which arXiv disables)
\usepackage[x11names,table]{xcolor} 
\definecolor{LightGray}{gray}{0.98}
\usepackage{mathtools}

\DeclarePairedDelimiter\floor{\lfloor}{\rfloor}

\DeclareMathOperator{\ocol}{ocol}
\DeclareMathOperator{\place}{place}
\DeclareMathOperator{\gap}{gap}		
\DeclareMathOperator{\newgap}{newgap}
\DeclareMathOperator{\depth}{depth}

\usepackage{booktabs}
\usepackage{tabularx}
\usepackage{array}

\theoremstyle{thmstyleone}%
\newtheorem{theorem}{Theorem}%  meant for continuous numbers

\newtheorem{lemma}[theorem]{Lemma}%
\newtheorem{corollary}[theorem]{Corollary}%
\theoremstyle{thmstyletwo}%
\newtheorem{remark}{Remark}%

\theoremstyle{thmstylethree}%
\newtheorem{definition}{Definition}%

\begin{document}

\title{On the Computational Complexity of Blind Quantum Machine Learning}
%\title{Investigating the Computational Complexity of Blind Quantum Machine Learning with MasonQ: A Universal Brickwork Transpiler}

%%=============================================================%%
%% GivenName	-> \fnm{Joergen W.}
%% Particle	-> \spfx{van der} -> surname prefix
%% FamilyName	-> \sur{Ploeg}
%% Suffix	-> \sfx{IV}
%% \author*[1,2]{\fnm{Joergen W.} \spfx{van der} \sur{Ploeg} 
%%  \sfx{IV}}\email{iauthor@gmail.com}
%%=============================================================%%

\author[1]{\fnm{Rex} \sur{Fleur}}
\email{rexfleur@icloud.com}
 \equalcont{These authors contributed equally to this work.}

\author*[2,3,4]{\fnm{Tzula B.} \sur{Propp}}\email{propp@physics.leidenuniv.nl}
 \equalcont{These authors contributed equally to this work.}

% \equalcont{These authors contributed equally to this work.}

\author[1,2]{\fnm{Sebastian} \sur{Feld}}\email{s.feld@tudelft.nl}

\affil*[1]{\orgdiv{Quantum Circuits, Architectures and Technology (QCAT)}, \orgname{Delft University of Technology}, \orgaddress{\street{Mekelweg 4}, \city{Delft}, \postcode{2628 CD}, \country{The Netherlands}}}

\affil[2]{\orgdiv{QuTech}, \orgname{Delft University of Technology}, \orgaddress{\street{Lorentzweg 1}, \city{Delft}, \postcode{2628 CJ}, \country{The Netherlands}}}
\affil[3]{\orgdiv{Kavli Institute of Nanoscience}, \orgname{Delft University of Technology}, \orgaddress{\street{Lorentzweg 1}, \city{Delft}, \postcode{2628 CJ}, \country{The Netherlands}}}
\affil[4]{\orgdiv{Leiden Institute of Physics (LION)}, \orgname{Leiden University}, \orgaddress{\street{Niels Bohrweg 2}, \city{Leiden}, \postcode{2333 CA}, \country{The Netherlands}}}

\abstract{
Can quantum advantages in computational efficiency coexist with quantum advantages in privacy? In this paper, we construct and use  \emph{MasonQ}---a universal brickwork transpiler---to answer this question with a resounding yes. We approach this question by bridging two seemingly disparate paradigms: universal blind quantum computation (UBQC), based on measurement-based quantum computation, and quantum machine learning (QML) algorithms such as the Harrow-Hassidim-Lloyd (HHL) algorithm and quantum recommendation systems. 
We provide new upper bounds on the depth of UBQC brickwork graphs, establishing that \( D = \mathcal{O}(\omega n) \), where \( \omega \) is the circuit width and \( n \) its computational gate complexity. This is enabled by our newly developed transpiler MasonQ, which maps arbitrary quantum circuits to computational brickwork graph states -- the resource underpinning UBQC. Given an optimal Qiskit circuit as input, we prove optimality (up to logarithmic corrections) of the MasonQ transpiler in terms of asymptotic scaling of the total brickwork graph width and depth after decomposition into bricks. This enables systematic cost analyses of blind implementations, from near-term experiments to asymptotic regimes. %The proof is given in Appendix \ref{app:opt_graph}.

Leveraging MasonQ, we derive explicit depth bounds for several representative blind quantum machine learning (B-QML) algorithms, including the quantum Fourier transform, the HHL algorithm, recommendation systems, and quantum transformers, with brickwork depths scaling as \( \mathcal{O}(\log^3 N) \), \( \mathcal{O}(\mathrm{poly}(\kappa, s, 1/\varepsilon)\,\mathrm{polylog}(N)) \), \( \mathcal{O}(L^2 \log L) \), and \( \mathcal{O}\big((p+d)^{2}\,(\log p + p\log d + \log(p+d))\big) \), respectively, where \( \kappa \) denotes the condition number, \( s \) the sparsity, \( L \) the database size of the recommendation system, \( p \) the number of input patches (i.e., sequence length), \( d \) the dimension of each patch (i.e., embedding size), \( N \times N \) the size of the input matrix, and \( \varepsilon \) a tolerated error. Lastly, we analyze concrete resource estimates for near-term demonstrations of blind quantum machine learning. We propose a minimal B-QML implementation requiring only a \(12\) qubit quantum server. 

To the best of our knowledge, MasonQ is the first functional and optimal universal circuit-to-brickwork-graph transpiler to implement the full brickwork transpilation process from first principles.}


%%\pacs[JEL Classification]{D8, H51}

%%\pacs[MSC Classification]{35A01, 65L10, 65L12, 65L20, 65L70}

\maketitle

\section{Introduction}
%background
\noindent

% \todoTz{Since we will submit to a nature family paper, we should use their template not IEEE https://www.overleaf.com/latex/templates/springer-nature-latex-template/myxmhdsbzkyd}

Quantum computation enables computational speedups \cite{Grover1997,Harrow2009,Mantri2013,BermejoVega2018,Seddon2021,Liu2021,Pokharel2023,Jordan2025,BeneWatts2025} and reductions in energy consumption \cite{FellousAsiani2023,Meier2025,Hou2025,Abderrahim2025} for certain computational tasks. For both, the key is that quantum algorithms and hardware are \emph{efficient} at solving tasks in comparison to the best known classical approaches, with quantum machine learning (QML) applications being of particular contemporary interest \cite{Biamonte2017,GilFuster2024,Cho2024,Li2025}\footnote{While QML algorithms do not generally yield an exponential advantage over classical machine learners \cite{Tang_dequantisation_2019,Tang2023}, they can still give a computational advantage over their classical counterparts for certain problems \cite{Liu2021,Huang2021,Huang2022}.}. Simultaneously, universal blind quantum computation (UBQC) enables private delegated computation \cite{ubqc,Morimae2012,Mantri2017,Fitzsimons2017,Drmota2024}, with preliminary experiments performed on near-term hardware \cite{Polacchi2023,Drmota2024,Polacchi2025}. Here, the quantum privacy advantage is ensured by the fragility of quantum information (e.g., no-cloning theorem \cite{Wootters1982}) and the indistinguishability of computation rounds from test rounds, which are used to verify the behavior of untrusted quantum devices. In UBQC, indistinguishability is often ensured by measurement-based quantum computation performed on a planar brickwork graph.  

This raises an immediate question: are quantum advantages in computational and energy efficiency compatible with quantum privacy, without sacrificing the former? UBQC is universal, but not necessarily efficient. This is especially important in the noisy intermediate-scale quantum (NISQ) era, where limited qubit counts and noise preclude fault tolerant implementations of resource-demanding quantum algorithms. While there have been improvements to the achievable rates of remote qubit state preparation \cite{vanDam2024,vanDam2025,Propp2025,Propp2026}---a primitive for UBQC---there is a need for estimates for the sizes of graphs (and thus the number of remotely prepared qubit states) needed to implement algorithms, especially those used in QML.

In this paper, we answer the above question in the affirmative; we show that quantum circuits can be implemented efficiently (i.e. with a polylogarthmic depth overhead as a function of the input circuit size) within UBQC. This result is enabled by MasonQ, a provably efficient transpiler we here develop and introduce\footnote{Of course, the underlaying circuit may itself be inefficient; the efficiency we guarantee here is in terms of transpilation only.}. MasonQ functions by decomposing a circuit into single-row bricks (a line graph of four entangled qubits, Fig. \ref{fig:single_rot}) and two-row bricks (two single-row brick graphs connected with a CZ entangling gate on the third column, Fig \ref{fig:2cx}), both of which we refer to as bricks as each implements an elementary gate from our generator set $\mathcal{G}$. Inter-brick CZ gates are interwoven to combine individual bricks into larger brickwork graphs (\ref{fig:swap_cx}). In translating an $\omega$-qubit quantum circuit into a measurement pattern, we prove that, for all circuits, 
%\todoTz{What is the qualifier here? I don't think I appreciated this! I thought the bounds were universal. Rex: We prove for MCU MCX CX AND U2, not every gate that is possible. Maybe some crazy quantum gate will be invented that isnt "efficiently decomposable", say needs $n^2$ CX gates to implement, them the bound doesnt hold.. We can reword it but thats my rationale. Maybe we can do a proof by contradiction but Im not aware of theorems for decomposition that will help us here..}  
UBQC requires at most a brickwork graph of depth $\mathcal{O}(\omega n)$, where $n$ is the gate complexity (i.e. the scaling of the number of gates with input size) of the original circuit decomposed into multi-qubit controlled unitaries, multi-qubit controlled $X$ gates, and single-qubit rotations\footnote{Gates without such an efficient exact decomposition can be approximated with
polylogarithmic overhead in the inverse accuracy using the Solovay--Kitaev
algorithm~\cite{kitaev2002classical,dawson2006solovay}.}. Focusing on blind quantum machine learning (B-QML), where privacy is particularly well motivated by its reliance on data, we investigate whether privacy compromises potential quantum advantages in computational complexity, which will also impact energy consumption. Analyzing an $\mathcal{O}(\omega n)$ blinding overhead for the studied QML circuits, we demonstrate, that quantum speed advantages can be retained due to the polylogarthmic (in the unblinded circuit width) additional expense of making the computation blind. %\todoTz{Again, what is the note about representative here? I think it's worth making it explicit! If there is  a case that falls outside our analysis lets comment on it explicitly. Or is this just you being tentative about making claims ;) Rex: do you mean instead of saying representative, naming the algos explicitly?}

In this paper we introduce our universal brickwork transpilation software package, MasonQ, use it to derive the upper bounds on brickwork graph size, and develop a proof of principle B-QML experiment amenable to near term implementation on a quantum internet \cite{Wehner2018}. Section \ref{secMethod} introduces our transpiler used for empirical corroboration of analytical bounds and B-QML numerical experiment generation. Section \ref{sec:results} gives an overview of how the upper bounds were established, formulates a method to find the general cost of blindness, and proposes the minimal near-term B-QML experiment requiring \(12\) physical qubits, \(750\) remotely prepared qubit states, \(750\) single-qubit measurements and rotations, and \(900\) CZ gates. This near-term experiment implements a blind quantum recommendation system where a user queries a delegated server and, without revealing their preferences, receives a recommendation from one of two options based on Hamming distance with a two-bit feature i.e. description of the option. Section \ref{sec:conclusion} contains our discussion of the results and possible future directions of the work. Appendix \ref{app:statapp} gives the recursive, constructive proof of our transpiler's optimality, Appendix \ref{app:circuit_optimisations} compares initial mapping and routing strategies, Appendix \ref{app:bw_transpiler} provides further details of the transpilation software including noise modeling and visualization, Appendix \ref{app:experiment} contains additional data from our numerical study, and Appendix \ref{app:qml_analysis} gives the detailed analysis of B-QML protocols and, for the blind quantum recommendation system (B-QRS), gives a comparison with the encrypted (i.e. classically blinded) classical counterpart.
% \section{Introduction}\label{sec1}

\section{The MasonQ Transpiler} \label{secMethod}
\noindent

% Our primary tool to investigate the scaling of blind quantum algorithms is a transpiler. Mathematically, we define a transpiler as a function $
% \mathsf{T}:{C} \to {G}_{\mathrm{BW}}$
% mapping any quantum circuit $C\in\mathcal{C}$ to its corresponding computational brickwork graph ${G}_{\mathrm{BW}}\in\mathcal{G}$, including the associated measurement-dependent single-qubit rotations (correction dependencies) needed to implement arbitrary quantum circuits.  We developed such a transpiler that converts any Qiskit quantum circuit to a computational brickwork graph, available on \href{https://github.com/Nonsens24/brickwork_transpiler}{GitHub}. 
Our primary tool to investigate the scaling of blind quantum algorithms is a transpiler. Mathematically, we define a transpiler as a function
\[
\mathsf{T} : \mathcal{C} \rightarrow \mathcal{G}_{\mathrm{BW}},
\]
where \(\mathcal{C}\) denotes the set of quantum circuits and \(\mathcal{G}_{\mathrm{BW}}\) the set of computational brickwork graphs\footnote{These measurement graphs are also referred to as \textit{patterns} in literature~\cite{The_measurement_calc}.}. The transpiler maps any circuit \(C \in \mathcal{C}\) to a brickwork graph \(G_{\mathrm{BW}} \in \mathcal{G}_{\mathrm{BW}}\) that implements \(C\), together with the associated measurement-dependent single-qubit rotations (i.e., correction dependencies) required for determinism in measurement-based quantum computation (MBQC). We developed such a transpiler, which converts any Qiskit quantum circuit into its corresponding computational brickwork graph. The implementation is available on GitHub: \href{https://github.com/Nonsens24/brickwork_transpiler}{github.com/Nonsens24/brickwork\_transpiler}\footnote{This repository is currently private due to patenting, but access for academic researchers will be granted upon request.}. 

This section introduces our transpiler, which takes a decomposed circuit and produces an optimal mapping onto the brickwork graph. Appendix~\ref{app:opt_graph} provides a formal proof of optimality in terms of the scheduling of decomposed bricks (i.e. brickwork depth) using the decompositions into the universal gate-set $\mathcal{G}$ defined in Section~\ref{sec:graph_comp}. We also provide illustrative examples and a discussion of additional functionalities of MasonQ (optimization, visualization). Our transpiler is designed both for transpilation of general quantum algorithms for prospective users and also to support the analysis of the QML algorithms studied here. Crucially, our transpiler can determine brickwork graph dimensions independently of correction dependencies, enabling stronger scaling analyses by allowing larger graphs to be transpiled without computing all dependencies. 

Our transpiler, MasonQ, goes beyond existing transpilers in several key regards \cite{sunami2022graphix,mcbeth2023,joshi2026}: firstly, it is efficient, supporting multi-qubit controlled gates (specifically, non-adjacent CX gates) which aids in identifying low-depth circuit decompositions; secondly, it is tailor-made to support the brickwork graph structure for UBQC both in terms of transpilation and in terms of visualization; thirdly, it integrates Qiskit's native optimization for circuit compilation, so that future upgrades to the Qiskit compiler will automatically be included ~\cite{qiskit}. These are what enable our study of the brickwork depths of QML circuits presented in this paper.

\subsection{Graph composition}\label{sec:graph_comp}

The universality of brickwork graphs \cite{ubqc} follows from the composability of measurement graphs together with a universal gate-set. Subgraphs implementing elementary operations compose under concatenation (Theorem~1 of \cite{The_measurement_calc}), allowing larger patterns to be constructed by sequentially appending these building blocks, now referred to as bricks.

It therefore suffices to realize a universal generator set, as arbitrary circuits follow by concatenation. We consider
\begin{equation}
\begin{aligned}
\mathcal{G} &= \{ R_x(\theta), R_z(\theta), \mathrm{CX} \}, \\
R_x(\theta) &= e^{-i\theta X/2} \\
R_z(\theta) &= e^{-i\theta Z/2}
\end{aligned}
\end{equation}
where $\mathrm{CX}$ denotes the controlled-NOT gate.

Notably, decomposing arbitrary single-qubit rotations into $R_x(\theta)$ and $R_z(\theta)$ reduces the required qubit count by two relative to alternative constructions~\cite{parsimonious_qgates}.

Since $\mathcal{G}$ is universal~\cite{Nielsen_Chuang_2010}, any quantum circuit admits a decomposition into elements of $\mathcal{G}$ (up to global phase), yielding a brickwork-graph realization suitable for UBQC. The width of the brickwork graph will be identical to the number of qubits needed in the original circuit (up to a constant number of ancilla qubits, required to obtain a decomposition into $\mathcal{G}$ with linear overhead), but the depth could be much larger in principle. To prove the universality of our transpiler we will now define each measurement graph corresponding to a gate in the generator set $\mathcal{G}$.

\vspace{5mm}

\subsubsection{Single-Qubit Gates}
A single-row brick -- four horizontally entangled qubits in the brickwork graph -- implements an arbitrary $\mathrm{SU}(2)$ unitary via appropriate choices of measurement angles and adaptive feedforward corrections \cite{parsimonious_qgates}. Since $\mathrm{SU}(2)$ is closed under composition, any finite sequence of single-qubit gates can be compiled into a single $\mathrm{SU}(2)$ rotation and realized by one such four-qubit graph. In brickwork graphs, these graphs that implement single-qubit rotations are colored green, as illustrated in Fig.~\ref{fig:single_rot}.

% A single-row brick, that is four horizontally entangled qubits in the brickwork graph, can implement arbitrary SU(2) rotations \cite{parsimonious_qgates}. This means that \textbf{a sequence of an arbitrary amount of single-qubit gates, can always be implemented in such a four-qubit single-row brick}. Nodes of subgraphs implementing single-qubit rotations are coloured green in the brickwork graphs, as illustrated in Fig. \ref{fig:single_rot}.

% \todoTz{Tzula: We define a gadget as a brick or set of bricks implementing a desired gate. Or omit}

\begin{figure}[htbp]
\centering
\includegraphics[width=0.75\textwidth]{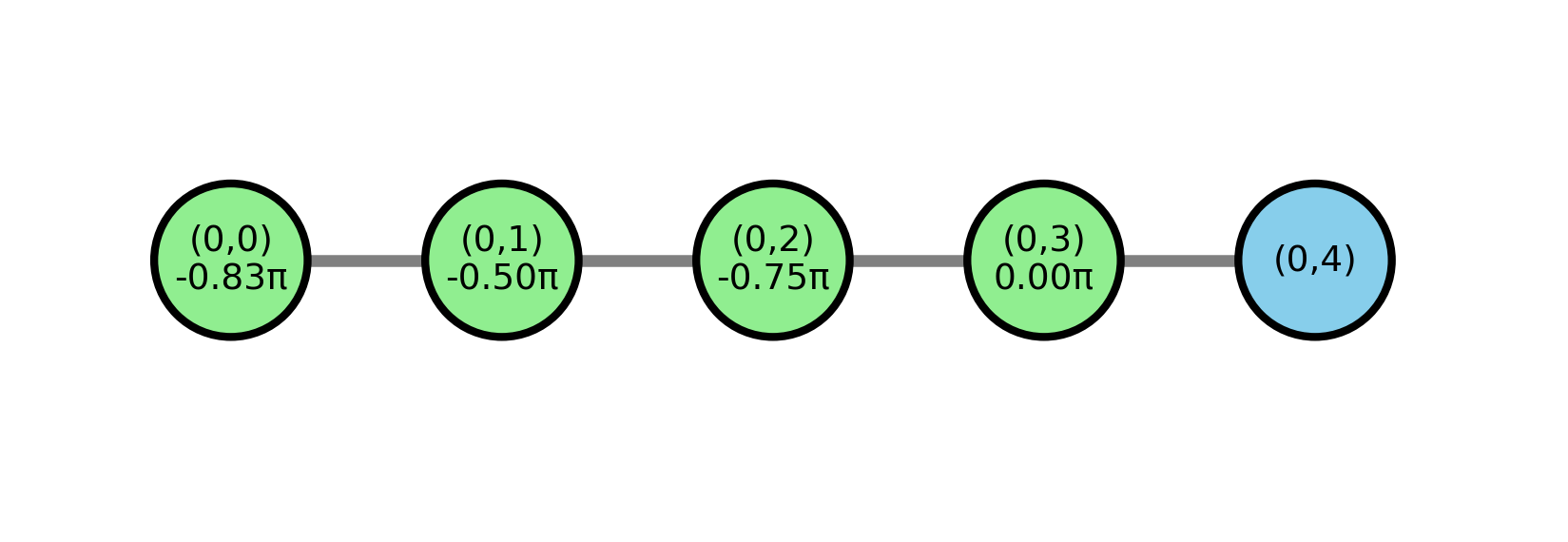}
\caption{Graphical representation of the composed single-qubit gates implementing a sequence of single-qubit gates ($H$, $R_x(\pi/3)$, and $T$) as a single single-row brick. The nodes and edges represent a graph that implements single-qubit rotations, labeled by their graph vertex index $(\text{row}, \text{column})$ and measurement angle. The blue node at the far right denotes the unmeasured output qubit of the measurement pattern containing the quantum state, and therefore its label contains no angle. Emphatically, while any sequence of single-qubit unitaries can be thus decomposed into a single brick, multi-qubit gates e.g. CX correspond to at least one brick in the graph.}
\label{fig:single_rot}
\end{figure}

\subsubsection{Adjacent \texorpdfstring{$\mathrm{CX}$}{CX} gates}

With $R_x(\theta)$ and $R_z(\theta)$ implemented, the remaining gate from the generator set $\mathcal{G}$ required for universality is $\mathrm{CX}$. Brickwork graphs impose connectivity constraints, allowing interactions only between qubits located in adjacent rows. As a result, transpiling $\mathrm{CX}$ operations between non-adjacent rows becomes nontrivial, at least for planar brickwork graphs \cite{Fitzsimons_Kashefi_verify_bqc}.

To address this, we distinguish between two cases: $\mathrm{CX}$ gates acting on adjacent qubits and those acting on non-adjacent qubits. This subsection focuses on the former, while the latter will be treated in subsequent sections. To formalize the notion of adjacency, we begin with the following definition.

\begin{definition}[Adjacency in a brickwork graph]\label{def:adjacent}
Two logical qubits (rows) in a brickwork graph are said to be \emph{adjacent} if and only if there exists a $\mathrm{CZ}$ edge directly connecting the corresponding vertices in the graph. Equivalently, the measurement pattern implements a $\mathrm{CZ}$ operation between these rows without requiring any permutation of the row ordering.
\end{definition}

The measurement angles implementing a $\mathrm{CX}$ gate between two logical qubits are specified in \cite{ubqc} and are used here without modification. To enhance the expressivity of the transpilation procedure, $\mathrm{CX}$ gates are supported in both orientations, namely with the control qubit above the target and vice versa. Figure~\ref{fig:2cx} illustrates these two configurations, encoded in the graph and distinguished by color to indicate their respective roles.

\begin{figure}[htbp]
    \centering

    \begin{minipage}[t]{0.48\textwidth}
        \centering
        \includegraphics[width=\linewidth]{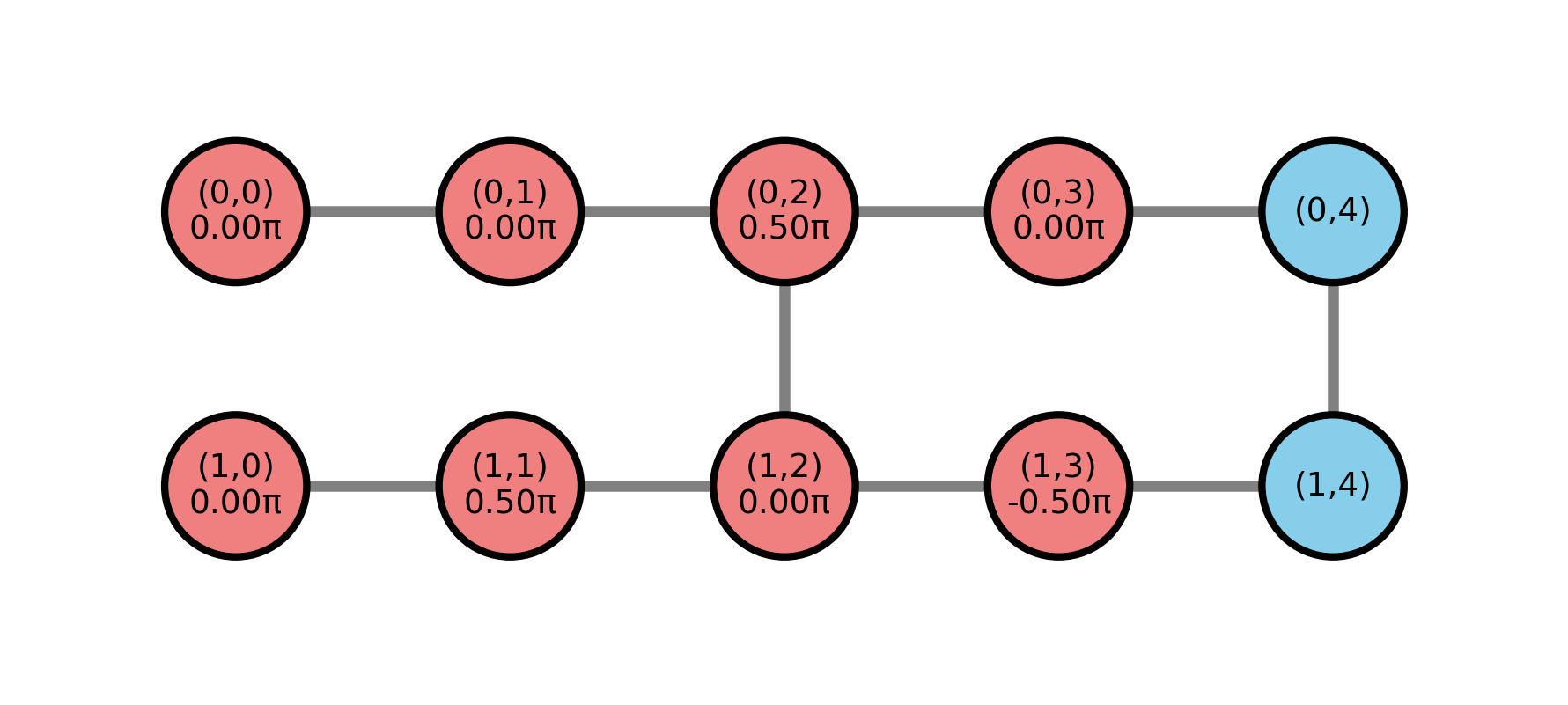}\\[-0.2em]
        \small (a)
    \end{minipage}
    \hfill
    \begin{minipage}[t]{0.48\textwidth}
        \centering
        \includegraphics[width=\linewidth]{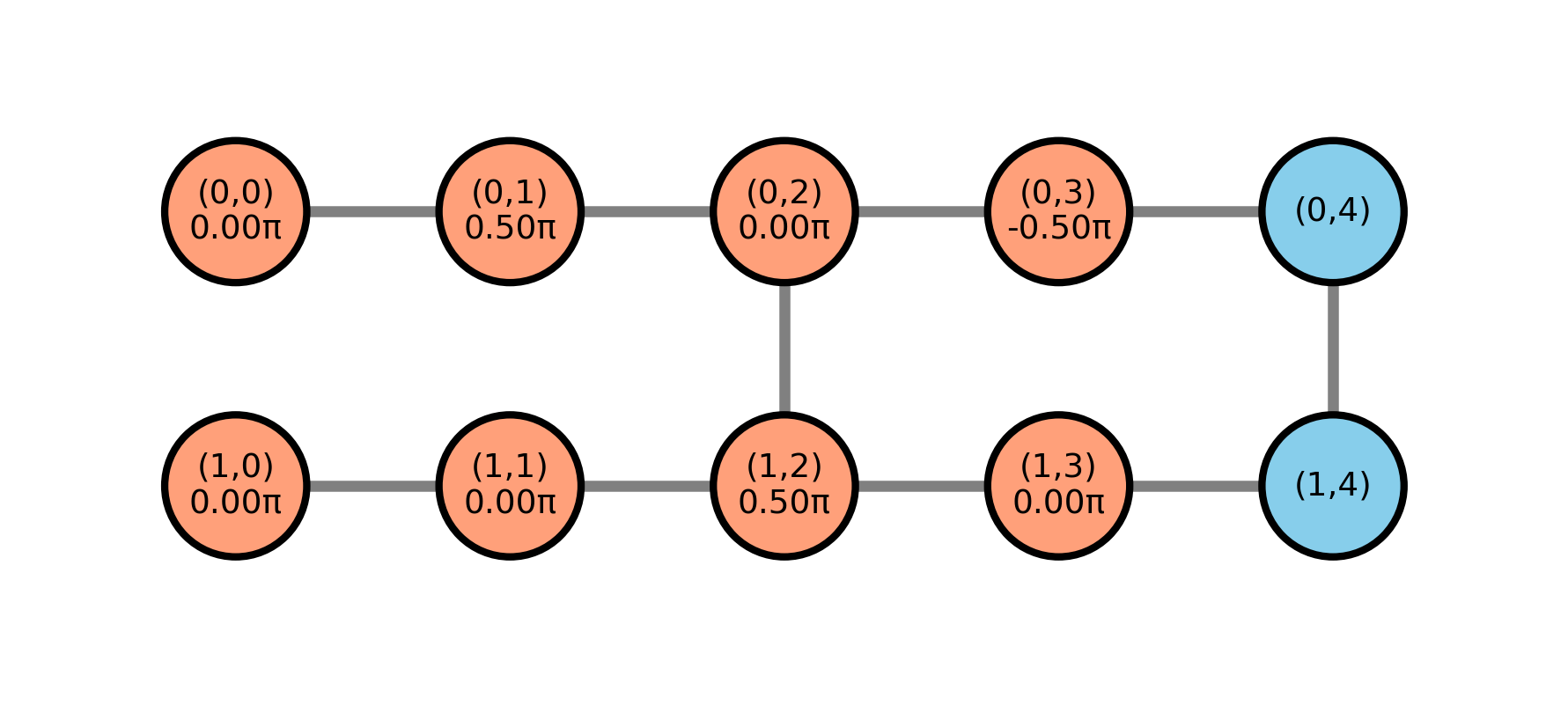}\\[-0.2em]
        \small (b)
    \end{minipage}

    \caption{
    \textbf{(a)} $\mathrm{CX}$ gate encoded in a single two-row brick with the logical control qubit on the upper row and the logical target qubit on the lower row.
    \textbf{(b)} The same construction with control and target interchanged.
    Control and target qubits are distinguished by color; we color the brick red to denote a lower-qubit target and upper-qubit control, and color the brick orange to denote the reversed case. 
    }
    \label{fig:2cx}
\end{figure}

We refer to the graphs in Figs.~\ref{fig:single_rot} and~\ref{fig:2cx} as \emph{bricks}, which, together with the identity brick (denoted with light blue nodes, see Appendix \ref{app:id_brick} and Fig. \ref{fig:id_brick}), are combined to build the brickwork graph.

%\todoTz{Why are the identity bricks linked by entanglement? Is this correct, or a mistake? I didn't think they needed a second entanglement link at the end. If it's not a mistake, a sentence explaining would be nice. Rex: The rightmost nodes here in blue represent the output qubits containing the state after computation of the left hand side of the graph. Two things: 1: the vertical CZ on these CX bricks are of no consequence and alternate (see fig 3) 2: The identity brick is at least 4 horizontally connected nodes that implement HHHH = II = I. There is no isolated image of this in the paper but you can see them in fig 3 between the CX bricks. Do you want me to add them? or is a sentence sufficient? -- for now I added a small appendix entry on it!}

% \subsubsection{Non-adjacent $\mathrm{CX}$ Gates}
% \subsubsection{\texorpdfstring{Non-adjacent $\mathrm{CX}$ Gates}{Non-adjacent $\mathrm{C}X$ Gates}}
\subsubsection{\texorpdfstring{Non-adjacent $\mathrm{CX}$ Gates}{Non-adjacent CX Gates}}

By Definition~\ref{def:adjacent}, non-adjacent logical qubits in a brickwork graph are not connected by a $\mathrm{CZ}$ edge between the corresponding rows. Consequently, the encoding of adjacent $\mathrm{CX}$ gates cannot be applied directly. Moreover, the rigid structure of the graph is essential for blindness in planar circuits \cite{Ma2024}, as permuting the logical-qubit ordering during execution would reveal structural information. Thus to implement a non-adjacent $\mathrm{CX}$ without violating blindness constraints, we swap logical qubit states until the relevant logical qubits become adjacent, the adjacent $\mathrm{CX}$ is applied, and the states are swapped back (if required by the larger algorithm to return the qubit states to their original placements; if a permuted output is permissible, this is always cheaper). Since $\mathrm{SWAP}_{1,2}=(\mathrm{CX}_{12}\,\mathrm{CX}_{21}\,\mathrm{CX}_{12})$ can be realized using three $\mathrm{CX}$ gates~\cite{Nielsen_Chuang_2010}, our transpiler decomposes any non-adjacent $\mathrm{CX}$ into a sequence of adjacent $\mathrm{CX}$ operations that effect the necessary swaps, followed by the desired adjacent $\mathrm{CX}$ (and inverse SWAP operations, if required). The direction of the SWAP operations is determined by the routing algorithm provided as input to the transpiler; see Appendix~\ref{app:circuit_optimisations} for examples.

Figure~\ref{fig:swap_cx} illustrates the SWAP-based routing procedure. The circuit in Fig.~\ref{fig:swap_cx}a, containing a non-adjacent $\mathrm{CX}$ gate, is decomposed into adjacent $\mathrm{CX}$ gates (Fig.~\ref{fig:swap_cx}b). Specifically, Fig.~\ref{fig:swap_cx}b comprises four operations: an initial $\mathrm{CX}$ (eight red nodes, see Fig.~\ref{fig:swap_cx}c), followed by an identity brick (i.e. a brick implementing the (trivial) identity gate, see Appendix \ref{app:id_brick}), a reoriented $\mathrm{CX}$ (orange), again an identity brick, and then two more $\mathrm{CX}$ gates that complete the SWAP and implement the desired $\mathrm{CX}$. This sequence is then mapped to the brickwork pattern in Fig.~\ref{fig:swap_cx}c using the generator set $\mathcal{G}$. The necessity of these identity bricks is due to an alignment problem i.e. that CX gates cannot be placed arbitrarily in the brickwork graph topology, necessitating the insertion of identity bricks as detailed in Appendix \ref{app:statapp}. 

% Figure~\ref{fig:swap_cx} illustrates the SWAP-based routing procedure. The circuit in Fig.~\ref{fig:swap_cx}a, which contains a non-adjacent $\mathrm{CX}$ gate, is first decomposed into a sequence of adjacent $\mathrm{CX}$ gates (Fig.~\ref{fig:swap_cx}b). In particular, Fig.~\ref{fig:swap_cx}b comprises four gates: a $\mathrm{CX}$ represented by the 8 red nodes, followed by an identity brick, which in turn is followed by an reoriented $\mathrm{CX}$ (whose nodes are shown in orange), and finally two further $\mathrm{CX}$ gates, the first finishing the swap and the second implementing the intended $\mathrm{CX}$ gate. This decomposition is then mapped to the corresponding brickwork pattern shown in Fig.~\ref{fig:swap_cx}c composing the bricks supported by the generator set $\mathcal{G}$. Light-blue nodes represent identity (wire) bricks, which implement state teleportation without modifying the quantum state.

% Figure~\ref{fig:swap_cx} illustrates the SWAP-based routing procedure. The circuit in Fig.~\ref{fig:swap_cx}a, which contains a non-adjacent $\mathrm{CX}$ gate, is first decomposed into a sequence of adjacent $\mathrm{CX}$ gates (Fig.~\ref{fig:swap_cx}b). This decomposition is then mapped to the corresponding brickwork pattern shown in Fig.~\ref{fig:swap_cx}c using the generator set $\mathcal{G}$. Light-blue nodes represent identity (wire) bricks, which implement state teleportation without modifying the quantum state.

\begin{figure}[htbp]
    \centering

    % Top row
    \begin{minipage}[t]{0.48\textwidth}
        \centering
        \includegraphics[width=\linewidth,height=4.4cm,keepaspectratio]{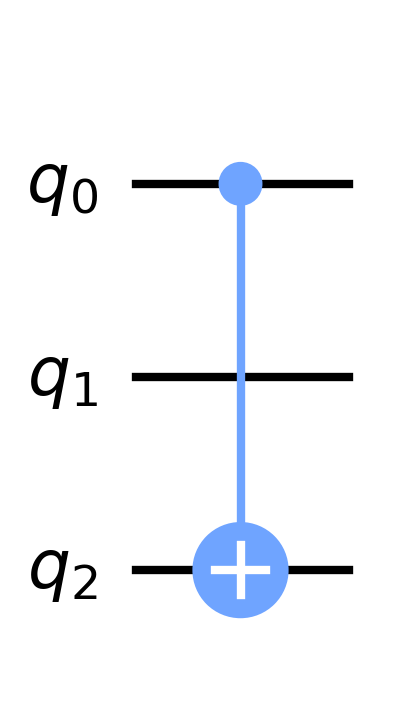}\\[-0.2em]
        \small (a)
    \end{minipage}
    \hfill
    \begin{minipage}[t]{0.48\textwidth}
        \centering
        \includegraphics[width=\linewidth,height=4.4cm,keepaspectratio]{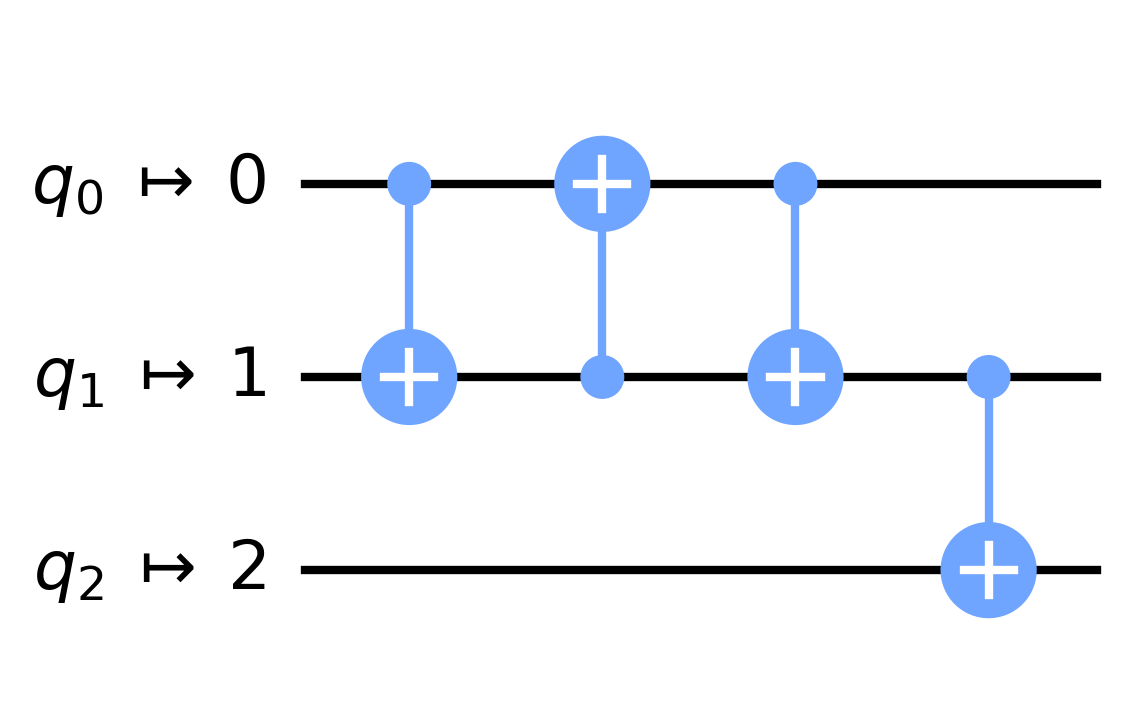}\\[-0.2em]
        \small (b)
    \end{minipage}

    \vskip0.8em

    % Bottom panel
    \begin{minipage}[t]{\textwidth}
        \centering
        \includegraphics[width=\linewidth]{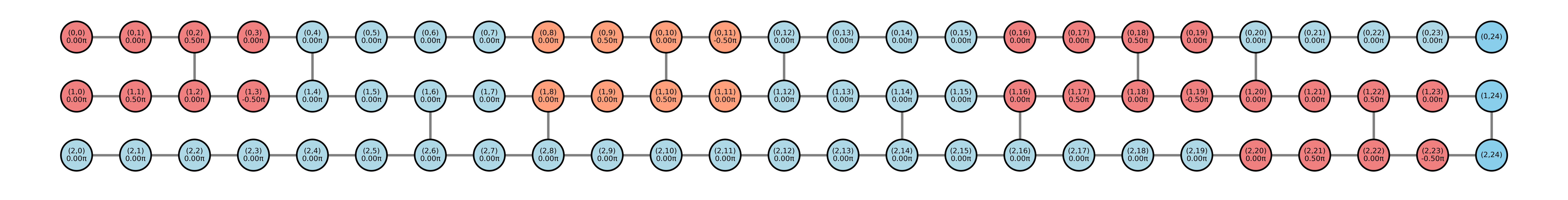}\\[-0.2em]
        \small (c)
    \end{minipage}

    \caption{
    Illustration of the decomposition of a non-adjacent $\mathrm{CX}$ gate and its mapping to a brickwork graph.
    \textbf{(a)} A non-adjacent $\mathrm{CX}$ gate.
    \textbf{(b)} Decomposition into adjacent $\mathrm{CX}$ gates. This decomposition does not permute the identity of the qubits i.e. they are (effectively) swapped back to their original placements, hence the arrows indicating $q_i \rightarrow i$.
    \textbf{(c)} Mapping of the aligned circuit to the corresponding brickwork graph; note that vertical inter-brick connections are alternated, forming a long-short-long-short pattern.
    Light-blue nodes represent wire, or identity, bricks that teleport the state without alteration. Here, the final rightmost column qubits are output qubits (darker blue) without a measurement angle specified; these may either be measured to complete the computation, or fed into another part of the larger algorithm. The placement of inter-row CZs is determined as follows: labeling the topmost row with $r=0$, CZ connections are made alternating based on both row $r$ and column $c$ (also starting from $c=0$), and we here define them as connecting downward i.e. from $r,c\to r,c+1$. For $r$ even, CZs downward are made for $c=2,4\,$ mod $8$. For $r$ odd, CZs downward are made for $c=6,8\,$ mod $8$, exempting the column $c=0$ (where CZs are not performed, in agreement with the definition given in Ref.~\cite{ubqc}; note the shift of index labeling by $1$ (i.e. we count from $0$ instead of $1$).    }
    \label{fig:swap_cx}
\end{figure}

% \subsubsection{Multi-controlled $\mathrm{U}$ Gates} \label{sec:mcx} 
\subsubsection{\texorpdfstring{Multi-controlled $\mathrm{U}$ Gates}{Multi-controlled U Gates}}\label{sec:mcx}% Do you want this explanation?

% Multi-controlled $\mathrm{U}$ ($\mathrm{MCU}$) gates, which generalize the $\mathrm{CX}$ operation to multiple control qubits and arbitrary rotations on the target qubit, play a central role in a variety of quantum machine learning algorithms such as the quantum recommendation system analysed in this thesis \cite{Sawerwain_Wróblewski_2019} and high-level quantum programming constructs. An $\mathrm{MCU}$ gate applies the unitary $\mathrm{U} \in \mathrm{U(2)}$ to a target qubit if and only if all the control qubits are in the $|1\rangle$ state. The efficient implementation of these gates is therefore crucial for decomposing quantum algorithms efficiently.

Vital to our proofs of the efficiency of B-QML later in this work is the ability to generate efficient circuit decomposition in terms of our native gate-set $\mathcal{G}$. Recent work by Rosa et al. \cite{Rosa2025optimizinggate_32n_multicx} provides a significant advancement in the decomposition of multi-controlled unitary ($\mathrm{MCU}$) gates. While previously the decomposition of an $n$-controlled single-target unitary into elementary one- and two-qubit gates required a number of $\mathrm{CX}$ gates scaling quadratically with the number of control qubits $n$ for general $\mathrm{U}(2)$ gates (and linearly for certain subclasses such as $\mathrm{SU}(2)$ rotation gates), Rosa et al. demonstrate that, by introducing a single auxiliary qubit initialized in the $|0\rangle$ state, it is possible to rewrite any multi-controlled $\mathrm{U}(2)$ gate---including multi-controlled not ($\mathrm{MCX}$) gates---as a composition of multi-controlled $\mathrm{SU}(2)$ gates with a phase correction. This approach leverages highly efficient linear decomposition algorithm available for $\mathrm{SU}(2)$ gates given in Ref.~\cite{Rosa2025optimizinggate_32n_multicx}, resulting in an implementation that requires at most $32n$ non-adjacent $\mathrm{CX}$ gates to implement an $\mathrm{MCU}$ gate with a constant number of ancillae. 

The linear gate count of this decomposition does not, by itself,
include the cost of satisfying nearest-neighbour connectivity.
For the asymptotic brickwork bound, we therefore use the
connectivity-aware construction of Zindorf and
Bose~\cite{ZindorfBose2025LNN}. Their Theorem~6 implements an
arbitrary multi-controlled $U(2)$ gate with $c\geq6$ controls
on a linear chain of $W\geq c+2$ qubits using one clean
ancilla. The construction permits arbitrary positions of the
controls, target, and ancilla, and requires at most
$12W+12c-56$ nearest-neighbour CX gates, $16c-32$ T gates,
$8c-8$ H gates, and eleven single-qubit rotations.

Since $c\leq W-2$, the total elementary gate count is
$\mathcal{O}(W)$, with connectivity already accounted for.
The ancilla is restored and can be reused between successive
gates. We reserve one clean ancillary qubit in addition to
the original circuit width $\omega$, giving
$W=\omega+1$. Single-qubit gates and adjacent CX gates admit
constant-cost brickwork implementations. Consequently, this
construction yields an $\mathcal{O}(\omega)$ brickwork-depth
bound for each multi-controlled $U(2)$ gate.

% Crucially, since this approach is linear it does not negatively affect the asymptotic upper-bounded complexity measured as brickwork depth. In the event such a decomposition is not efficient, it can always be approximated efficiently using the Solovay-Kitaev algorithm \cite{dawson2006solovay,kitaev2002classical}, which leads in turn to a universal upper bounded depth for $\varepsilon$-approximated quantum gates (theorem~\ref{thm:bw_sk_upper_bound}). 

%\todoTz{Intuitively, this should mess up blidnness, since no wwe are using computational basis states instead of just rotated $+$ states in our computation! But I think it doesn't. Can you comment on how? I think it's by hiding our ancilla with rotations, but I am curious for your thoughts here. A short comment would be helpful! REX: So in UBQL the graph is initialized into this big $+$ state with some thetas (irrelevant for scaling) The first operation in ubqc is therefore usually initialisation, that is single qubit rotations to get it into the state you want from $+$. These rotations and computations happen onto the existing state which works, I believe, as a consequence of the stabilizer formalism. Im not entirely sure where your confusion comes from since we always measure in the XY basis and apply our compations in that manner. But lets discuss this!}

\subsection{Circuit Decomposition}

As discussed, arbitrary quantum circuits are transpiled into a brickwork graph by decomposing them into sequences of elementary operations drawn from the gate-set $\mathcal{G}$. This restriction to our gate-set $\mathcal{G}$ is essential to ensure compatibility with the brickwork graph structure, which admits only a fixed set of local operations (but does not restrict the arguments in terms of computational complexity due to the Solovay-Kitaev algorithm as mentioned above). In our implementation, this decomposition is carried out using Qiskit~\cite{qiskit}, whose connectivity-aware decomposition routines produce circuit representations restricted to $\mathcal{G}$ while accounting for qubit placement and routing constraints. The resulting circuits are thus both gate-set compliant and more structurally similar to the brickwork graph topology compared to the initial input circuit.

The transpilation procedure is configurable and can be tailored to different optimization objectives through the choice of decomposition strategies, optimization levels, and routing methods. Further implementation details and configuration options are provided in Appendix~\ref{app:bw_transpiler} and in the public repository for MasonQ, our proposed brickwork transpiler on \href{https://github.com/Nonsens24/brickwork_transpiler}{github.com/Nonsens24/brickwork\_transpiler}. An illustrative example of these optimizations is given in Appendix~\ref{app:circuit_optimisations}. 

% As discussed, arbitrary quantum circuits are transpiled into a brickwork graph by decomposing them into sequences of elementary operations drawn from the gate-set $\mathcal{G}$. Our transpiler employs Qiskit~\cite{qiskit} to perform this decomposition, leveraging its optimised decomposition routines, including qubit routing, initial placement, and connectivity-aware transformations. This yields efficient circuit representations that are restricted to $\mathcal{G}$ which gets us closer to the brickwork graph topology.

% The transpilation procedure is configurable and can be tailored to different optimisation objectives through the choice of decomposition strategies, optimisation levels, and routing methods. Implementation details and configuration options are available in our public repository (\href{https://github.com/Nonsens24/brickwork_transpiler}{brickwork\_transpiler}). An illustrative example of these optimisations is provided in Appendix~\ref{app:circuit_optimisations}.

%For this purpose, the Python-based Qiskit framework  was employed, selected due to its active development, comprehensive library of quantum functionalities, and robust community support.

\subsection{Details of Implementation} \label{subsec:high_level_transpiler}

To simulate the resulting execution encoded by measurements of the brickwork graph, we employ \emph{Graphix}, an open-source Python library for the generation, optimization, and simulation of measurement-based quantum computing (MBQC) patterns~\cite{uldemolins2026grpahix034, sunami2022graphix}. Graphix provides a native abstraction of MBQC in the form of a \texttt{Pattern} object, which encodes the measurement graph together with its associated correction dependencies and measurement angles. 

Our transpilation pipeline, shown in Fig.~\ref{fig:pipeline}, leverages this abstraction in terms of \texttt{Pattern} objects directly: it takes as input a \texttt{QuantumCircuit} defined in Qiskit and produces as output a corresponding Graphix \texttt{Pattern}. In this way, the transpiler maps circuit-based descriptions into an MBQC representation that can be further optimized and simulated within the Graphix framework. Our  transpiler is also able to visualize the pattern for debugging and analysis purposes.

\begin{figure}[ht!]
\centering
\includegraphics[width=\linewidth]{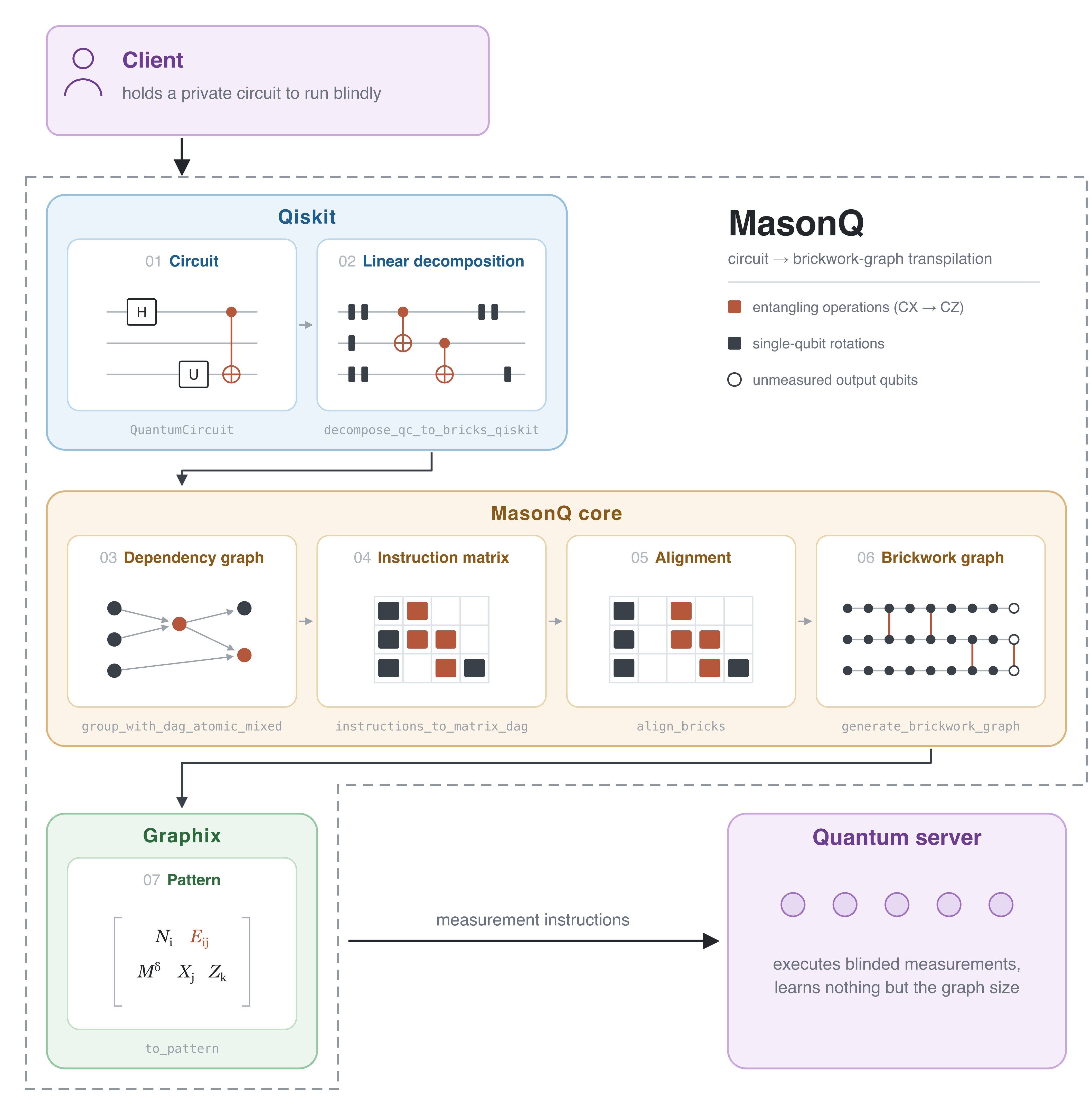}
\caption{The transpilation pipeline used by our software package \texttt{MasonQ}, as envisioned embedded in a larger UBQC protocol. First, a client specifies a quantum circuit they wish to run blindly on a remote quantum server, which they construct as a Qiskit circuit. This circuit is then decomposed into the supported gateset $\mathcal{G}$. From this, MasonQ generates a dependency graph for the adaptive feedforward correction angles for MBQC. This is used to generate a matrix of instructions that records all interactions between qubits, arranging them to align with the brickwork graph topology. MasonQ then aligns the instruction matrix to the actual brickwork graph, as not all gates can be performed in all locations within the graph. This results in a final brickwork graph which is input to Graphix as a pattern, which can be sent as instructions to a quantum server enabling implementation of the desired UBQC protocol.}
\label{fig:pipeline}
\end{figure}

This section provides a high-level overview of how circuits are mapped into measurement-based patterns and why both representations are needed: the former captures the gate-model description used by Qiskit optimizations, while the latter supports efficient pattern simulation, correction dependency computation, and supports optimizations of the measurement calculus \cite{The_measurement_calc}. To the best of our knowledge, MasonQ is the first functional and optimal universal circuit-to-brickwork-graph transpiler, implementing the full brickwork transpilation process from first principles.

The entry-point function accepts a Qiskit \texttt{QuantumCircuit} (\texttt{qc\_in}), an optimization level \texttt{opt} (for Qiskit's decomposition/optimization of the circuit to the generator set), Qiskit’s routing layout choices (passed through to the underlying transpilation stack), and an option to use ancillae for linear decomposition of $\mathrm{MCU}$ gates into adjacent $\mathrm{CX}$ gates\footnote{For a comprehensive list of supported optimization, routing, and layout options, see the Qiskit transpiler documentation available at \href{https://quantum.cloud.ibm.com/docs/en/api/qiskit/transpiler}{quantum.cloud.ibm.com/docs/en/api/qiskit/transpiler}.}. The transpiler returns (i) a Graphix \texttt{Pattern} (a graph with measurement angles and dependencies), (ii) a color map used by our visualizer, and (iii) the transpiled circuit itself which was used to potentially undo permutations of the output qubit wires from the original quantum circuit, which can then be compared with a reference solution. The color map is included to make the visualization interpretable, enabling the user to clearly follow the evolution of the pattern. %Without it, the resulting graphs would be difficult to interpret.

Each step in this process is briefly detailed below.

\begin{enumerate}
    \item \textbf{Decomposing the Circuit:} The decomposition of the input circuit gates to the generator set is the first step of transpilation. Having specified our generator set $\mathcal{G}$, this method creates a coupling map to enforce brickwork graph connectivity constraints; together with our optimization methods this process decomposes the circuit using the Qiskit decomposer. The user has the option to use ancilla qubits which adds a constant number of ancillae making the decomposition of $\mathrm{MCU}$ gates linear in the number of control qubits (see Sec. \ref{sec:mcx} for more details).

    \item \textbf{Building the Instruction Matrix:} To extract and preserve the necessary information from a Qiskit circuit, we construct a dedicated matrix that records all instructions, including both the qubits they act upon and their execution order using a dependency graph (see Appendix~\ref{sec:instr_mat}). This representation not only encodes the circuit into a matrix form but also arranges the gates such that they align with the brickwork graph topology. This matrix is referred to as the instruction matrix of the circuit.
    
   \item \textbf{Constructing the Graph:} Once the instruction matrix has been aligned, there exists a direct correspondence between the positions of the encoded gates and the structure of the brickwork graph. Thus, the brickwork graph has its size determined by the instruction matrix, and is now constructed. We represent it as a NetworkX (NX) graph---a standard Python graph data structure consisting of nodes and edges---required for generating a Graphix \texttt{Pattern}.

    \item \textbf{Creating the Pattern:} The final step is to construct a Graphix \texttt{Pattern}, which serves as the central object for representing, optimizing, and simulating the MBQC computation. This is achieved using the Graphix function \texttt{generate\_from\_graph()}, which requires three inputs: a list of measurement angles for the graph nodes, a list of input qubits, and a list of output qubits.
    
    At this stage, the transpiler assigns measurement angles to the graph according to the instruction matrix, thereby rendering the \texttt{Pattern} operational. No further modifications are required to analyze the graph structure or its complexity scaling. Although UBQC introduces an additional encryption layer through random $\theta$ rotations, which is supported by our transpiler, we do not include this layer here because it is not needed for the computational complexity analysis. Its implementation is therefore deferred to Appendix~\ref{app:encryption}. The output qubits are identified with the qubits in the final column of the brickwork graph. A detailed description of the transpiler, together with proofs of its correctness and optimality in terms of brickwork depth, is given in Appendix~\ref{app:bw_transpiler}, which also includes a discussion of the incorporation and visualization of noise; see Appendix~\ref{app:noise_elab}.

\end{enumerate}

With the instruction generation complete, they may be sent to a remote quantum server to perform the delegated UBQC.

\section{Results} \label{sec:results}

This section presents the main results. We derive an asymptotic upper bound on the scaling of brickwork graphs, use it to characterize the resource requirements of B-QML algorithms, and propose a minimal experimental realization of B-QML for a quantum recommendation system amenable to implementation on a near-term quantum internet \cite{Wehner2018}.

\subsection{Upper Bound on Brickwork Graph Scaling}\label{sec:upp_bound_bw}

Understanding the scaling of brickwork graphs is essential for quantifying the resource requirements of measurement-based quantum computations, i.e., the number of qubits and the depth overhead induced by the brickwork transpilation.
Throughout this work, the depth $D(\mathcal{B})$ of a brickwork graph $\mathcal{B}$ is defined as
the number of brick columns in the aligned instruction matrix. In the explicit
graph representation used by the implementation, each brick column is expanded
into four measured graph columns and one final output column is appended. Thus
the graph-column depth is
\[
D_{\mathrm{graph}}(\mathcal{B})=4D(\mathcal{B})+1.
\]
To this end, we first establish an upper bound on the additional brickwork depth due to alignment of bricks in the graph, i.e., the problem of mapping bricks onto the brickwork structure under the constraint that $\mathrm{CX}$ bricks can only be placed at specific locations. This constraint directly governs the scaling behavior of the graph. The following theorems derive this bound in an abstract setting, with further details -- including its realization within our specific transpiler -- provided in Appendix~\ref{sec:practical_alignment_depth_bound}; the intuition of our argument is that, for each brick, alignment requires at most a shift by a single column in the graph. 
%explains and proves how the transpiler obeys the assumptions of these proofs, thereby adhering to these upper bounds.

\begin{theorem}[Brickwork alignment bound]\label{thm:abstract_bw_bound_2n}
For an instruction matrix of depth $n$, independently aligning
each brick to the brickwork parity constraint (that is, the location of $\mathrm{CX}$ bricks dictated by the brickwork structure) increases the number of columns
by at most a factor of two. Hence the resulting brickwork depth satisfies
\[
D(B) \leq 2n.
\]
% Independently aligning the bricks of an unaligned instruction matrix with \(n\) columns increases the number of columns by at most a factor of two; in particular,
% the depth of the resulting brickwork graph is at most $2n$.
% $D(\mathcal{B}) \leq 2n$.
% Let $M$ be a row-isolated instruction matrix of depth $n$ produced by the
% transpiler scheduler. Then independently aligning the CX bricks to the
% brickwork parity constraint increases the number of columns by at most a
% factor of two. In particular, the depth of the aligned brickwork graph
% satisfies
% \[
% D(\mathcal{B}) \leq 2n.
% \]
% This aligns the theorem with Appendix A.3/Theorem 10 and the code.
\end{theorem}

Theorem \ref{thm:abstract_bw_bound_2n} establishes the abstract alignment bound used in the scaling
analysis. The implementation-specific alignment procedure is described
separately in Appendix C.

\noindent\emph{Intuition.} Because vertical-\(\mathrm{CZ}\) support alternates with period two (i.e. the brickwork structure alternates vertical connections between bricks so that only every-other brick in a row gets a vertical connection to its upper neighbor), an adjacent \(\mathrm{CX}\) brick is either already aligned or can be made aligned by a shift of one column. Thus independent alignment never incurs more than a one-column shift per brick, yielding the stated factor-two bound. A complete proof is given in Appendix~\ref{thm:bw_alignment_bound_nor_prac}.

Next, we establish that implementing state swaps within the graph incurs a cost that scales linearly with the number of adjacent $\mathrm{CX}$ gates. This property is required for deriving the upper bound on the scaling.

\begin{lemma}\label{lem:nonadj_cx_overhead}
Consider a quantum circuit mapped to a one-dimensional nearest-neighbor
architecture, and let $\omega$ denote an upper bound on the number of intermediate
rows of qubits that must be crossed over in order to realize any non-adjacent two-qubit
interaction in the circuit. Then any non-adjacent CX gate can be implemented
using at most
\[
6\omega+1
\]
nearest-neighbor CX gates. In particular, the nearest-neighbor routing overhead
of a non-adjacent CX gate is $O(\omega)$.
% Consider a quantum circuit mapped to a one-dimensional nearest-neighbor architecture, and let \(\omega\) denote an upper bound on the number of intermediate qubits that must be crossed in order to realize any non-adjacent two-qubit interaction in the circuit. Then any non-adjacent \(\mathrm{CX}\) gate can be implemented using at most
% \begin{equation}
%     6\omega+2
% \end{equation}
% adjacent \(\mathrm{CX}\) gates.
% Hence, the overhead for implementing a non-adjacent \(\mathrm{CX}\) gate is linear in the width parameter \(\omega\).
\end{lemma}

% \noindent\emph{Intuition.} In a one-dimensional architecture, a non-adjacent interaction must be made local by swapping one qubit along the line until the two participating qubits become adjacent. If at most \(m\) intermediate qubits lie between them, then at most \(m\) swaps are needed to bring them together and at most \(\omega\) swaps are needed to restore the original ordering. Since each swap costs three \(\mathrm{CX}\) gates and each \(\mathrm{CX}\) gate takes at most two bricks to implement, the routing contributes at most \(6\omega\) \(\mathrm{CX}\) gates, and the remaining local implementation contributes two bricks. Therefore the total cost is bounded by \(6\omega+2\). The complete proof is given in Appendix \ref{app:nonadj_cx_overhead_app}.
\noindent\emph{Intuition.} In a one-dimensional architecture, a non-adjacent interaction is
made local by swapping one qubit along the line until the two participating
qubits become adjacent. If at most $\omega$ intermediate qubits lie between them,
then at most $\omega$ swaps are needed to bring them together and at most $\omega$ swaps
are needed to restore the original ordering. Since each swap decomposes into
three CX gates, the swap routing contributes at most $6\omega$ nearest-neighbor CX
gates. Once the qubits are adjacent, one additional nearest-neighbor CX
implements the desired interaction. Thus the total number of nearest-neighbor
CX gates is at most $6\omega+1$. The complete proof is given in Appendix \ref{app:nonadj_cx_overhead_app}.

% Combining theorem \ref{thm:abstract_bw_bound_2n} and lemma \ref{lem:nonadj_cx_overhead} with the linear cost of implementing $\mathrm{MCU}$ gates (Section~\ref{sec:mcx}, \cite{Rosa2025optimizinggate_32n_multicx}), a more general statement can be made on the cost of blind quantum computation using UBQC of quantum circuits built using these gate classes.\\

% ROSA TEXT:
% By combining Lemma \ref{lem:nonadj_cx_overhead} with the fact that an $\mathrm{MCU}$ gate can be implemented at cost $32n$ (Section~\ref{sec:mcx}), where $n$ denotes the number of $\mathrm{CX}$ gates appearing in its decomposition, we derive the general bound $32(6\omega + 2)$ for $\mathrm{MCU}$ gates. Invoking Theorem \ref{thm:abstract_bw_bound_2n} then gives an upper bound on the depth scaling of brickwork graphs for algorithms built from this gate-set.

\begin{lemma}[Nearest-neighbour MCU brickwork bound]
\label{lem:mcu_lnn_bw}
An arbitrary multi-controlled single-target $U(2)$ gate
acting within an $\omega$-qubit register admits a brickwork
implementation of depth $\mathcal{O}(\omega)$ using one
additional clean ancilla, which is restored after the gate.
\end{lemma}

\begin{proof}
Write $c$ for the number of controls and
$W=\omega+1$ for the implemented width, including the ancilla.
For $c\geq6$, Theorem~6 of~\cite{ZindorfBose2025LNN} provides an implementation
containing $\mathcal{O}(W+c)=\mathcal{O}(W)$ elementary gates
on a linear nearest-neighbour architecture, including the
cost of arbitrary qubit placement.

The remaining cases $c<6$ involve a bounded number of active
qubits. These can be brought into a contiguous region using
$\mathcal{O}(W)$ adjacent SWAP gates, followed by a
constant-size decomposition and restoration of the original
placement. Their elementary gate count is also
$\mathcal{O}(W)$.

Each elementary gate has a constant-cost translation into
the brickwork gate family. In particular, H and T gates
can be expressed using a constant number of native rotations,
up to an irrelevant overall phase. The brickwork alignment
step contributes only a constant factor. The resulting
depth is therefore $\mathcal{O}(W)=\mathcal{O}(\omega)$.
\end{proof}

\begin{theorem}[Brickwork Depth Upper Bound]
\label{thm:bw_general_alignment}
Let \(\mathcal{C}\) be a quantum circuit containing \(n\) original gates
acting on \(\omega\geq1\) logical qubits, with each
multi-controlled single-target \(U(2)\) gate counted once.
Reserve one additional clean ancilla and write
\(W=\omega+1\) for the implemented width. Let
\[
n_{\mathrm{rot}},\quad
n_{\mathrm{adjCX}},\quad
n_{\mathrm{CX}},\quad
n_{\mathrm{MCU}},
\]
denote the numbers of occurrences of native single-qubit
rotations, adjacent \(\mathrm{CX}\) gates, non-adjacent
\(\mathrm{CX}\) gates, and multi-controlled single-target
\(U(2)\) gates, respectively, in \(\mathcal{C}\). These categories
are taken to be disjoint, so that
\begin{equation}
n_{\mathrm{rot}}
+n_{\mathrm{adjCX}}
+n_{\mathrm{CX}}
+n_{\mathrm{MCU}}
=n.
\end{equation}

Using the nearest-neighbour $\mathrm{MCU}$ construction of 
Lemma~\ref{lem:mcu_lnn_bw}, there is a constant
\(a_{\mathrm{MCU}}\), independent of \(n\), \(\omega\),
and the number of controls, such that the following
upper bounds hold for the brick-column cost before
global alignment:
\begin{equation}
\begin{aligned}
\text{rotation}
    &:\; 1,
&
\text{adj.\ }\mathrm{CX}
    &:\; 2,
\\
\text{non-adj.\ } \mathrm{CX}
    &:\; 2(6W+1),
&
\mathrm{MCU}
    &:\; a_{\mathrm{MCU}}W.
\end{aligned}
\end{equation}
Consequently, \(\mathcal{C}\) admits a brickwork implementation
\(\mathcal{B}\) of width \(W\) whose depth satisfies
\begin{equation}
\label{eq:bw_bound}
D(\mathcal{B})=\mathcal{O}(\omega n).
\end{equation}
\end{theorem}

\begin{proof}
We first specify the MCU contribution. Let \(b_{\mathrm{MCU}}(W,c)\) denote the number of brick
columns used to implement one MCU with \(c\) controls on
\(W\) rows, before global alignment. For an MCU with
\(c\geq6\) controls, Theorem~6 of
Ref.~\cite{ZindorfBose2025LNN} gives a nearest-neighbour
implementation with one clean ancilla using at most
\(12W+12c-56\) CX gates, \(16c-32\) T gates,
\(8c-8\) H gates, and eleven native single-qubit
rotations. These counts already include the cost of
qubit placement along the line.

We conservatively bound an adjacent CX by two brick columns. Up to
global phase, a T gate requires one native rotation
and an H gate can be implemented using three native
rotations. The brick-column cost before alignment
therefore satisfies
\begin{equation}
\begin{aligned}
b_{\mathrm{MCU}}(W,c)
&\leq
2(12W+12c-56)
+(16c-32)
+3(8c-8)
+11
\\
&=24W+64c-157
\\
&\leq88W-285
\\
&\leq88W,
\end{aligned}
\end{equation}
where \(c\leq W-2\).

For \(c<6\), an MCU acts on at most six logical qubits:
five controls and one target. Such a gate admits an
elementary decomposition whose gate count is bounded
independently of the total width \(W\). Routing each
two-qubit interaction and restoring the original placement
costs at most \(\mathcal{O}(W)\). Since the number of these
interactions is bounded, the total pre-alignment cost is
at most \(a_0W\), for some constant \(a_0\) independent
of \(W\) and the choice of \(U\).

Combining this with the bound for \(c\geq6\), we obtain
\begin{equation}
b_{\mathrm{MCU}}(W,c)
\leq
\begin{cases}
a_0W, & c<6,\\
88W,  & c\geq6.
\end{cases}
\end{equation}
Defining
\begin{equation}
a_{\mathrm{MCU}}=\max\{88,a_0\}
\end{equation}
therefore gives the uniform bound
\begin{equation}
b_{\mathrm{MCU}}(W,c)\leq a_{\mathrm{MCU}}W
\end{equation}
for all allowed control counts. The clean ancilla is
restored whenever used and can therefore be reused
throughout the circuit.

The depth of the brickwork graph \(\mathcal{B}\) is
upper bounded by summing the contributions from the
individual gate classes after transpilation into
brickwork form. Using the bounds above, together with
the alignment factor \(2\) from
Theorem~\ref{thm:abstract_bw_bound_2n}, we obtain
\begin{equation}
\begin{aligned}
D(\mathcal{B})\leq
2\Bigl(
&n_{\mathrm{rot}}
+2n_{\mathrm{adjCX}}
\\
&+2(6W+1)n_{\mathrm{CX}}
+a_{\mathrm{MCU}}Wn_{\mathrm{MCU}}
\Bigr).
\end{aligned}
\end{equation}

Since the count of each gate type is bounded by the
total number of original gates, we have
\begin{equation}
n_{\mathrm{rot}},\quad
n_{\mathrm{adjCX}},\quad
n_{\mathrm{CX}},\quad
n_{\mathrm{MCU}}
\leq n.
\end{equation}
Substituting these bounds yields
\begin{equation}
\begin{aligned}
D(\mathcal{B})
&\leq
2\Bigl(
n+2n+2(6W+1)n+a_{\mathrm{MCU}}Wn
\Bigr)
\\
&=
\bigl[(2a_{\mathrm{MCU}}+24)W+10\bigr]n.
\end{aligned}
\end{equation}
Therefore, with \(W=\omega+1\),
\begin{equation}
\boxed{
D(\mathcal{B})
\leq
\bigl[(2a_{\mathrm{MCU}}+24)(\omega+1)+10\bigr]n
=
\mathcal{O}(\omega n)
}.
\end{equation}
Hence, the depth of the constructed brickwork graph
scales at most linearly in both the original width
\(\omega\) and the original gate count \(n\).
\end{proof}

Thus, the main result of this section establishes that circuits
consisting of these gate classes admit a brickwork implementation
whose depth scales at most proportionally to the original gate
count multiplied by the original circuit width, using one
additional clean ancilla. The same conclusion applies to circuits
that can be decomposed into these gate classes with total gate
count $\mathcal{O}(n)$ and width $\mathcal{O}(\omega)$.

Now follows a corollary useful for exact gate bounds in experimental settings.

%\todoTz{REX: I added this corollary with Tzulas comments on experimentalists interests in numerical bounds in mind. Do let me know what you think!}

%OLD MCU
% \begin{corollary}[Conservative Numerical Brickwork Depth Bound]
% \label{cor:explicit_bw_gate_count}
% Let \(C\) be a quantum circuit acting on \(\omega\) logical qubits and containing exactly \(n\) gates, each belonging to one of the classes considered in Theorem~\ref{thm:bw_general_alignment}. Then the depth \(D(\mathcal{B})\) of the corresponding brickwork graph \(\mathcal{B}\) satisfies
% \begin{equation}
% D(\mathcal{B}) \le (384\omega+64)\, n.
% \end{equation}
% \end{corollary}

% \begin{proof}
% Apply the proof of Theorem~\ref{thm:bw_general_alignment} and note that, among the allowed gate classes, the largest transpilation cost is induced by a multi-controlled \(\mathrm{U}\) gate, namely
% \[
% 2\cdot 32(6\omega+1)=384\omega+64,
% \]
% where the factor \(2\) is the alignment overhead from Theorem~\ref{thm:abstract_bw_bound_2n}. Since the circuit contains exactly \(n\) gates, the total brickwork depth is bounded by \(n\) times this worst-case per-gate cost:
% \[
% D(\mathcal{B}) \le (384\omega + 64)\,n.
% \]
% \end{proof}

\begin{corollary}[Conservative Brickwork Depth Bound]
\label{cor:explicit_bw_gate_count}
Let \(\mathcal{C}\) be a quantum circuit acting on \(\omega\geq1\)
logical qubits and containing exactly \(n\) gates, each
belonging to one of the classes considered in
Theorem~\ref{thm:bw_general_alignment}.
With one additional clean ancilla, \(\mathcal{C}\) admits a brickwork
implementation \(\mathcal{B}\) of width \(W=\omega+1\)
whose depth satisfies
\begin{equation}
D(\mathcal{B})
\leq 2a_{\mathrm{MCU}}(\omega+1)n,
\end{equation}
where \(a_{\mathrm{MCU}}=\max\{88,a_0\}\) is the
constant defined in the proof of
Theorem~\ref{thm:bw_general_alignment}.
\end{corollary}

\begin{proof}
Apply the proof of
Theorem~\ref{thm:bw_general_alignment}.
Since \(a_{\mathrm{MCU}}\geq88\) and \(W\geq2\),
the pre-alignment brick-column cost of every allowed
gate is bounded by \(a_{\mathrm{MCU}}W\):
\begin{equation}
\max\bigl\{
1,\;2,\;2(6W+1),\;a_{\mathrm{MCU}}W
\bigr\}
=a_{\mathrm{MCU}}W.
\end{equation}
The circuit contains exactly \(n\) original gates.
Summing their costs and applying the alignment factor
\(2\) from Theorem~\ref{thm:abstract_bw_bound_2n}
therefore yields
\begin{equation}
D(\mathcal{B})
\leq 2a_{\mathrm{MCU}}Wn
=2a_{\mathrm{MCU}}(\omega+1)n.
\end{equation}
\end{proof}

We now extend Theorem~\ref{thm:bw_general_alignment} by deriving a universal upper bound on the brickwork depth for circuits containing gates outside the linearly decomposable gate classes---also referred to as non-native gates---using Solovay--Kitaev approximations to translate such gates into the native gate-set~\cite{dawson2006solovay,kitaev2002classical}.

To make this precise, we now consider circuits containing arbitrary constant-locality (that is, involving a set of at most $q$ qubit wires)
unitary gates which do not belong to the gate classes considered in Theorem \ref{thm:bw_general_alignment} and are assumed non-native\footnote{Non-native gates are the unknown multi qubit unitaries $U$ that could break the $\mathcal{O}(\omega n)$ bound by not being efficiently decomposable into $\mathcal{G}$, were it not for the Solovay--Kitaev theorem.} in terms of $\mathcal{G}$. Such gates are first compiled
to a fixed finite universal gate-set using the Solovay--Kitaev algorithm, and the resulting
finite gate sequence is then translated exactly into the brickwork-native gate family
\[
\mathcal G
=
\{R_x(\theta),R_z(\theta),\mathrm{CX}\}.
\]
This gives a universal brickwork depth bound with only polylogarithmic dependence on the inverse approximation accuracy.

We also note that our generator set
\[
\mathcal{G}
=
\{R_x(\theta),R_z(\theta),\mathrm{CX}\}
\]
can implement, with constant overhead, the finite approximately universal gate-set
\[
\mathcal G_{\mathrm{SK}}^{(q)}
=
\{H,T,T^\dagger,\mathrm{CX}\}.
\]
Indeed,
\[
T = e^{i\pi/8}R_z(\pi/4),
\qquad
T^\dagger = e^{-i\pi/8}R_z(-\pi/4),
\]
and
\[
H = i\,R_z(\pi/2)R_x(\pi/2)R_z(\pi/2),
\]
while \(\mathrm{CX}\) is already contained in \(\mathcal G\).
Hence every gate in \(\mathcal G_{\mathrm{SK}}^{(q)}\) has a constant-size implementation
over \(\mathcal G\). Therefore, for the purpose of applying the
Solovay--Kitaev theorem, arbitrary \(q\)-local non-native unitary gates may first be
approximated over \(\mathcal G_{\mathrm{SK}}^{(q)}\), and the resulting gates may then be
translated into \(\mathcal G\) with constant overhead.

\begin{theorem}[Universal Brickwork Depth Bound under Solovay--Kitaev Approximation]
\label{thm:bw_sk_upper_bound}
Let \(C\) be a quantum circuit of gate complexity \(n\) acting on \(\omega\) logical qubits.
Suppose that \(C\) is expressed over an arbitrary finite approximately universal gate-set,
and that every non-native gate has locality at most \(q=O(1)\)\footnote{Here, the locality of a gate means the number of logical qubits on which the gate acts nontrivially. Thus a
one-qubit rotation has locality \(1\), a two-qubit \(\mathrm{CX}\) gate has locality \(2\),
and a \(q\)-local unitary is a unitary acting nontrivially on at most \(q\) adjacent logical qubits.
The condition \(q=O(1)\) means that this number is bounded independently of the circuit
width \(\omega\).}.
Let \(\varepsilon>0\) be the desired total approximation error, measured in operator norm.\footnote{The operator norm is defined by
\(\|A\|:=\sup_{\|\psi\|=1}\|A\psi\|\). Thus \(\|U-V\|\leq\varepsilon\) gives a worst-case state-vector error bound over all normalized input states.}.

Assume that every non-native gate is first approximated over the finite gate-set
\(\mathcal G_{\mathrm{SK}}^{(q)}\) using the Solovay--Kitaev algorithm, and that the
resulting gates are then translated with constant overhead into our generator set \(\mathcal G\). Then the resulting brickwork graph
\(\mathcal B_\varepsilon\) has depth
\begin{equation}
D(\mathcal B_\varepsilon)
=
\mathcal O\!\left(
\omega n \log^{\kappa}\!\frac{n}{\varepsilon}
\right),
\end{equation}
where \(\kappa>0\) is the Solovay--Kitaev exponent depending on the particular construction used e.g. for the standard
Solovay--Kitaev construction \(\kappa \approx 3.97\).
\end{theorem}

\begin{proof}
Let the target total approximation error be \(\varepsilon\). Since the input circuit
contains \(n\) gates, it suffices to approximate each non-native gate to accuracy
\begin{equation}
\varepsilon_{\mathrm{loc}}
=
\frac{\varepsilon}{n}.
\end{equation}
Then the accumulated error is bounded by
\begin{equation}
n\varepsilon_{\mathrm{loc}}
=
\varepsilon.
\end{equation}

By the Solovay--Kitaev theorem, each non-native gate can be approximated over the finite
approximately universal gate-set \(\mathcal G_{\mathrm{SK}}^{(q)}\) using at most
\begin{equation}
\mathcal O\!\left(
\log^{\kappa}\!\frac{1}{\varepsilon_{\mathrm{loc}}}
\right)
=
\mathcal O\!\left(
\log^{\kappa}\!\frac{n}{\varepsilon}
\right)
\end{equation}
gates, where \(\kappa>0\) is the Solovay--Kitaev exponent. Since \(q=O(1)\), the constants
hidden in this bound are independent of \(n\), \(\omega\), and \(\varepsilon\).

Let \(C_\varepsilon\) be the circuit obtained after replacing every non-native gate by its
Solovay--Kitaev approximation. The cumulative error satisfies
\begin{equation}
\|C-C_\varepsilon\|
\le
\sum_{j=1}^{n}\varepsilon_{\mathrm{loc}}
=
n\varepsilon_{\mathrm{loc}}
=
\varepsilon.
\end{equation}

Therefore, after this preliminary gate-set translation, the gate complexity of the
approximating circuit \(C_\varepsilon\) is asymptotically bounded above by
\begin{equation}
n_\varepsilon
=
\mathcal O\!\left(
n\log^{\kappa}\!\frac{n}{\varepsilon}
\right).
\end{equation}

Each gate in \(\mathcal G_{\mathrm{SK}}^{(q)}\) is then translated with constant overhead
into our brickwork-native gate family
\[
\mathcal {G}
=
\{R_x(\theta),R_z(\theta),\mathrm{CX}\}.
\]
Hence the same asymptotic bound holds for the circuit \(C_\varepsilon\) expressed in terms of $\mathcal{G}$.

We now apply Theorem~\ref{thm:bw_general_alignment} to the translated circuit
\(C_\varepsilon\). Since \(C_\varepsilon\) is expressed over the gate classes handled by
the brickwork transpiler, namely single-qubit rotations, adjacent \(\mathrm{CX}\) gates,
non-adjacent \(\mathrm{CX}\) gates, and multi-controlled \(\mathrm U\) gates, its brickwork
depth satisfies
\begin{equation}
D(\mathcal B_\varepsilon)
=
\mathcal O(\omega n_\varepsilon).
\end{equation}
Substituting the bound on \(n_\varepsilon\), we obtain
\begin{equation}
D(\mathcal B_\varepsilon)
=
\mathcal O\!\left(
\omega n\log^{\kappa}\!\frac{n}{\varepsilon}
\right).
\end{equation}
Hence the brickwork depth grows linearly in the circuit width \(\omega\), linearly in the
original gate complexity \(n\), and only polylogarithmically in the inverse approximation
accuracy.
\end{proof}

Since the finite set \(\mathcal G_{\mathrm{SK}}^{(q)}\) is approximately universal and is implementable with constant overhead over
\[
\mathcal G=\{R_x(\theta),R_z(\theta),\mathrm{CX}\},
\]
any constant-locality unitary gate can be approximated over the brickwork-native family with Solovay--Kitaev overhead
\[
\log^\kappa\frac{n}{\varepsilon}
\]
after allocating local error \(\varepsilon/n\). Consequently, the brickwork depth remains
\[
\boxed{
D(\mathcal B_\varepsilon)
=
\mathcal O\!\left(
\omega n\log^\kappa\frac{n}{\varepsilon}
\right).
}
\]
    
% \section{General Brickwork Transpilation Cost}
%\subsection{Brickwork depth scaling of quantum learning circuits}
\subsection{Brickwork Depth Scalings of Blind Quantum Learning Protocols}

\noindent

Having established an upper bound on the cost of implementing blind quantum algorithms for a range of gate classes, we now specialize to a selected set of quantum machine learning (QML) circuits. We begin by showing how our transpilation procedure yields an upper bound on the brickwork depth $D$ for a well-known and analytically tractable example, the quantum Fourier transform (QFT), before turning to more general classes of QML algorithms and calculating their blind complexity i.e. the depth of the brickwork graph state needed to implement them blindly. This adds overhead compared to the unblinded circuit complexity for our selected QML circuits, but, as we will show, sometimes smaller one than may be expected based on the classical intuition that privacy (blindness) makes learning difficult.

For the QFT on $\log_2 N$ qubits, corresponding to an $N$-dimensional input state, the gate complexity (i.e. the number of gates) is $\mathcal{O}(\log^2 N)$. Since the circuit width is $\omega=\log_2 N$, Theorem~\ref{thm:bw_general_alignment} implies
\begin{equation}
\label{DepthQFT}
D(\mathcal{B}_{QFT})  = \mathcal{O}\!\big(\omega\log^2 N\big)
   = \mathcal{O}\!\big(\log^3 N\big).
\end{equation}

Thus, for the exact QFT, the brickwork graph depth grows polylogarithmically with the input size $N$. We verify this scaling by transpiling QFT circuits of increasing size and recording the depth of the resulting brickwork graphs. The observed behavior is shown in Fig.~\ref{fig:res_qft}a.

\begin{figure*}[!ht]
    \centering

    \begin{minipage}[t]{0.48\textwidth}
        \centering
        \includegraphics[width=\linewidth]{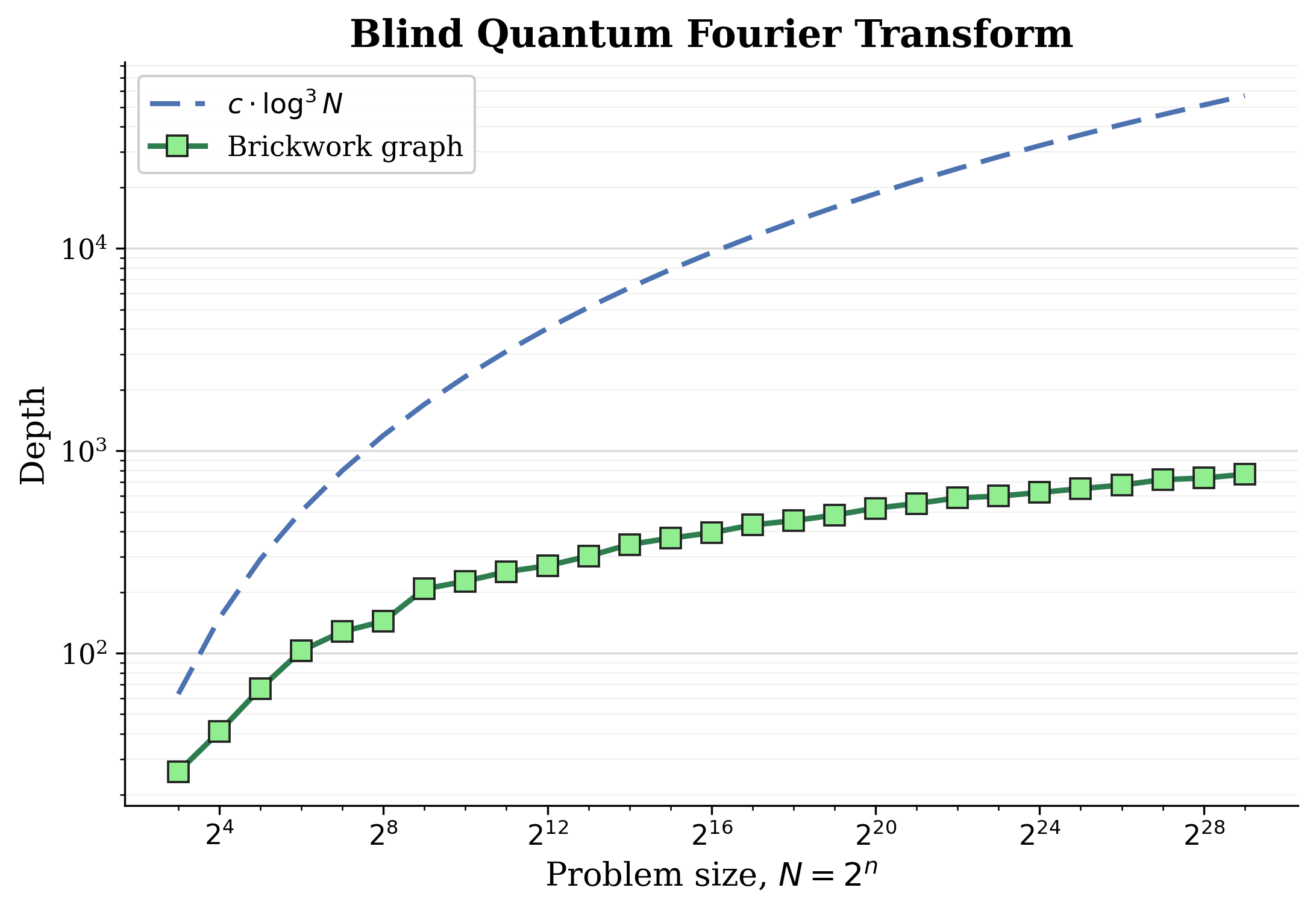}
        \small (a)
    \end{minipage}
    \hfill
    \begin{minipage}[t]{0.48\textwidth}
        \centering
        \includegraphics[width=\linewidth]{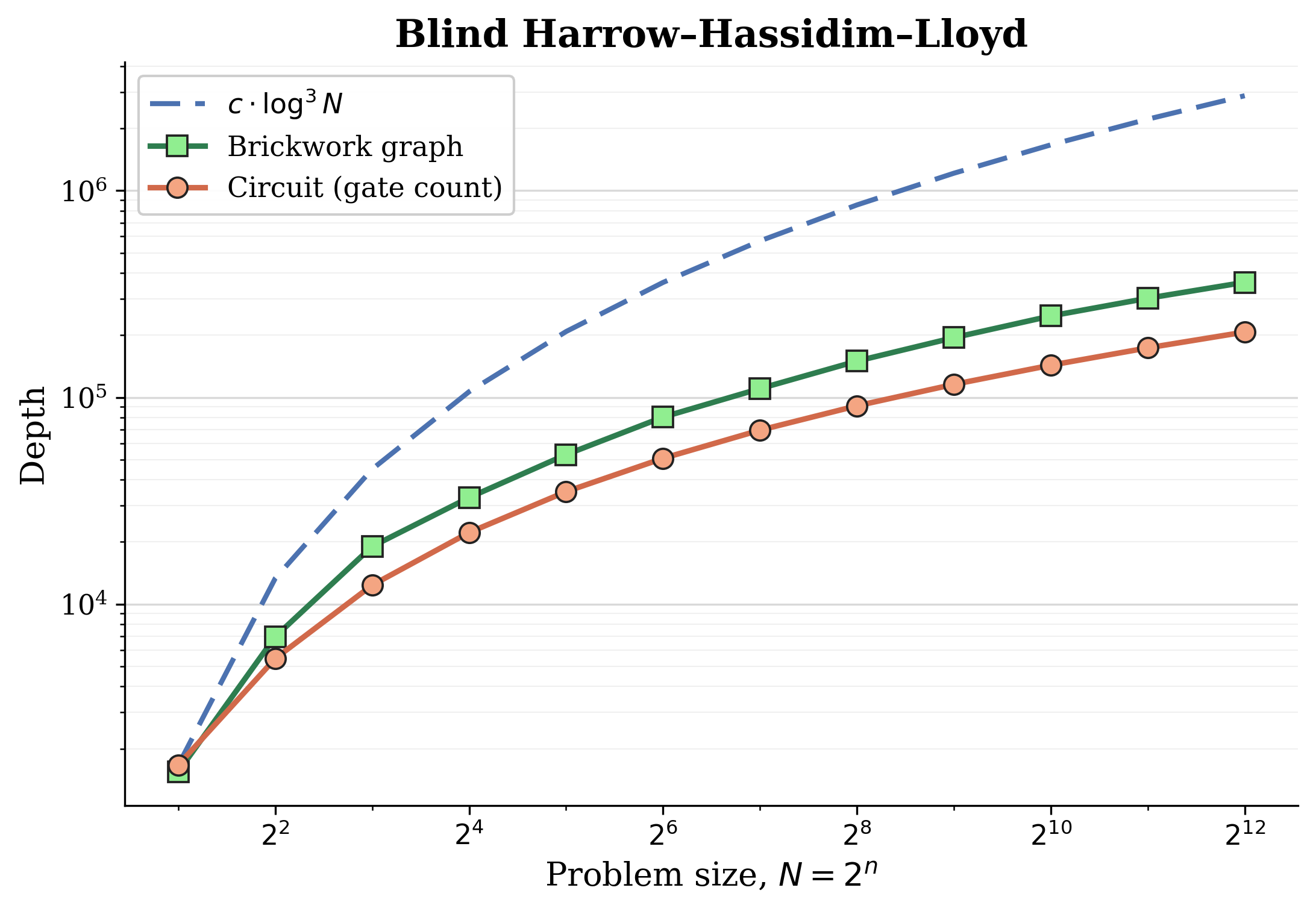}
        \small (b)
    \end{minipage}

\caption{Scaling of brickwork depth $D(\mathcal{B})$ for transpiled blinded QFT and HHL circuits as functions of the problem size $N$. \textbf{(a)} QFT circuits, with growth of the depth consistent with a scaling upper bounded by $\mathcal{O}(\log^3 N)$, with $N$ the dimension of the input state. \textbf{(b)} HHL circuits for the tridiagonal Toeplitz matrix with a phase register of fixed size, with growth of the depth consistent with a scaling upper bounded by $\mathcal{O}(\log^3 N)$. Further details on Blind HHL, including inputs, sparsity, condition number, and error, are given in Appendix~\ref{sec:res_hhl}.}\label{fig:res_qft}
%     \caption{
% Scaling of brickwork depth $D(\mathcal{B})$ for transpiled blinded QFT and HHL circuits as functions of the problem size $N$. \textbf{(a)} QFT circuits, with growth of the depth consistent with a scaling upper bounded by $\mathcal{O}(\log^3 N)$, with $N$ the dimension of the input state. \textbf{(b)} HHL circuits, showing qualitatively polylogarithmic scaling with the problem size $N$, which grows exponentially in the number of qubits $n$. Further details on Blind HHL, including inputs, sparsity, condition number, and error, are given in Appendix~\ref{sec:res_hhl}.
%     }
\end{figure*}

%\todoTz{Order so top legend is brickwork depth upperbound, middle is calculated brickwork depth, bottom is original gate complexity}

%\todoTz{One more plot! Maybe as double to plot above. Use subfigure so it doesn't need full span. Plot Depth of B-HHL vs HHL using fixed gate-set native for MBQC. Use the nice numerical transpiler you made! REX: Great point, I will make plots later after the text has improved! }

% \subsubsection{An Algorithm for Bounding General Transpilation Cost}
This example thus motivates a general strategy, which we introduce here, for deriving upper bounds on the resource overhead required to implement quantum algorithms in blind brickwork form. To our knowledge, no such general procedure has previously been available. Given a quantum circuit, the estimation proceeds in four steps:
\begin{enumerate}[label=(\roman*)]
    \item Verify that the gates used by the circuit admit a linear decomposition into the gate classes used in the proof of Theorem~\ref{thm:bw_general_alignment}. If they do not, then the present procedure does not directly apply, and a different decomposition or an extension of the gate-set is required.
    \item Identify the circuit width $\omega$ (number of qubits) and gate complexity $n$.
    \item Apply Theorem~\ref{thm:bw_general_alignment} to obtain the bound $\mathcal{O}(\omega n)$ on the resource overhead.
    \item Rewrite this bound in terms of the natural input parameter of the algorithm, and simplify the resulting expression.
\end{enumerate}

Applying this procedure yields the detailed analyses of several quantum machine learning algorithms presented in Appendix~\ref{app:qml_analysis}. The resulting bounds form the basis of the catalogue of blind quantum machine learning algorithms summarized in Table~\ref{tab:complexity_comparison}.

Taken together, these examples show that blindness does not generally change the qualitative scaling of efficient quantum machine learning algorithms, but instead introduces a structured overhead governed by circuit width. For algorithms that are already polylogarithmic in the problem size, such as the QFT and HHL with bounded condition number and sparsity, the blind implementation remains polylogarithmic, see Fig. \ref{fig:res_qft}b. For database-driven quantum recommendation systems (QRS), blindness contributes an
additional logarithmic factor, yielding $\mathcal{O}(L^2\log L)$ when data loading is
included (Fig.~\ref{fig:res_qrs}b) and $\mathcal{O}(\sqrt{L}\log L)$ when inference is
performed on a pre-encoded quantum database (Fig.~\ref{fig:res_qrs}a). Here, $L$ denotes the number of entries in the database. Notably, this inference complexity is sublinear in the database size \(L\) and scales more favorably than that of classical recommendation systems under comparable blindness constraints (Appendix~\ref{sec:qrs_no_db}).

Quantum transformer architectures provide another class of quantum learning protocols amenable to blinding. In the quantum vision transformer of~\cite{Cherrat_2024}, an image is divided into \(n\) patches, each represented by a \(d\)-dimensional feature vector, and an attention mechanism learns correlations between these patches. The resulting \(n\times d\) patch--feature matrix is encoded using an \(n\)-qubit patch register and a \(d\)-qubit feature register, giving a total circuit width of \(n+d\) qubits. Applying the same transpilation principle as for the B-QRS therefore yields a width-dependent upper bound: the blind complexity, measured by the depth of the corresponding brickwork graph,
is at most the native circuit depth multiplied by $(n+d)^{2}$, since a circuit of
depth $T$ on $\omega=n+d$ qubits contains at most $\omega T$ gates and
Theorem~\ref{thm:bw_general_alignment} charges one factor of $\omega$ per gate. Overall, Table~\ref{tab:complexity_comparison} shows that in all our studied cases the cost of blindness is systematic and analyzable. Furthermore, for HHL, and the inference part of QRS, the addition of blindness does not disrupt the favorable asymptotic scaling of the underlying quantum algorithm, especially in comparison to their classically blinded counterparts as discussed in Appendix~\ref{app:comparison_scope} and Appendix~\ref{app:qml_analysis}. Of course, we only analyze here a select number of QML circuits as examples; our framework and the MasonQ software package enable further, open-ended study into B-QML protocols, incorporating new QML circuits and strategies as they are developed. 

\begin{figure*}[!ht]
    \centering

    \begin{minipage}[t]{0.48\textwidth}
        \centering
        \includegraphics[width=\linewidth]{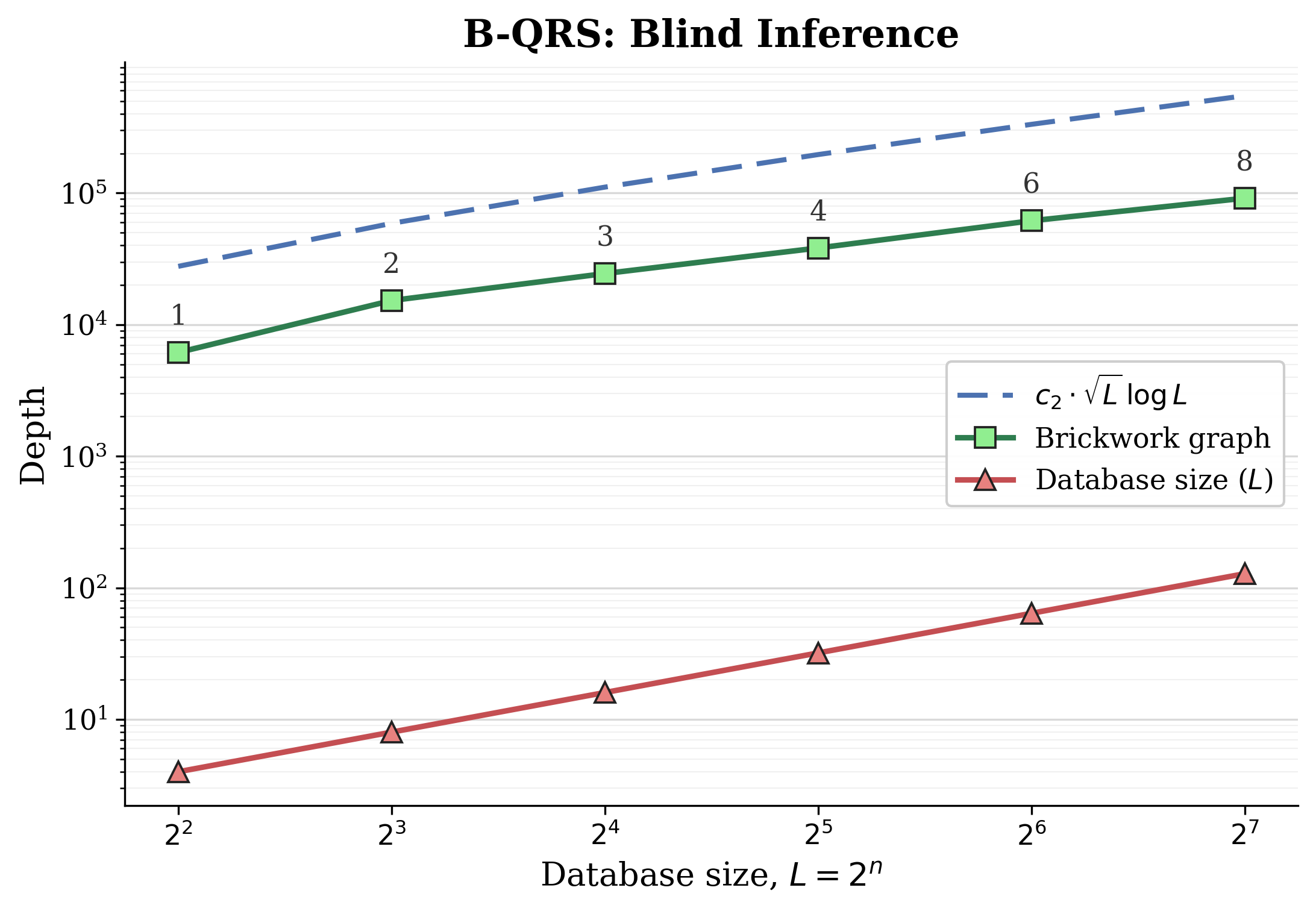}
        \small (a)
    \end{minipage}
    \hfill
    \begin{minipage}[t]{0.48\textwidth}
        \centering
        \includegraphics[width=\linewidth]{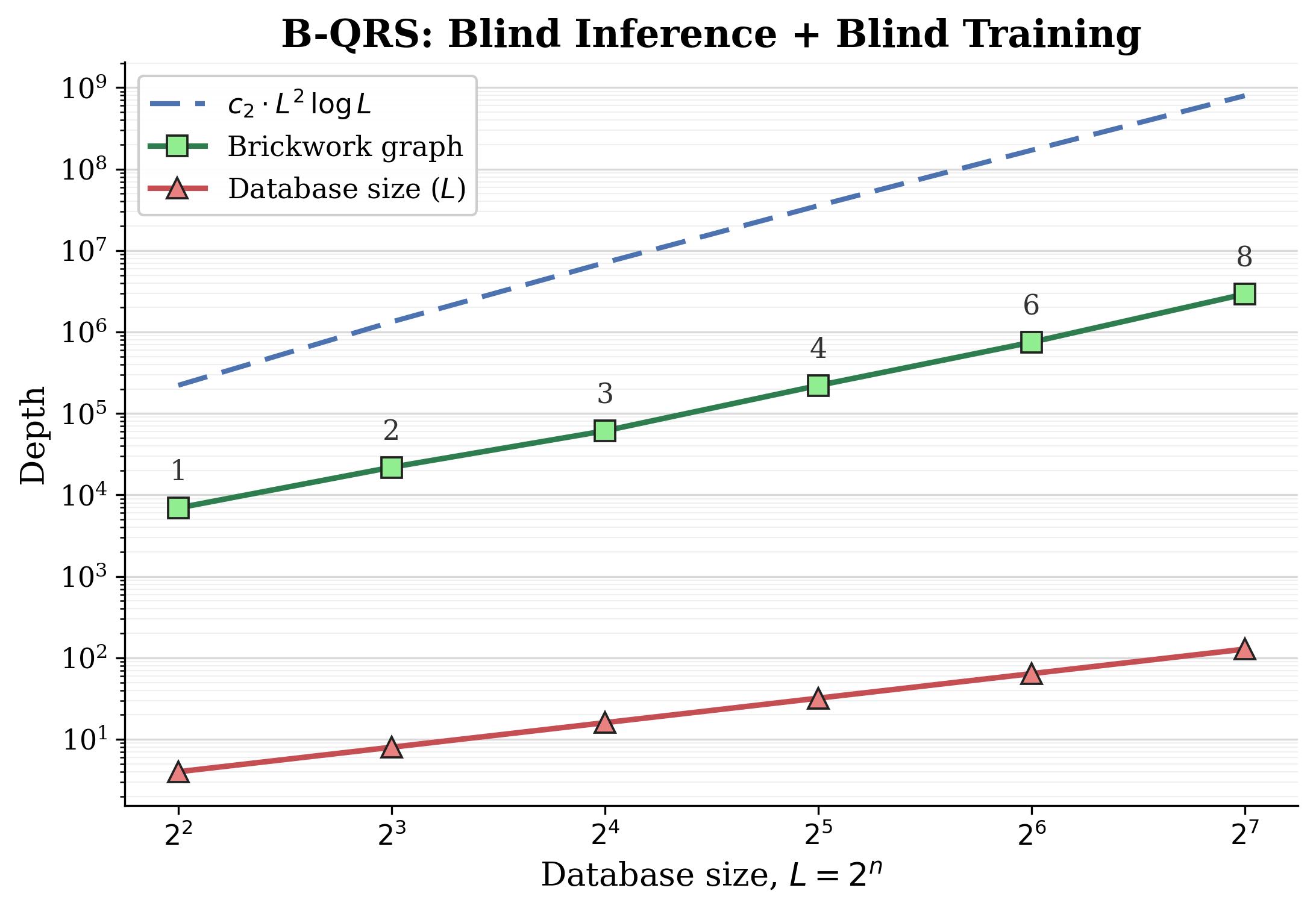}
        \small (b)
    \end{minipage}

\caption{Scaling of brickwork depth $D(\mathcal{B})$ for a transpiled blind quantum recommendation system (B-QRS) is plotted in green, with the number of Grover iterations/amplifications as a number above the line, including (right) and excluding (left) the cost of blindly training the model. Here the output of the algorithm is a single match from the quantum recommendation system, so that running the algorithm $k$ times yields $k$ recommendations; we relax this assumption in Fig. \ref{fig:res_qrs_full_plots}. Here, the horizontal axis is the size of the database, $L=2^n$, with $n$ the number of qubits used to represent the database. We plot in red the database size $L$ for reference. We also include our analytic upper bound (dashed line). 
\textbf{(a)} Blind inference from a pre-trained model, with growth consistent with $\mathcal{O}(\sqrt{L}\,\log L)$.
\textbf{(b)} Blind inference including the cost of blindly training the model with growth consistent with $\mathcal{O}(L^2 \log L)$. Note the large difference in vertical axes; there is a cost to performing the training blindly as well.
Detailed analyses are given in Appendices~\ref{sec:qrs_no_db} and~\ref{app:bqrs_with_db}.
}
    \label{fig:res_qrs}
\end{figure*}

\begin{table*}[t]
\centering
\caption{B-QML algorithm catalogue, showing upper bounds
on unblinded computational resources and brickwork
depth under the assumptions stated in the text.
For QFT, \(N\) is the transform dimension; for HHL,
\(N\) is the matrix dimension, \(\kappa\) its condition
number, \(s\) the maximum number of nonzero entries
per row, and \(\varepsilon\) the target approximation
error. For QRS, \(L\) is the number of database entries,
with fixed feature width; inference-only bounds exclude
database encoding and assume the stated coherent-access
model. For the Quantum Compound Transformer, \(n\) is
the number of input patches and \(d\) the feature
dimension per patch.}
% \caption{The B-QML algorithms catalogue, with upper bounds to both unblinded computational complexity and the blinded complexity derived in this paper.}
\label{tab:complexity_comparison}
\small
\setlength{\tabcolsep}{5pt}
\begin{tabularx}{\textwidth}{
  >{\raggedright\arraybackslash}p{0.28\textwidth}
  >{\raggedright\arraybackslash}X
  >{\raggedright\arraybackslash}X
}
\toprule
\toprule
Algorithm & Upper-bound computational complexity & Upper-bound blind complexity \\
\midrule
\midrule
Quantum Fourier Transform
& $\mathcal{O}(\log^2 N)$
& $\mathcal{O}(\log^3 N)$ \\
\midrule

HHL \cite{Harrow2009}
& $\mathrm{poly}\!\bigl(\kappa,\, s,\, 1/\varepsilon\bigr)\cdot \mathrm{polylog}(N)$
& $\mathrm{poly}\!\bigl(\kappa,\, s,\, 1/\varepsilon\bigr)\cdot \mathrm{polylog}(N)$ \\
\midrule
QRS (QK-NN + Grover) $+$ database encoding \cite{Sawerwain_Wróblewski_2019}
& $\mathcal{O}(L^2)$
& $\mathcal{O}(L^2 \log L)$ \\
\addlinespace
QRS (QK-NN + Grover) inference only
& $\mathcal{O}(\sqrt{L})$
& $\mathcal{O}(\sqrt{L}\log L)$ \\
\midrule
Quantum Compound Transformer~\cite{Cherrat_2024} &
$\mathcal{O}\big(\log n + n\log d + \log(n+d)\big)$ (depth) &
$\mathcal{O}\big((n+d)^{2}\,(\log n + n\log d + \log(n+d))\big)$ \\
% & $\mathcal{O}\!\bigl(\log n + n \log d + \log (n + d)\bigr)$
% & $\mathcal{O}\!\bigl((n + d)(\log n + n \log d + \log (n + d))\bigr)$ \\
\bottomrule
\bottomrule
\end{tabularx}
\end{table*}

\subsection{A Rudimentary B-QML Experiment}

In light of the encouraging scalings summarized in Table~\ref{tab:complexity_comparison}, we propose an %next ask whether these asymptotic results can support a concrete 
experimental demonstration of blind quantum machine learning with near-term resources. This task is important because the practical value of the derived bounds lies not only in classifying asymptotic overheads (i.e. the differences between the columns in Table~\ref{tab:complexity_comparison}), but also in identifying protocols that remain implementable at nontrivial system sizes. Here, %the aim is therefore not only to present asymptotic resource estimates, but to 
we explicitly construct a blind implementation of a nontrivial quantum machine-learning routine and determine its exact overhead for a minimal working instance.

Motivated by this, we propose an experimental implementation of a blind quantum recommendation system (B-QRS), requiring $12$ logical qubits and $750$ qubits prepared via remote state preparation (RSP). %For this smallest nontrivial instance, the blindness overhead is known exactly rather than only asymptotically.In particular, 
Starting from a six-qubit circuit-level implementation, we explicitly transpile the computation into its brickwork graph-state form and obtain a resource count for the brickwork graph. Specifically, we find we need a brickwork graph $6$ logical qubits wide and $125$ qubits deep. This can be implemented on hardware holding $12$ logical qubits in memory using a cylindrical register (i.e. recycling the qubits in a brickwork column after its measurement), provided their memories are sufficiently long to allow for RSP of $750$ qubit states. This gives a concrete benchmark for the cost of blindness in this setting, and a goalpost for experimentalists looking to implement this protocol on near-term hardware. 

We consider the smallest nontrivial instance of the QRS algorithm introduced in \cite{Sawerwain_Wróblewski_2019}. Specifically, we use a database with two entries, each represented by a feature vector of size \(2\). We choose the database entries to be
\[
\mathcal{D}=\{(0,0),(0,1)\},
\]
so that the database can be prepared from the initial computational basis state using a single \(\mathrm{CX}\) gate. The user feature register consists of two qubits encoding
the query feature vector \(t=(0,1)\) in the computational basis. Starting from \(\lvert 00\rangle\), this register is initialized using a single \(X\) gate, completing the data-loading stage. The QKNN routine compares this user vector with the database entries to determine the recommendation. This setting corresponds to blind utilization of the recommendation system: the recommendation is inferred from the encoded database entries, and the encoding required for this inference is performed blindly.

The corresponding circuit is shown in Fig.~\ref{fig:experiment_circs2}a, while the corresponding decomposed version is presented in Fig.~\ref{fig:experiment_circs2}b. The associated brickwork graphs, both the encrypted and the colored representations, are omitted from the appendix and instead are hosted on the project’s GitHub repository \url{https://github.com/Nonsens24/brickwork_transpiler/blob/main/images/graphs/Brickwork%20graph%3A%20Minimal%20QRS.pdf} as their diagrams exceed the size of an A4 page. The experiment could also be extended to support $\varepsilon$-verifiability by interweaving test rounds with trap qubits as detailed elsewhere ~\cite{Fitzsimons_Kashefi_verify_bqc}.

\begin{figure}[htbp]
    \centering

    \begin{minipage}[t]{0.8\textwidth}
        \centering
        \includegraphics[width=\linewidth]{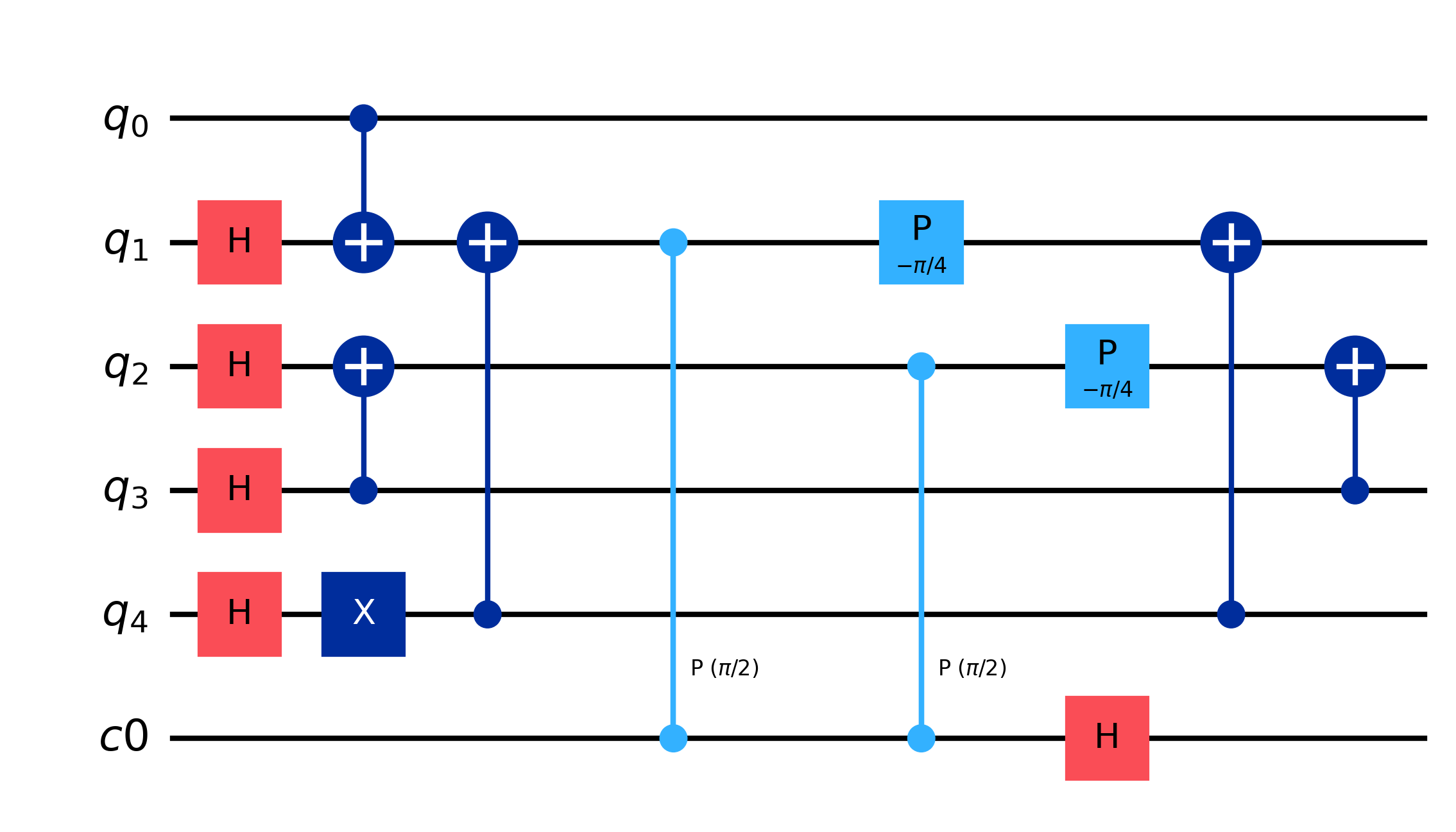}\\[-0.2em]
        \small (a)
    \end{minipage}

    % \vskip0.8em

    \begin{minipage}[t]{\textwidth}
        \centering
        \includegraphics[width=\linewidth]{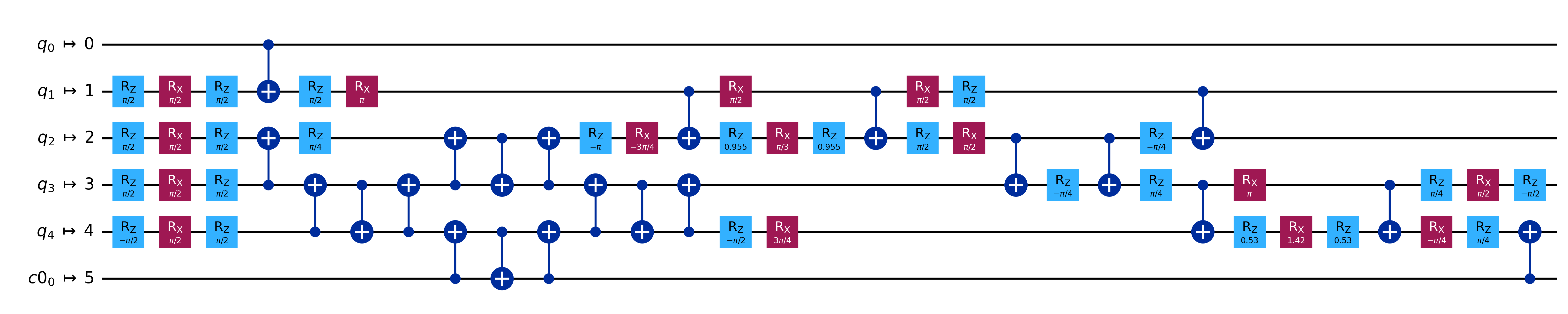}\\[-0.2em]
        \small (b)
    \end{minipage}

    \caption{
    \textbf{(a)} Non-decomposed minimal QRS circuit.
    \textbf{(b)} Decomposed circuit satisfying the gate constraints imposed by the brickwork graph.
    The minimal quantum recommendation system contains only the quantum \(k\)-nearest-neighbor (QKNN) and database-encoding components, based on the quantum recommendation system proposed by Sawerwain and Wróblewski \cite{Sawerwain_Wróblewski_2019}.
    }
    \label{fig:experiment_circs2}
\end{figure}

The second stage implements the quantum \(k\)-nearest-neighbor (QKNN) routine~\cite{qknn_svm,Sawerwain_Wróblewski_2019}. This consists of computing the Hamming distance between the user feature vector and each database entry, followed by a summation of these distances into an auxiliary register. The intermediate distance information is then uncomputed, leaving the database register in a state whose measurement probabilities encode the relative Hamming distances to the user vector. Measuring the database register therefore yields a recommendation biased towards the nearest entry. If desired, this bias can be amplified using Grover amplitude amplification~\cite{Grover1997, Sawerwain_Wróblewski_2019}, at the cost of increased circuit depth, and is therefore omitted in this experiment.

The full circuit requires only six qubits in total, one of which is a single ancilla used for QKNN. The main caveat is that the results must be post-selected on the measurement outcome of the ancilla qubit $c_0$, retaining only those runs where $c_0 = 0$. For the feature vector $\langle 0,1\rangle$, this occurs with probability
\begin{equation}
  \mathbf{P}(c_0=|0\rangle)
  \;=\;\frac{1}{L}\sum_{p}\cos^{2}\!\Big(\frac{\pi}{2l}\,d(t,r^{p})\Big)
  \;\rightarrow\;\frac{3}{4},
\end{equation}
where $L$ is the size of the database, $l=2$ the number of feature bits, and
$d(t,r^{p})$ the Hamming distance between the user feature vector $t$ and the
$p$th database entry $r^{p}$; the two entries have $d=0$ and $d=1$. Conditioned on
$c_0=0$, the database register is measured in the matching entry with probability
$2/3$ (Fig.~\ref{fig:four_plots_results_experiment}b). The probabilistic nature of $c_0$ can be classically mitigated with parallel computing and therefore does not impose restrictions on the complexity~\cite{Sawerwain_Wróblewski_2019}. We performed a numerical experiment using Graphix, verified against noiseless Qiskit Aer simulations. After aligning qubit ordering, the state vectors agreed up to global phase. Simulated histograms from \(4096\) shots after post-selection on \(c_0=0\) and the Graphix state-vector output distributions are presented in Appendix~\ref{app:experiment}, Figs.~\ref{fig:four_plots_from_statevector_gx} and~\ref{fig:four_plots_results_experiment}, respectively.

% Simulated experimental data was obtained using Graphix simulations, and the verified against noiseless state-vector simulations in Qiskit Aer. After aligning the qubit ordering with that of the transpiled circuit, the two state vectors were found to be equivalent up to global phase. Measured histograms were generated from \(4096\) shots after post-selection on \(c_0 = 0\) and are presented in Appendix~\ref{app:experiment},   Fig.~\ref{fig:four_plots_from_statevector_gx}, while the corresponding state-vector output distributions from Graphix are shown in Appendix \ref{app:experiment}, Fig.~\ref{fig:four_plots_results_experiment}, confirming agreement between the two approaches.

The key message of this example is that, for a minimal but nontrivial recommendation task, % we do not only estimate the asymptotic blindness overhead, but explicitly compile the circuit into its blind measurement-based form and determine the exact resource count. That is, a six-qubit circuit-level implementation becomes -- 
a blind implementation requires $12$ logical qubits, $750$ RSP qubits, and a brickwork graph with depth $125$. Although deliberately small, this example provides a concrete benchmark at a scale that is now directly relevant to current logical-qubit experiments, thereby linking the asymptotic resource analysis to experimentally accessible fault-tolerant regimes~\cite{Acharya2025}. 

However, we emphasize that we still require our \(12\) qubits to preserve their coherence long enough to make \(750\) RSP qubit states, \(750\) single-qubit measurements and rotations, and \(900\) CZ gates; even though each qubit is only used for $1/125$th of the protocol before being measured and reset, the effects of decoherence accumulate with depth as in any quantum circuit. Long lived qubits and fast operations, including RSP, will be needed; next-generation RSP of graph states will further decrease the needed resources \cite{Cai26}. 

\section{The Interplay of Learning and Hiding}

A central question throughout this work has been whether
machine learning and blind quantum computation are fundamentally
in tension. At first sight, they appear to be: machine learning
aims to extract structure from data, whereas blindness is designed
to conceal those data and the computation from the server.
These objectives nevertheless concern different parties: the
client should learn from the result, while the server should
learn nothing beyond the prescribed leakage. Classical approaches
such as fully homomorphic encryption
\cite{pyte2023homomorphic}, approximate homomorphic encryption
\cite{Yuan2025AHE}, and garbled circuits \cite{Yao_garb_circs}
also reconcile these objectives, with computational costs that
depend on the task and security requirements.

Our results identify regimes in which the additional cost of
blind quantum computation is limited. At the level of the
measurement-based implementation, blindness is obtained by
local randomization and adaptive measurement instructions on
the brickwork resource state. For the gate classes considered
here, the brickwork-depth bound scales with the original gate
count multiplied by the circuit width. For the QRS encoding,
this width grows only logarithmically with the number of
database entries \(L\) at fixed feature width. As shown in the
two QRS rows of Table~\ref{tab:complexity_comparison}, the
resulting blind depth bounds are
\(\mathcal{O}(L^2\log L)\) including database encoding and
\(\mathcal{O}(\sqrt{L}\log L)\) for inference only.
Thus, under the stated input-access assumptions, the
logarithmic transpilation factor preserves the sublinear
inference-depth bound.

This distinction is especially important in settings where
data preparation and inference can be separated. Protocol 2
of \cite{ubqc}, Universal Blind Quantum Computation with
Quantum Inputs, permits blind delegation on supplied quantum
states. In the QRS setting of
Appendix~\ref{app:bqrs_with_db}, treating the database encoder
as a trusted resource on the client side places preparation
outside the delegated inference stage
(Appendix~\ref{sec:qrs_no_db}). This separates the costs rather
than eliminating them: the trusted side must supply the input
states and perform the local operations required by the
quantum-input protocol. A logarithmic input-register size
does not by itself establish efficient preparation or coherent
database access, including in a QRAM-style architecture
\cite{Jaques2025qramsurveycritique}. The inference bound assumes
that all coherent operations required during inference are
implemented at the stated cost.

Taken together, our results establish a width-dependent upper
bound on the depth of transpiled blind brickwork computations
and identify regimes in which favorable asymptotic scaling
survives compilation. They do not by themselves establish an
advantage over classical privacy-preserving learning. Such a
comparison requires matched tasks, output guarantees,
data-access assumptions, and leakage models, together with
a common resource measure.
Appendix~\ref{app:comparison_scope} discusses these requirements
and distinguishes transpilation bounds from end-to-end
advantage claims.

\section{Conclusions and Future Work} \label{sec:conclusion}
\noindent

In this paper, we introduced a transpilation framework for blind quantum algorithms, and used this to analyze machine learning that combines the cryptographic guarantees of universal blind quantum computation with the computational advantages of quantum learning algorithms. 

Within this framework and using the associated transpiler to assemble bricks into a quantum brickwork state, MasonQ, we derived complexity bounds for blind implementations of several quantum machine learning algorithms and illustrated the construction through a minimal blind quantum recommendation system experiment using $k$-nearest neighbors. Taken together, these results provide a concrete basis for assessing the feasibility of privacy-preserving quantum machine learning across near- to long-term quantum hardware scales. Furthermore, using our toolkit any researcher can analyze the cost of making their circuits blind in the UBQC framework, as unlike prior transpilation tools MasonQ can be used for UBQC transpilation even for large-brickwork depths ($\sim10^7$ in this work).

More specifically, we obtained brickwork-depth scalings of $\mathcal{O}(\log^3 N)$ for the quantum Fourier transform, $\mathcal{O}(\mathrm{poly}\!\bigl(\kappa,\, s,\, 1/\varepsilon\bigr)\mathrm{polylog}(N))$ for HHL, $\mathcal{O}(L^2\log L)$ for a quantum recommendation system, and $\mathcal{O}\!\left[(n+d)^2\bigl(\log n + n\log d + \log(n+d)\bigr)\right]$ for quantum transformers. Here, $N$ denotes the dimension of the register’s state space, $\kappa$ the condition number, $s$ the sparsity, $\varepsilon$ the target precision, $L$ the number of entries in the database, n is the number of image patches, and d is the size of the feature vector associated with each patch, so the image is represented as an n×d array of patch embeddings. To the authors' knowledge, these constitute the first rigorous asymptotic upper bounds on the brickwork depth of the B-QML implementations considered here.

We then compared these explicit QML circuit implementations with our UBQC brickwork-depth upper bound, \(\mathcal{O}(\omega n)\), which applies to circuits composed of single-qubit rotations, adjacent and non-adjacent CX gates, and multi-controlled single-target \(U(2)\) gates, with each multi-controlled gate counted once and one additional clean ancilla available. This comparison shows that, for the circuit families studied here, blindness does not necessarily introduce only a generic worst-case overhead. In several cases, the blind implementations preserve nontrivial asymptotic structure inherited from the original quantum algorithm. In particular, adding blindness to these QML circuits does not always change the qualitative scaling of efficient QML algorithms; For instance, the blind HHL depth bound remains polylogarithmic in \(N\), the dimension of the \(N\times N\) system matrix, for fixed sparsity, condition number, and target accuracy. Therefore, when the quantum speedup of a QML protocol exceeds the security overhead, for which we provide upper bounds, the protocol can in principle achieve a two-fold quantum advantage: privacy and computational efficiency.\footnote{As well as, possibly, an advantage in energy consumption \cite{FellousAsiani2023,Meier2025,Hou2025,Abderrahim2025}, though this is not yet studied in the quantum network context necessary for blind quantum machine learning.}.

Practically, B-QML provides an avenue to explore private delegated quantum learning architectures  where privacy, efficiency, and user control are treated as foundational design requirements from the outset. Furthermore, B-QML brings together two tasks that are in tension in classical settings: learning and hiding. Our analysis shows that, in the quantum domain, this tension need not translate into prohibitive overhead at the level of computational-graph complexity. A useful way to understand this is through the structure of quantum information itself: it is both natively private (due to the sensitivity to unknown phase rotations) and efficient for information encoding and manipulation (due to the superposition principle). In the constructions studied here, the principal overhead associated with blindness arises not from an intrinsic incompatibility between privacy and quantum learning, but largely from architectural constraints, especially the need to realize non-nearest-neighbor interactions through sequences of SWAP gates. This overhead could be further reduced through non-planar graph constructions as proposed in \cite{Ma2024}, partial blindness \cite{Lee2025,poshtvan2025selectivelyblindquantumcomputation_ref1}, and further algorithmic design tailored for blinded quantum learning tasks, which we propose exploring in future works. We expect that these may reduce the asymptotic upper-bounds in complexity derived in the present work.

In practice non-asymptotic optimality is also desirable, as are non-universal implementations of B-QML based on other schemes beyond measurement based quantum computation with potentially reduced resource requirements and amenable to NISQ era devices and networks \cite{Li2024,pasini2026}. It is our intention that MasonQ will continue being developed, expanded beyond planar brickwork graphs, and utilized to support future research into blind quantum computation, both in the asymptotic regime and for the near-term implementation. 

%There is strong motivation at the present stage of quantum science to develop applications that are not only relevant to current NISQ era devices, but can also become practical to everyday people as quantum hardware matures.

%Blind quantum machine learning contributes to this roadmap by showing that privacy can be incorporated at the level of protocol design, rather than being added only as a secondary constraint after algorithmic performance is established. In this sense, B-QML provides not only a model for private delegated quantum learning, but also a route toward quantum learning architectures in which privacy, efficiency, and user control are treated as foundational design requirements from the outset. 

%MasonQ also enables new lines of investigation into brickwork transpilation beyond B-QML; although in this work we focus on blinding quantum machine learning, MasonQ can be used for UBQC transpilation, even for large-brickwork depths ($\sim10^7$ in this work). With MasonQ, researchers can input an (optimal) Qiskit circuit and receive an output that is asymptotically guaranteed to be the optimal measurement pattern in terms of brickwork depth. Of course, in practice non-asymptotic optimality is also desirable, as are non-universal implementations of B-QML based on other schemes beyond measurement based quantum computation \cite{Li2024,pasini2026}. It is our intention that MasonQ will continue being developed, expanded beyond planar brickwork graphs, and utilized to support future research into blind quantum computation. 

\section{Acknowledgments}

We gratefully acknowledge useful advice from Stephanie Wehner. TzBP gratefully acknowledges support from the Quantum Software Consortium Ada Lovelace Fellowship. This work was funded by the European Union's Horizon Europe research and innovation programme under grant agreement No.~101102140 – QIA Phase 1.

\section{Author Contributions}

RF wrote the MasonQ transpiler, derived the theorems associated with the results, and wrote the majority of the manuscript. TzBP conceptualized the project, provided day-to-day supervision, and co-wrote the manuscript with RF. SF provided feedback and guidance throughout on the project and paper.

\begin{appendices}

% An appendix contains supplementary information that is not an essential part of the text itself but which may be helpful in providing a more comprehensive understanding of the research problem or it is information that is too cumbersome to be included in the body of the paper.

%%=============================================%%
%% For submissions to Nature Portfolio Journals %%
%% please use the heading ``Extended Data''.   %%
%%=============================================%%

%%=============================================================%%
%% Sample for another appendix section			       %%
%%=============================================================%%

%% \section{Example of another appendix section}\label{secA2}%
%% Appendices may be used for helpful, supporting or essential material that would otherwise 
%% clutter, break up or be distracting to the text. Appendices can consist of sections, figures, 
%% tables and equations etc.

\bibliography{sn-bibliography}% common bib file

\section{Proofs} \label{app:statapp}

\subsection{Brickwork Alignment Bound}\label{thm:bw_alignment_bound_nor_prac}
\textbf{Proof of Theorem \ref{thm:abstract_bw_bound_2n}}
    Independently aligning bricks of an input matrix of depth \( n \) never increases the depth beyond \( 2n \).

\begin{proof}[Proof by Structural Induction]

\textbf{Base Case (\( n=1 \)):} Consider a brickwork graph containing of one brick. This brick can either already be aligned (i.e., placed correctly according to allowed positions) or require shifting one additional column. Thus, for the base case of \( n = 1 \), the resulting brickwork depth is at most \( 2n = 2 \). See cases 1 and 2 of Fig.~\ref{fig:proof_alignment} for an illustration.

\textbf{Inductive Hypothesis:} Assume that the following statement holds:  for any brickwork graph \( B_k \) consisting of \( k \) brick, the depth $n_k$ after alignment satisfies
\begin{equation}
    n_k \leq 2k.
\end{equation}

\textbf{Inductive Step (\( k \to k+1 \)):} Now consider extending the brickwork graph \( B_k \) by one more brick, forming the (unaligned) graph \( B_{k+1} \). Since the alignment algorithm processes bricks sequentially from left to right, at the point of the \( k+1 \)th brick either it is aligned (so depth $n_{k+1}=2k+1$) or it is not aligned, necessitating a shift by a single brick column (so depth $n_{k+1}=2k+2$). Thus, adding brick can shift the alignment horizontally by at most one additional brick column, and thus increases the depth by at most two. Consequently, the depth of the graph \( B_{k+1} \) satisfies:
\begin{equation}
    n_{k+1} \leq 2(k + 1).
\end{equation}

By the principle of mathematical induction, this establishes that shifting bricks for alignment purposes given an input matrix with depth \( n \) yields a depth of at most \( 2n \).
\end{proof}

\begin{figure}[htbp]
\centering
\includegraphics[width=0.6\textwidth]{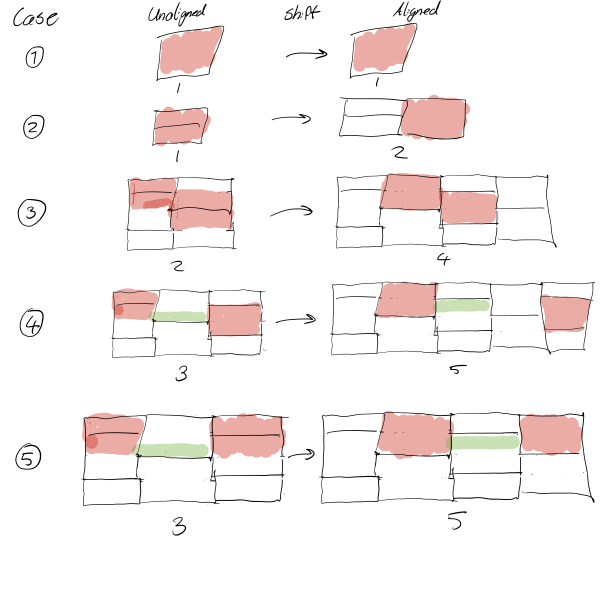}
\caption{Sketch providing the intuition on the upper bound of brick alignment.}
\label{fig:proof_alignment}
\end{figure}

\subsection{Non-adjacent CX Cost bound}\label{app:nonadj_cx_overhead_app}

\textbf{Proof of Lemma \ref{lem:nonadj_cx_overhead}}
Consider a quantum circuit mapped to a one-dimensional nearest-neighbor
architecture, and let $\omega$ denote an upper bound on the number of intermediate
qubits that must be crossed in order to realize any non-adjacent two-qubit
interaction in the circuit. Then any non-adjacent CX gate can be implemented
using at most
\[
6\omega+1
\]
nearest-neighbor $\mathrm{C}X$ gates. In particular, the nearest-neighbor routing overhead of a non-adjacent $\mathrm{C}X$ gate is $\mathcal{O}(\omega)$.
% Let a quantum circuit be mapped to a one-dimensional nearest-neighbor architecture. Suppose \(\omega\) is an upper bound on the number of intermediate qubits separating the two qubits involved in any non-adjacent logical $\mathrm{CX}$ gate. Then the total number of nearest-neighbor $\mathrm{CX}$ gates required to realize such a non-adjacent interaction is bounded by
% \begin{equation}
%     n_{\mathrm{CX}} \;\le\; 6\omega+1.
% \end{equation}
% Thus the overhead grows linearly with \(\omega\).

\begin{proof}
Let the logical $\mathrm{CX}$ act on two qubits that are not adjacent in the one-dimensional layout, and suppose that at most \(\omega\) intermediate qubits lie between them.

To realize this interaction on a nearest-neighbor architecture, one first repeatedly swaps one of the qubits towards the other until the two qubits become adjacent. Since there are at most \(\omega\) intermediate qubits, this requires at most \(\omega\) \textsc{swap} operations. After the local interaction has been applied, restoring the original qubit ordering requires at most another \(\omega\) \textsc{swap} operations. Hence the routing stage uses at most
\begin{equation}
    2\omega
\end{equation}
\textsc{swap} gates in total.

Now each adjacent \textsc{swap} admits the standard three-$\mathrm{CX}$ decomposition
\begin{equation}
    \textsc{swap}_{(i,j)}
    \;=\;
    \mathrm{CX}_{i, j}\,
    \mathrm{CX}_{j, i}\,
    \mathrm{CX}_{i, j},
\end{equation}
so the \(2\omega\) swaps contribute
\begin{equation}
    3(2\omega)=6\omega
\end{equation}
nearest-neighbor $\mathrm{CX}$ gates.

Once the two qubits are adjacent, the desired logical CX is applied using one
nearest-neighbor CX gate. Therefore the total number of nearest-neighbor CX
gates required is
\begin{equation}
n_{\mathrm{CX}} \leq 6\omega + 1.
\end{equation}

% Once the two qubits are adjacent, the remaining local implementation contributes only a constant number of nearest-neighbor $\mathrm{CX}$ gates. In the brickwork topology used by the transpiler, this contribution is bounded by \(2\). Therefore the total number of nearest-neighbor $\mathrm{CX}$ gates required is
% \begin{equation}
%     n_{\mathrm{CX}}
%     \;\le\;
%     6\omega+2.
% \end{equation}

Therefore the non-adjacent $\mathrm{CX}$ overhead is linear in \(\omega\), i.e. $\mathcal{O}(\omega)$.
\end{proof}

\subsection{Alignment Procedure Used by the Transpiler}\label{proof:bw_alignment_gaps}

This section analyzes the results of the rotation-aware alignment procedure
described in Appendix~\ref{app:outlineAlignmentAlgo}. They are not required for the
abstract scaling bound of
Theorem~\ref{thm:abstract_bw_bound_2n}, but establish a
factor-two depth bound for the placement and
rotation-insertion rules specified below.
% These results describe the alignment procedure used by the transpiler. They are
% not required for the abstract scaling bound of Theorem 1, but show that the
% implemented scheduling policy satisfies analogous locality and parity
% constraints.

% Macros assumed from the preamble: \ocol \place \gap \depth \newgap \top
%   and \floor{·} = \lfloor·\rfloor. The placement rule below is the one
%   implemented by align_cx_matrix.
\newcommand{\base}{\mathrm{base}}

\begin{definition}[Bricks and original column]\label{def:brick}
A \emph{brick} \(B\) is the unordered pair of matrix entries encoding a single
adjacent \(\mathrm{CX}\) gate; its set of rows is
\(\mathrm{rows}(B)\subset\mathbb{Z}\) with \(|\mathrm{rows}(B)|=2\). Its
\emph{original column} in the unshifted (that is, before alignment) instruction matrix is defined
\(\ocol(B):=\min\{c\mid B\text{ has an instruction in column }c\}\),
and we write \(\top(B):=\min\mathrm{rows}(B)\) for the top row.
\end{definition}

\begin{definition}[Row-isolated CX-matrix]\label{def:row-isolated}
A CX-matrix \(M\) is \emph{row-isolated} if no two bricks of \(M\) share both a row and a
column. Equivalently, if bricks \(A\neq B\) share a row and
\(\ocol(A)\le\ocol(B)\) then \(\gap(A,B):=\ocol(B)-\ocol(A)\ge 1\).
\end{definition}

\begin{definition}[Column extent and depth]\label{def:depth}
The original instruction matrix \(M\) has \(d\geq1\)
columns, indexed from \(0\) to \(d-1\). Each CX brick
has both endpoints in the same original column, and
each single-qubit cell is realisable by one brick.
The aligned CX extent is
\begin{equation}
D_{\mathrm{CX}}(M')
:=
\begin{cases}
1+\max_B\place(B), & \text{if CX bricks are present},\\
0, & \text{otherwise}.
\end{cases}
\end{equation}
The full depth \(D(M')\) is one plus the largest occupied
column index after single-qubit cells have been reinserted,
or zero for an empty output. Both quantities are measured
from the same column origin.
\end{definition}
%FIX A
% \begin{definition}[Row and global depth]\label{def:depth}
% Fix a row \(\rho\) and list the bricks meeting \(\rho\) from left to right as
% \(C_1,\dots,C_t\). The \emph{row depth} is
% \(\depth_\rho(M):=1+\sum_{i=1}^{t-1}\gap_\rho(C_i,C_{i+1})\), and the
% \emph{global depth} is \(\depth(M):=\max_\rho\depth_\rho(M)\).
% \end{definition}

\paragraph{Parity constraint.}
The brickwork embedding admits a brick with top row \(\top(B)\) only at columns
congruent to \(\top(B)\pmod 2\). Define the parity-aligned base column
\begin{equation}\label{eq:base}
  \base(B):=
  \begin{cases}
  \ocol(B), & \ocol(B)\equiv\top(B)\pmod 2,\\
  \ocol(B)+1, & \text{otherwise.}
  \end{cases}
\end{equation}

\paragraph{The algorithm.}
Bricks are processed in order of \(\ocol\), with ties broken by prioritizing the top brick, so
that when a brick \(B\) is placed, the bricks already placed on a row
\(r\in\mathrm{rows}(B)\) are exactly those lying to its left along \(r\). Let
\(L_r(B)\) denote the nearest such left neighbour on row \(r\), if one exists,
and set
\[
  \chi_r(B):=
  \begin{cases}
  1 & \text{a rotation occupies a column strictly between }\ocol(L_r(B))\text{ and }\ocol(B)\text{ on }r,\\
  0 & \text{otherwise.}
  \end{cases}
\]
The placement keeps each brick separated from its left neighbours by one column,
plus a second column when a rotation must be threaded between them. It computes
the \emph{floor}
\begin{equation}\label{eq:floor}
  \floor{B}:=\max_{\substack{r\in\mathrm{rows}(B)\\ L_r(B)\text{ exists}}}
            \bigl(\place(L_r(B))+1+\chi_r(B)\bigr),
\end{equation}
with the convention that the empty maximum is \(-\infty\), and selects the
smallest admissible column
\begin{equation}\label{eq:place}
  \place(B):=\min\bigl\{\,c\;:\;c\ge\max(\base(B),\floor{B}),\ \ c\equiv\top(B)\!\!\!\pmod 2\,\bigr\}.
\end{equation}
Unlike a gap-preserving rule, \eqref{eq:floor} reserves a column only where a
rotation is actually interleaved; an empty original gap is free to close.

\begin{lemma}[Locality of shifts]\label{lem:locality}
The placement of \(B\) is determined by its base column
and the placements of its immediate left neighbours on
the two rows it touches, together with the corresponding
rotation indicators. Bricks on disjoint rows impose no
direct constraint on one another, although delays may
propagate through intermediate bricks. No placement is
ever revised.
\end{lemma}

\begin{proof}
By \eqref{eq:floor}, \(\floor{B}\) reads only \(\place(L_r(B))\) and
\(\chi_r(B)\) for \(r\in\mathrm{rows}(B)\); no other brick enters, and
\(\base(B)\) depends only on \(B\). Placements are assigned once and never
changed.
\end{proof}

\begin{lemma}[Separation embodied by the floor]\label{lem:floor-lb}
Any placement that preserves per-row order, avoids
collisions, and reserves a free column whenever single-qubit
operations are interleaved between consecutive CX bricks
must satisfy
\begin{equation}
\place(B)\geq\place(L_r(B))+1+\chi_r(B)
\end{equation}
on every shared row \(r\) with an existing left neighbour.
The floor in \eqref{eq:floor} is the maximum of these
row-wise lower bounds and therefore enforces all such
separation requirements. The base-column and parity
constraints are imposed separately.
\end{lemma}

\begin{proof}
On row \(r\), brick \(B\) lies strictly to the right of \(L_r(B)\), so
\(\place(B)\ge\place(L_r(B))+1\). If \(\chi_r(B)=1\) a rotation must occupy a
column strictly between them, which forces \(\place(B)\ge\place(L_r(B))+2\).
Taking the maximum over the rows of \(B\) gives \eqref{eq:floor}.
\end{proof}
\begin{lemma}[Column bound]\label{lem:delay-bound}
Every CX brick satisfies
\begin{equation}
\place(B)\leq2\ocol(B)+1.
\end{equation}
\end{lemma}

\begin{proof}
We proceed by induction over the processing order.
Write \(t=\ocol(B)\). For any left neighbour
\(A=L_r(B)\),
\begin{equation}
\ocol(A)\leq t-1-\chi_r(B),
\end{equation}
because \(\chi_r(B)=1\) requires an original column
strictly between the two bricks. By the induction
hypothesis,
\begin{equation}
\begin{aligned}
\place(A)+1+\chi_r(B)
&\leq2\ocol(A)+2+\chi_r(B)\\
&\leq2t-\chi_r(B)
\leq2t.
\end{aligned}
\end{equation}
Taking the maximum over both rows gives
\(\floor{B}\leq2t\) whenever a predecessor exists.
For \(t=0\), no predecessor exists and
\(\place(B)=\base(B)\leq1\).
For \(t\geq1\), we also have
\(\base(B)\leq t+1\leq2t\).
Equation~\ref{eq:place} rounds the resulting lower bound
to the required parity, adding at most one column.
Thus \(\place(B)\leq2t+1\).
\end{proof}

\paragraph{Rotation insertion.}
For the full-depth bound, we use the following insertion
rule. Consider a single-qubit cell originally in column
\(j\), and let \(P\) and \(R\) be its nearest preceding
and following CX bricks on the same row, when present.
If \(P\) occupies original column \(j-1\), place the cell
at \(\place(P)+1\). Otherwise, if \(R\) occupies original
column \(j+1\), place it at \(\place(R)-1\).
In all remaining cases, place it at \(\place(P)+1\)
if \(P\) exists, or at column \(0\) otherwise.
Cells assigned to the same position are composed in
their original execution order and realised by one
single-qubit brick. Existing CX columns are not shifted.

\begin{lemma}[Rotation insertion]\label{lem:rotation-insertion}
Under the insertion rule above, every single-qubit cell
is placed without overlapping a CX endpoint on its row,
and the order of operations is preserved. Its destination
column is at most \(2d-2\).
\end{lemma}

\begin{proof}
Suppose first that both neighbouring CX bricks \(P\)
and \(R\) exist. Since a single-qubit cell lies between
them, \(\chi_r(R)=1\), and
Lemma~\ref{lem:floor-lb} gives
\begin{equation}
\place(R)\geq\place(P)+2.
\end{equation}
Both possible destinations, \(\place(P)+1\) and
\(\place(R)-1\), therefore lie strictly between the
two aligned bricks and contain no CX endpoint on that row.

If only \(R\) exists, the destination is either
\(\place(R)-1\) or \(0\). Both are nonnegative and
precede \(R\), since
\(\place(R)\geq\ocol(R)\geq1\).
If only \(P\) exists, the destination is
\(\place(P)+1\), after the final CX on that row.
If neither exists, the destination is \(0\).

Within an original gap, all cells are placed immediately
after \(P\), or at column \(0\) when \(P\) is absent,
except possibly the last cell, which may be placed
immediately before \(R\). Their destinations are
therefore non-decreasing in original execution order.
Composing cells in that order preserves the operation
on the row.

Finally, whenever \(P\) exists,
\begin{equation}
\place(P)+1
\leq2\ocol(P)+2
\leq2j
\leq2d-2.
\end{equation}
Whenever \(R\) exists,
\begin{equation}
\place(R)-1
\leq2\ocol(R)
\leq2d-2.
\end{equation}
The remaining destination, \(0\), satisfies the same
bound because \(d\geq1\).
\end{proof}

\begin{theorem}[Global depth bound]
\label{thm:global_depth_bound_impl}
Let \(M\) be a row-isolated instruction matrix with
\(d\geq1\) columns. Under the placement and
rotation-insertion rules specified above, the aligned
matrix satisfies
\begin{equation}
D(M')\leq2d.
\end{equation}
\end{theorem}

\begin{proof}
Lemma~\ref{lem:delay-bound} places every CX brick in
a column at most \(2d-1\).
Lemma~\ref{lem:rotation-insertion} places every
single-qubit cell in a column at most \(2d-2\).
The full column extent is therefore at most \(2d\).
\end{proof}

% \begin{remark}[Tightness]\label{rem:tight}
% The bound is attained for every \(d\geq1\).
% Take \(d\) CX bricks in consecutive original columns
% on the same pair of rows, with an odd top-row index.
% The first brick is placed in column \(1\), and the
% separation and parity constraints place each subsequent
% brick two columns later. The final occupied column is
% \(2d-1\), giving \(D(M')=2d\).
% \end{remark}
\begin{remark}[Tightness]\label{rem:tight}
The bound is attained for every \(d\geq1\).
Take \(d\) CX bricks in consecutive original columns
on the same pair of rows, with an odd top-row index.
The first brick is placed in column \(1\), and the
separation and parity constraints of \eqref{eq:place} place each subsequent
brick two columns later. The final occupied column is
\(2d-1\), giving \(D(M')=2d\) while \(D(M)=d\); the factor of two in
Theorem~\ref{thm:global_depth_bound_impl} is therefore unavoidable.
\end{remark}

% \begin{remark}[Tightness]\label{rem:tight}
% The factor of two is unavoidable. Take a single row of \(t\) unit-spaced bricks
% sharing one top-row parity. Their original columns alternate in parity, so the
% correction in \eqref{eq:place} adds a column at every step and the placed columns
% are \(0,2,4,\dots,2(t-1)\). Then \(\depth_\rho(M)=t\) while
% \(\depth_\rho(M')=2t-1\).
% \end{remark}

% \subsection{Optimality and Correctness of the Alignment Procedure}\label{proof:correctness_opt_bw_alignment}
% Optimality is meant in the sense of \emph{earliest-feasible} placement under the
% locality policy. These results describe the alignment procedure used by the transpiler. They are
% not required for the abstract scaling bound of Theorem 1, but show that the
% implemented scheduling policy satisfies analogous locality and parity
% constraints.
\begin{theorem}[Optimality under locality]
\label{app:thm:optimality}
For any feasible placement \(p\) satisfying
Definition~\ref{def:feasible},
\begin{equation}
\place(B)\leq p(B)
\end{equation}
for every CX brick \(B\).
Consequently, the algorithm minimizes the CX extent
and, under the fixed rotation-insertion rule of
Lemma~\ref{lem:rotation-insertion}, the full depth.
\end{theorem}

\begin{proof}
We proceed by induction over the processing order.
Suppose \(\place(A)\leq p(A)\) for all previously
processed bricks. For each left neighbor \(A=L_r(B)\),
feasibility gives
\begin{equation}
p(B)\geq p(A)+1+\chi_r(B)
     \geq\place(A)+1+\chi_r(B).
\end{equation}
Thus \(p(B)\geq\floor{B}\). Feasibility also requires
\(p(B)\geq\base(B)\) and the same parity as \(\top(B)\).
Hence \(p(B)\) is admissible in
Equation~\ref{eq:place}, whose least admissible value
is \(\place(B)\).

The insertion destination of each single-qubit cell
is fixed by its original neighbours and is either
\(\place(P)+1\), \(\place(R)-1\), or \(0\).
Each is a non-decreasing function of the relevant
CX placement. Pointwise minimality therefore also
minimizes the largest occupied column after insertion.
\end{proof}

\begin{theorem}[Semantic correctness]
\label{thm:sem_correctness_in_app}
Under the specified placement and rotation-insertion
rules, the aligned matrix \(M'\) preserves the order
of operations on every row, introduces no CX collisions,
respects the brickwork parity constraint, and preserves
the circuit operation up to global phase.
\end{theorem}

\begin{proof}
Lemma~\ref{lem:floor-lb} separates consecutive CX bricks
on each row by at least one column, preserving their
order and preventing collisions. Parity holds by
construction in Equation~\ref{eq:place}.
When single-qubit operations are interleaved, the floor
reserves an additional column.
Lemma~\ref{lem:rotation-insertion} places those operations
in their original order relative to neighbouring CX
bricks. Cells sharing a destination are composed in
execution order and represented by an equivalent
single-qubit brick. Thus all operations sharing a
qubit retain their order; changes in the order of
operations on disjoint qubits do not affect the circuit
operation.
\end{proof}

\begin{definition}[Feasible alignments under locality]\label{def:feasible}
A placement \(\place(\cdot)\) is \emph{feasible} if for every brick \(B\):
\begin{enumerate}
\item \(\place(B)\equiv\top(B)\pmod 2\) and \(\place(B)\ge\base(B)\);
\item on each shared row \(r\) with an existing left
neighbour,
\(\place(B)\ge\place(L_r(B))+1+\chi_r(B)\);
\item earlier placements are not changed when \(B\) is placed.
\end{enumerate}
It is \emph{earliest-feasible} if \(\place(B)\) is the least integer satisfying
(1)–(2) for each \(B\) in processing order.
\end{definition}

Constraining each brick only against its immediate left neighbours loses nothing:
the separations in (2) compose along a row, so they imply the corresponding
bounds against every earlier brick.

\begin{lemma}[The algorithm is earliest-feasible]\label{lem:efp}
For every brick \(B\), the column \(\place(B)\) from
\eqref{eq:floor}–\eqref{eq:place} is the least \(c\) meeting (1)–(2).
\end{lemma}

\begin{proof}
Condition (2) is exactly \(c\ge\floor{B}\) by Lemma~\ref{lem:floor-lb}, and (1)
is \(c\ge\base(B)\) with \(c\equiv\top(B)\pmod2\). Equation~\eqref{eq:place}
returns the least such \(c\).
\end{proof}

\section{Brickwork Graphs}

\subsection{Circuit Optimizations}\label{app:circuit_optimisations}

During decomposition, the transpiler considers many optimizations designed to make the circuit as short as possible. This section considers two of these strategies to reduce the final brickwork depth: initial mapping and routing. Consider the circuit in Fig. \ref{fig:opt_example}a consisting of one non-adjacent $\mathrm{CX}$ gate. Without these optimizations, the circuit would transpile to the graph shown in Fig. \ref{fig:opt_example}d. Applying the routing method, the depth of the circuit is reduced by half, swapping states to each other as can be seen in Fig. \ref{fig:opt_example}c. If we then also implement initial mapping in addition to our routing strategy, we change the order of the qubits such that the qubits of interest are adjacent, making the graph one brick deep, as shown in Fig. \ref{fig:opt_example}b.

\begin{figure*}[t]
    \centering

    % Top row
    \begin{minipage}[t]{0.48\textwidth}
        \centering
        \includegraphics[width=\linewidth,height=5cm,keepaspectratio]{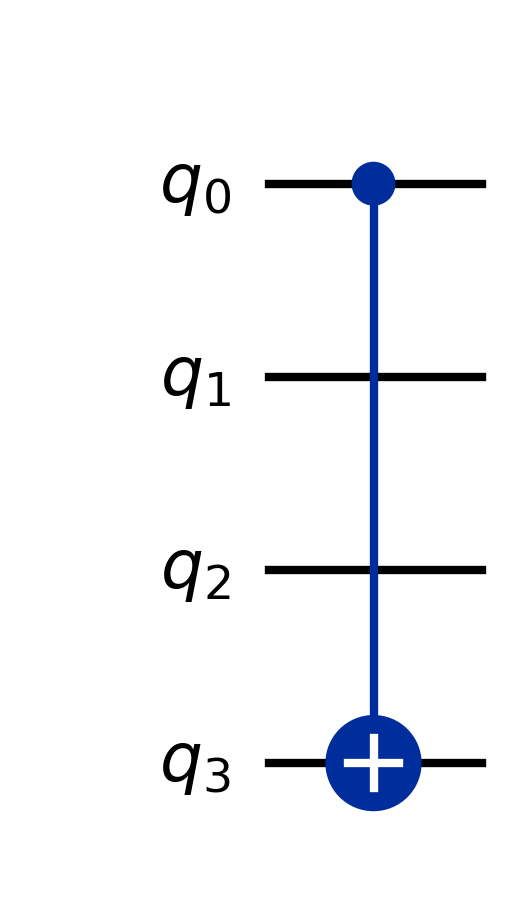}\\[-0.2em]
        \small (a)
    \end{minipage}
    \hfill
    \begin{minipage}[t]{0.48\textwidth}
        \centering
        \includegraphics[width=\linewidth,height=5cm,keepaspectratio]{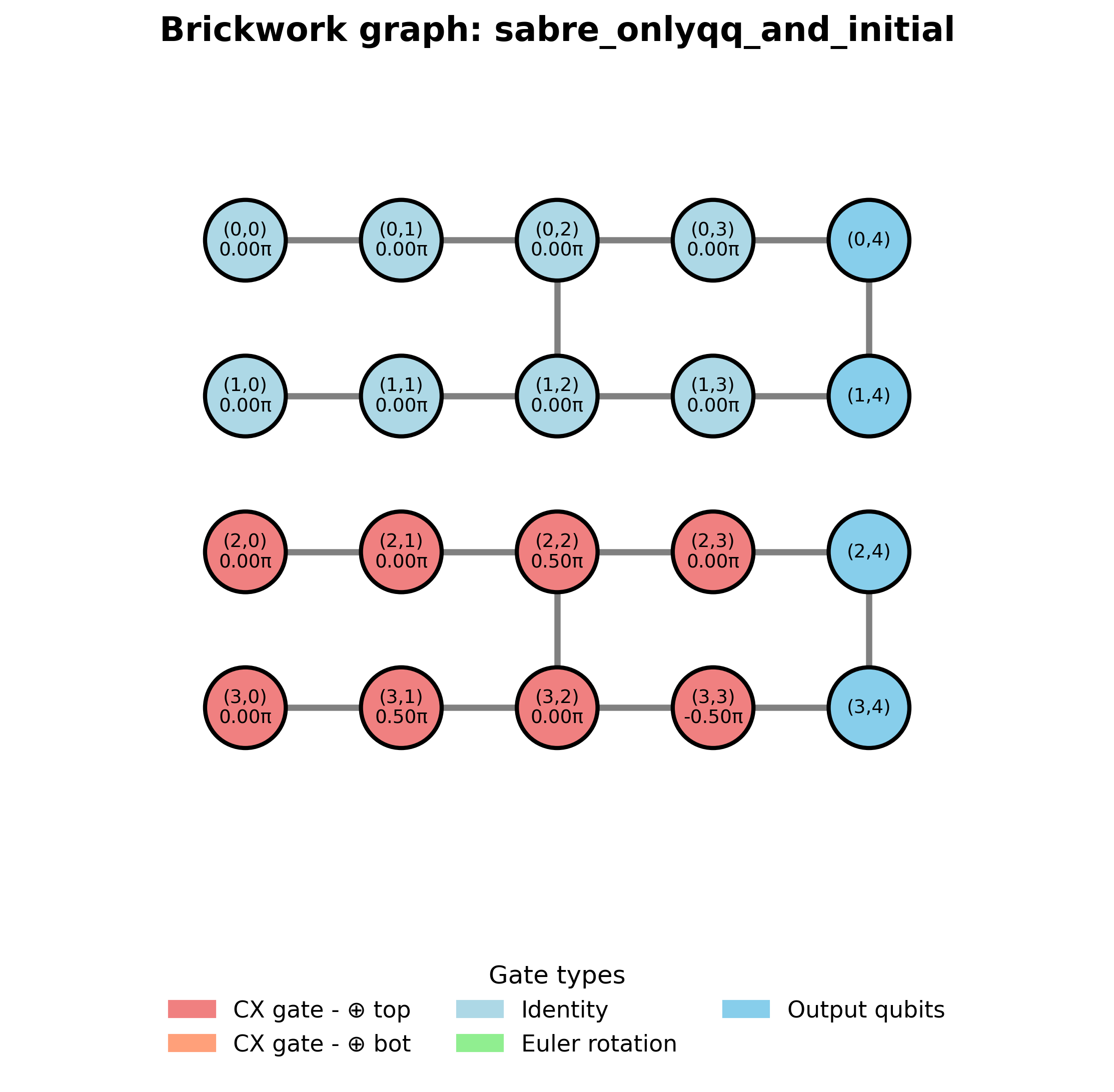}\\[-0.2em]
        \small (b)
    \end{minipage}

    % \vskip0.8em

    % Middle row
    \begin{minipage}[t]{\textwidth}
        \centering
        \includegraphics[width=\linewidth]{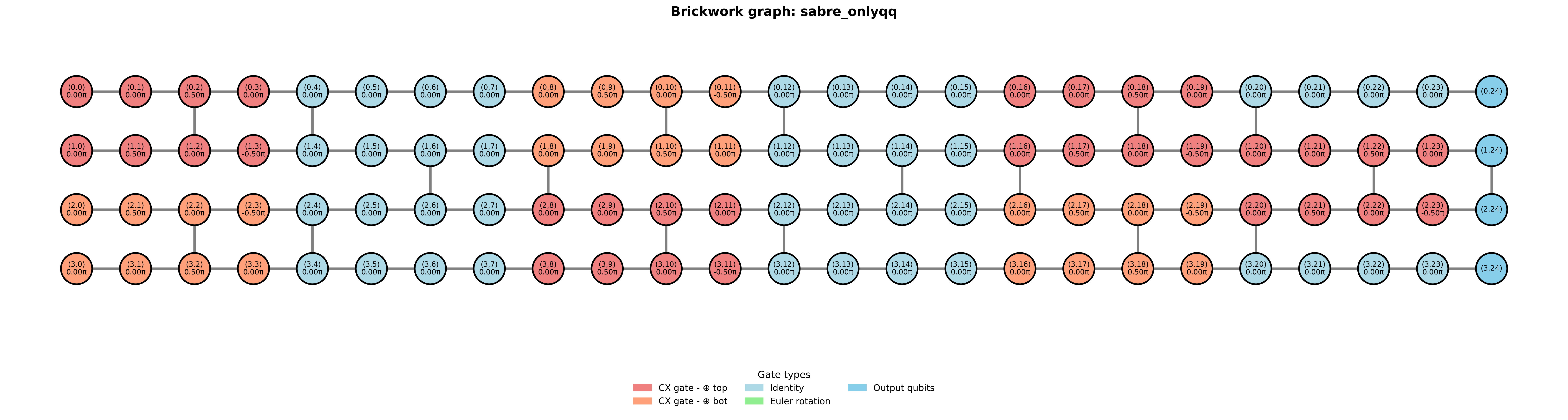}\\[-0.2em]
        \small (c)
    \end{minipage}

    % \vskip0.8em

    % Bottom row
    \begin{minipage}[t]{\textwidth}
        \centering
        \includegraphics[width=\linewidth]{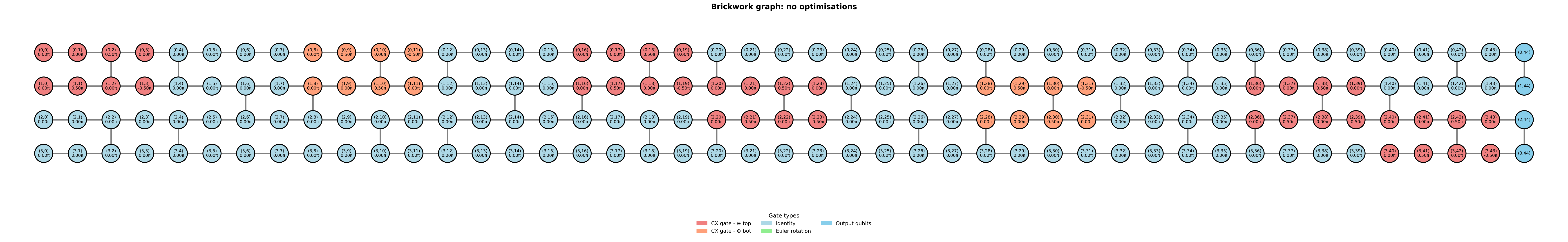}\\[-0.2em]
        \small (d)
    \end{minipage}

    \caption{
    Comparison of an example circuit and its corresponding brickwork graphs under different optimization strategies.
    \textbf{(a)} Example quantum circuit consisting of one non-adjacent $\mathrm{CX}$ gate.
    \textbf{(b)} Brickwork graph of the exemplary circuit optimized with initial mapping.
    \textbf{(c)} Brickwork graph of the exemplary circuit optimized with SABRE routing.
    \textbf{(d)} Brickwork graph of the exemplary circuit without routing or mapping optimizations.
    }
    \label{fig:opt_example}
\end{figure*}

\subsection{The Identity Brick}\label{app:id_brick}

From \cite{parsimonious_qgates} follows that measuring a qubit in the brickwork graph with angle \(0\) implements a Hadamard gate \(H\) in our framework. It follows that an identity brick, that is, a subgraph that teleports the logical state further along the brickwork graph unmodified, can be implemented by applying four such measurements in a single row, corresponding to one logical gate; see Fig.~\ref{fig:id_brick}. The resulting operation is
\begin{equation}
H_4 H_3 H_2 H_1 = H^4 = (H^2)^2 = I ,
\end{equation}
and hence the brick acts as the identity on the encoded state.
\begin{figure}[ht!]
\centering
\includegraphics[width=0.75\linewidth]{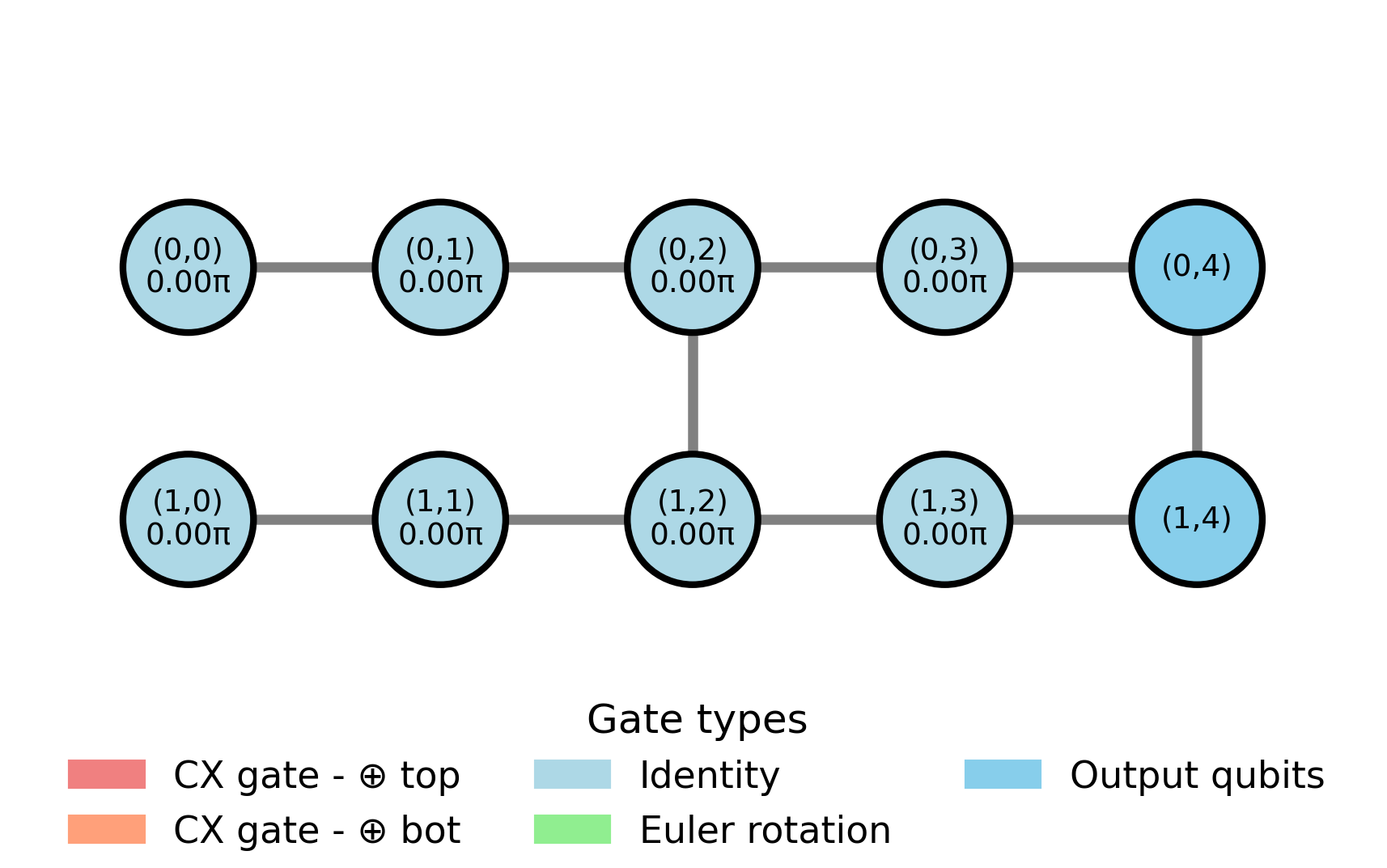}
\caption{Identity brick in the brickwork graph. Measuring the four qubits in each row with angle \(0\) implements four Hadamard gates, whose product is the identity, thereby teleporting the logical state forward without changing it.}
\label{fig:id_brick}
\end{figure}

\section{The Brickwork Transpiler} \label{app:bw_transpiler}
This section serves as an extension to the high level overview discussed in Section \ref{subsec:high_level_transpiler}.

\subsection{Instruction Matrix}\label{sec:instr_mat}
The instruction matrix encodes the essential information of a quantum circuit,
including the qubits each gate acts upon, the corresponding gate parameters, and
their execution dependencies. Moreover, the instruction matrix matches the
topology of a brickwork graph up to alignment, functioning as a primitive for
constructing the measurement graph.
An instruction matrix \(\mathcal{M}\) is generated by first constructing a
directed acyclic graph (DAG) of gate dependencies for the decomposed circuit.
Conceptually, if this dependency graph is rotated so that dependency order runs
from left to right, its layers define matrix columns: rows correspond to logical
qubits and columns correspond to sets of instructions that can be executed
contiguously. This representation matches the topology of the brickwork graph up
to the subsequent alignment of \(\mathrm{CX}\) bricks. The dependency-graph
construction also provides the scheduling structure used in the local optimality
argument below.
Figure~\ref{fig:instr_matrix_dag} shows the dependency graph generated for the
example circuit in Fig.~\ref{fig:instr_matrix_example_circuit}, rotated and
reflected to visually match the direction of time. Fig.~\ref{fig:instr_matrix_example_decomposed} and Fig.~\ref{fig:instr_matrix_example_graph}
show the decomposed circuit and the corresponding
brickwork graph generated by the transpiler respectively. The associated instruction matrices
are shown in Tables~\ref{tab:unaligned_instr_matrix}--\ref{tab:full_instr_matrix}.
\begin{figure}[ht]
    \centering
    \reflectbox{
    \rotatebox{270}{
        \includegraphics[width=0.30\textwidth]{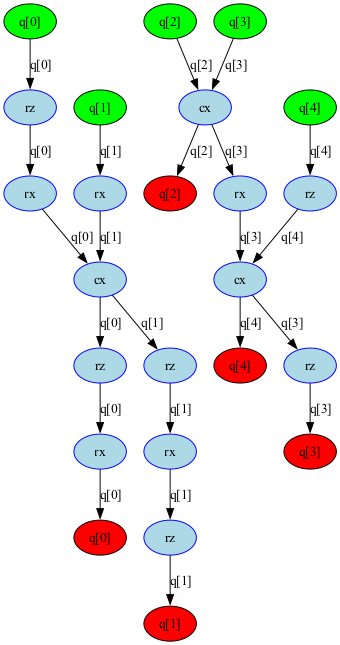}
    }}
    \caption{Dependency graph generated for the circuit in
    Fig.~\ref{fig:instr_matrix_example_circuit}, rotated by \(90^\circ\) and
    reflected over the vertical axis to visually match the direction of time.}
    \label{fig:instr_matrix_dag}
\end{figure}
\begin{figure}[t]
    \centering
    \includegraphics[
        width=0.85\linewidth,
        height=4.4cm,
        keepaspectratio
    ]{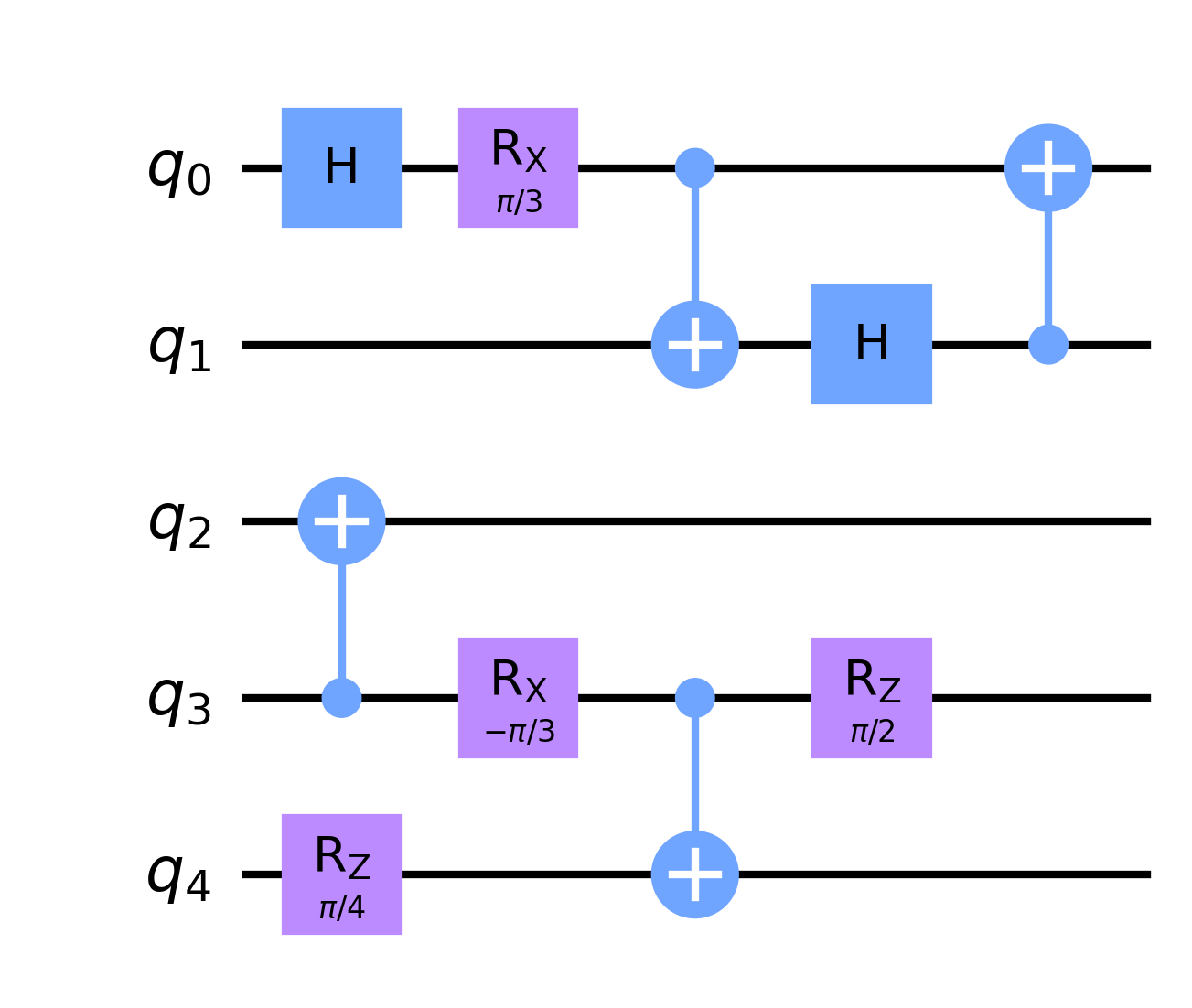}
    \caption{Example quantum circuit defined in Qiskit.}
    \label{fig:instr_matrix_example_circuit}
\end{figure}

\begin{figure}[t]
    \centering
    \includegraphics[
        width=0.85\linewidth,
        height=4.4cm,
        keepaspectratio
    ]{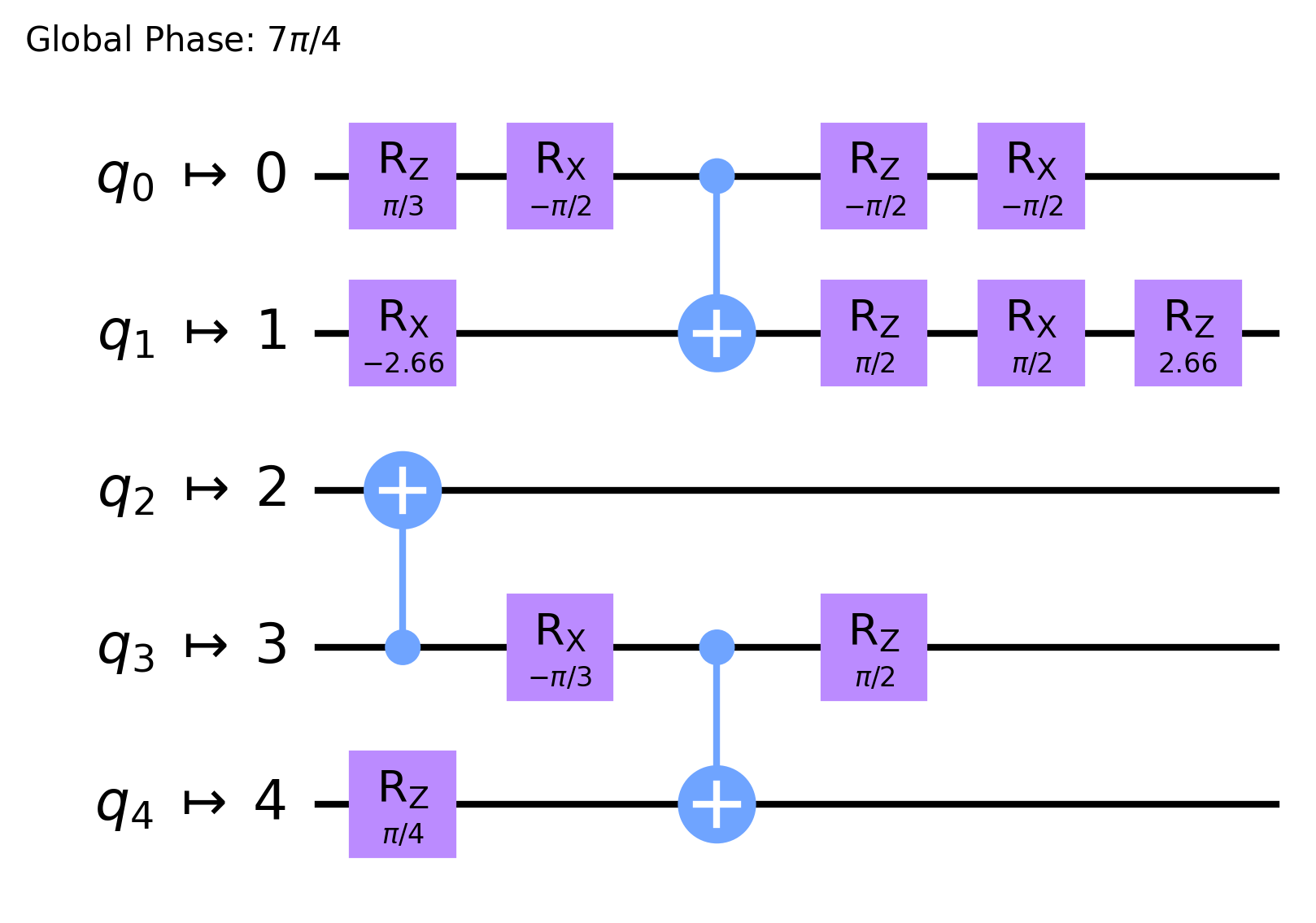}
    \caption{Decomposed quantum circuit used for the instruction-matrix construction.}
    \label{fig:instr_matrix_example_decomposed}
\end{figure}

\begin{figure}[t]
    \centering
    \includegraphics[
        width=\linewidth,
        keepaspectratio
    ]{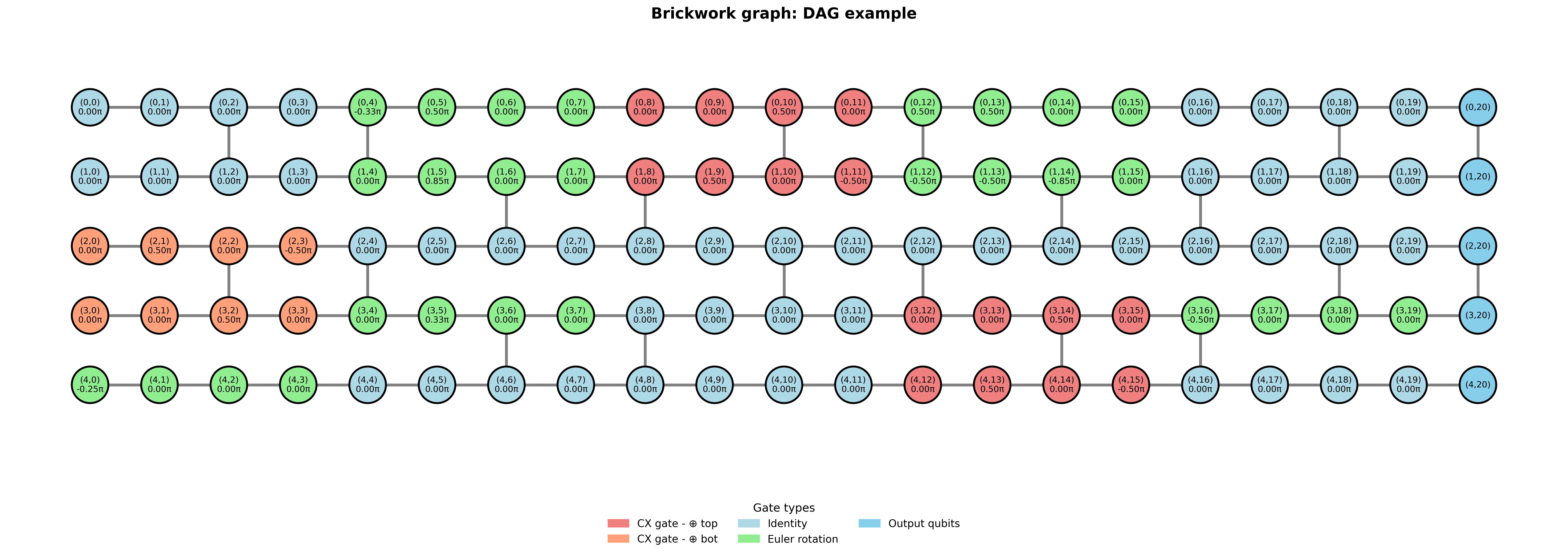}
    \caption{Brickwork graph generated by the transpiler from the decomposed circuit.}
    \label{fig:instr_matrix_example_graph}
\end{figure}
\begin{table}[ht]
\caption{Unaligned instruction matrix. The \(\mathrm{CX}\) gates carry an
integer identifier distinguishing the gate instance, together with a suffix
\texttt{c} or \texttt{t} indicating control or target. These identifiers are
used during alignment.}
\centering
\ttfamily
\begin{tabular}{l|c|c|c|c}
\toprule
\shortstack{Qubit - }{Brickwork Column} & Group 0 & Group 1 & Group 2 & Group 3 \\
\midrule
Qubit 0 & rz, rx & cx1c & rz, rx & - \\
Qubit 1 & rx & cx1t & rz, rx, rz & - \\
Qubit 2 & cx0t & - & - & - \\
Qubit 3 & cx0c & rx & cx2c & rz \\
Qubit 4 & rz & - & cx2t & - \\
\bottomrule
\end{tabular}
\label{tab:unaligned_instr_matrix}
\end{table}
\begin{table}[ht]
\caption{Aligned \(\mathrm{CX}\) matrix before reinserting single-qubit
rotations. When the single-qubit gates are inserted, an additional column may be
appended.}
\centering
\ttfamily
\begin{tabular}{l|c|c|c|c}
\toprule
\shortstack{Qubit - }{Brockwork Column} & Group 0 & Group 1 & Group 2 & Group 3 \\
\midrule
Qubit 0 & - & - & cx1c & - \\
Qubit 1 & - & - & cx1t & - \\
Qubit 2 & cx0t & - & - & - \\
Qubit 3 & cx0c & - & - & cx2c \\
Qubit 4 & - & - & - & cx2t \\
\bottomrule
\end{tabular}
\label{tab:aligned_instr_matrix}
\end{table}
\begin{table}[ht]
\caption{Fully aligned instruction matrix after reinserting single-qubit
rotations, corresponding to the layout of the graph in
Fig.~\ref{fig:instr_matrix_example_graph}.}
\centering
\ttfamily
\begin{tabular}{l|c|c|c|c|c}
\toprule
\shortstack{Qubit - }{Brickwork Column} & Group 0 & Group 1 & Group 2 & Group 3 & Group 4 \\
\midrule
Qubit 0 & rz, rx & - & cx1c & rz, rx & - \\
Qubit 1 & rx & - & cx1t & rz, rx, rz & - \\
Qubit 2 & cx0t & - & - & - & - \\
Qubit 3 & cx0c & rx & - & cx2c & rz \\
Qubit 4 & rz & - & - & cx2t & - \\
\bottomrule
\end{tabular}
\label{tab:full_instr_matrix}
\end{table}

\subsection{Optimality of the Brickwork Graph}\label{app:opt_graph}

The use of the dependency graph, which we denote $D$, sets up a proof of optimality (in terms of final brickwork depth) for the scheduling of bricks given the generator set and supported gates. The idea is as follows: every gate is inserted into the graph as soon as all the gates it depends on have been executed. Note that this proof only proves optimality from the starting point of a Qiskit circuit that is already optimally decomposed (in terms of circuit depth) into our native gateset $\mathcal{G}$; if the initial decomposition is sub-optimal, so will be our transpilation.

Let $D=(V,E)$ be a directed acyclic graph where each node $v \in V$ 
represents a gate with unit execution time, and each edge $(u,v)\in E$ 
represents the precedence constraint that gate $u$ must finish before 
gate $v$ can start, that is the dependency graph. We emphasize that the single-qubit rotations within the dependency graph presented in Fig. \ref{fig:instr_matrix_dag} are themselves separate vertices, but that in the dependency graph $D$ they are merged into a single brick that implements the combined single-qubit rotation $R_zR_xR_z$ to reduce the final brickwork depth of the algorithm they are embedded within.

\paragraph{Definitions}

\begin{itemize}
    \item For each $v\in V$, define its \emph{level} (or \emph{longest-path distance})
    \begin{equation}
        \ell(v) := \max_{P \leadsto v} |P|,
    \end{equation}
    where the maximum is taken over all directed paths $P$ through the dependency graph $D$ starting at its leftmost column and ending at column $v$,
    and $|P|$ is the number of nodes in $P$.
    \item The \emph{critical path length} is
    \begin{equation}
        L := \max_{v \in V} \ell(v)
    \end{equation} giving the path corresponding to the shortest scheduling time (i.e. where all dependencies have been evaluated), hence affecting the calculation of the minimum graph depth. 
    \item The \emph{Greedy schedule} $A$ executes, at each timestep $t=1,2,\dots$,
    all gates $v$ whose predecessors have completed by time $t-1$.
\end{itemize}

\paragraph{Key properties}

\begin{lemma}[Feasibility of Levels]\label{lemma:feasible_levels}
Any feasible schedule must execute gate $v$ no earlier than time $\ell(v)$.
\end{lemma}

\begin{proof}
By definition of $\ell(v)$, there exists a path of length $\ell(v)$ ending at $v$.
Since gates on a path must be executed sequentially, $v$ cannot be scheduled 
before time $\ell(v)$. 
\end{proof}

\begin{lemma}[Greedy Completes Levels On Time]\label{lemma:greedy_on_time}
For all $t \ge 1$, after timestep $t$ the Greedy schedule has executed 
exactly the set 
\[
    \{ v \in V : \ell(v) \le t \}.
\]
\end{lemma}

\begin{proof}[Proof by induction.]

\emph{Base case ($t=1$):} $\ell(v)=1$ iff $v$ has no predecessors, so Greedy
executes precisely those gates at time $1$.

\emph{Induction step:} Assume the claim holds for timestep $t-1$.
Consider timestep $t$. Any gate with $\ell(v)\le t$ has all predecessors 
satisfying $\ell(u)\le t-1$, hence by induction they are completed by time $t-1$. 
Thus Greedy executes all such $v$ at time $t$. Conversely, if $\ell(v) > t$, 
then some predecessor requires $\ge t$ steps, so $v$ cannot yet be executed. 
Hence the claim holds.
\end{proof}

\begin{theorem}[Optimality of Greedy]
The Greedy schedule completes all gates in exactly $L$ timesteps, which is optimal.
\end{theorem}

\begin{proof}
We establish optimality using a ``greedy stays-ahead'' argument.
By Lemma \ref{lemma:greedy_on_time}, after timestep $L$ all gates with $\ell(v) \le L$ 
are completed, which means all $v \in V$ are done since $\ell(v) \le L$ 
for all $v$. Thus Greedy finishes in $L$ steps.

By Lemma \ref{lemma:feasible_levels}, no schedule can complete any gate earlier than its 
level, so in particular no schedule can finish the DAG in fewer than $L$ steps.
Therefore Greedy stays ahead and uses the minimum possible timesteps to schedule the graph.
\end{proof}

For circuits with different gate-sets than our set $\mathcal{G}$ defined in Section \ref{sec:graph_comp}, the additional cost to translate to our specified gate-set is at most logarithmic in both depth and accuracy of approximation, as implemented by the {Solovay-Kitaev} algorithm \cite{kitaev2002classical,dawson2006solovay}. Thus, up to the standard Solovay--Kitaev polylogarithmic approximation overhead, the transpilation procedure preserves polynomial-size efficiency for arbitrary polynomial-size input circuits over finite universal gate-sets.

% Thus, in terms of asymptotic scaling, our transpilation algorithm is optimal for all quantum circuits.

\paragraph{Alignment to the brickwork topology}

We now add the brickwork constraint. Let $\mathcal{C}$ denote the set of
schedules that (i) respect precedence in the dependency graph $D$ and (ii) conform to the brickwork
topology (e.g., parity/lattice placement constraints). Let
$\tau(b)$ denote the \emph{earliest feasible timestep} for brick $b$
under these two constraints given the already scheduled predecessors.

For now we assume that the alignment procedure places each brick $b$ at time $\tau(b)$, i.e., the
minimum timestep consistent with both the dependency constraints and the
brickwork topology. That the transpiler enforces this is proved in Section \ref{sec:practical_alignment_depth_bound}, Lemma \ref{lem:floor-lb}.\\

\begin{corollary}[Depth minimality under brickwork]\label{cor:brickwork}
The aligned schedule has the smallest possible depth over $\mathcal{C}$.
\end{corollary}
\begin{proof}
Let $S^\star\in\mathcal{C}$ be any schedule. If it places some brick $b$ earlier
than the aligned schedule, then it places $b$ before $\tau(b)$, contradicting
Lemma~\ref{lem:floor-lb}. Hence no schedule in $\mathcal{C}$ finishes earlier.
\end{proof}

\subsection{Outline of the Alignment Algorithm}\label{app:outlineAlignmentAlgo}

Our algorithm begins by performing a pass (that is, a sweep) of the instruction matrix. Specifically, the $\mathrm{CX}$ brick alignment pass arranges two-qubit operations
(``$\mathrm{CX}$ bricks'') within the instruction matrix into a compact,
collision-free layout that respects the brickwork parity constraint and
preserves the logical order of execution on each row. Each brick corresponds to
the two tagged endpoints of a single $\mathrm{CX}$ instruction, which share a
group identifier, occupy two adjacent rows after routing, and originate from a
single column of the input instruction matrix. Bricks are collected, ordered by
their original column, and placed sequentially from left to right, with stable
tie-breaking by prioritizing placement of top bricks before bottom bricks.
\paragraph{Core idea.}
The alignment pass maintains row-local placement information describing the
already occupied portion of each row. When a new brick is processed, only
previously placed bricks sharing one of its rows can constrain its placement.
The algorithm therefore computes the earliest brick column that simultaneously
satisfies the brickwork parity constraint, preserves per-row order, avoids
collisions, and, when required, leaves space for single-qubit rotations that
must be inserted between consecutive $\mathrm{CX}$ bricks.
\paragraph{Parity-compatible base column.}
For a brick $B$, let $\mathrm{ocol}(B)$ denote its original column and let
$\top(B)$ denote its top row. The brickwork pattern admits a $\mathrm{CX}$ brick
with top row $\top(B)$ only in columns with matching parity. The parity-aligned
base column is therefore
\begin{equation}
\mathrm{base}(B)=
\begin{cases}
\mathrm{ocol}(B), & \mathrm{ocol}(B)\equiv \top(B) \pmod 2,\\
\mathrm{ocol}(B)+1, & \text{otherwise}.
\end{cases}
\end{equation}
\paragraph{Placement floor.}
The placement floor is the smallest column index at which a brick can be placed
without violating the row-local constraints imposed by previously placed
bricks. Different but closely related choices of this floor lead to different
alignment policies.
A conservative gap-preserving policy maintains, for each row, an occupancy
frontier that records how far the rightmost placed brick on that row has been
shifted relative to its original column. This policy preserves original
separations between overlapping bricks and therefore guarantees that no
originally available column is lost. It is simple and robust, but may introduce
unnecessary empty columns when the original gap contains no operation that must
be preserved.
The implementation used for the resource estimates in this work uses a tighter
rotation-aware policy. For each row $r$ touched by a brick $B$, let $L_r(B)$
denote the most recently placed brick on row $r$, if it exists. If no such brick
exists, row $r$ imposes no lower bound beyond the parity-compatible base
column. If $L_r(B)$ exists, then $B$ must be placed at least one column to its
right in order to preserve per-row order. Moreover, if the original instruction
matrix contains a single-qubit rotation on row $r$ in a column strictly between
$\mathrm{ocol}(L_r(B))$ and $\mathrm{ocol}(B)$, then an additional free column
must be reserved so that the rotation can later be inserted without colliding
with a $\mathrm{CX}$ brick.
Equivalently, define
\begin{equation}
    \chi_r(B) =
    \begin{cases}
        1, & \text{if row \(r\) of \(B\) admits an interleaved rotation},\\
        0, & \text{otherwise}.
    \end{cases}
\end{equation}
The rotation-aware row-wise constraint is
\begin{equation}
\mathrm{place}(B)\geq \mathrm{place}(L_r(B)) + 1 + \chi_r(B).
\end{equation}
Taking the maximum of these constraints over the rows of $B$ gives the
rotation-aware placement floor. The selected placement column is the smallest
column that is at least both this floor and $\mathrm{base}(B)$, and whose parity
matches $\top(B)$.
\paragraph{Rotation insertion.}
After all $\mathrm{CX}$ bricks have been aligned, the single-qubit rotations
from the original instruction matrix are reinserted next to their neighbouring
$\mathrm{CX}$ bricks. If a rotation originally had a $\mathrm{CX}$ neighbour on
its left, the algorithm places it immediately to the right of the aligned
$\mathrm{CX}$ when that cell is free. Otherwise, if it had a $\mathrm{CX}$
neighbour on its right, the algorithm places it immediately to the left of that
aligned $\mathrm{CX}$ when possible. Rotations with no neighbouring
$\mathrm{CX}$ are placed in the leading rotation column. The rotation-aware
floor ensures that a free column is available whenever a rotation must be
threaded between consecutive $\mathrm{CX}$ bricks on the same row.
\paragraph{Collision avoidance and efficiency.}
Because every brick is placed strictly to the right of the most recently placed
brick on each row it touches, the aligned matrix preserves per-row program
order and introduces no $\mathrm{CX}$ collisions. Under the conservative
gap-preserving policy, original gaps are preserved. Under the rotation-aware
policy used in the present implementation, empty original gaps are allowed to
close unless they are needed for an interleaved rotation. This gives a tighter
layout while preserving the same parity and ordering constraints. The algorithm
is deterministic for the fixed processing order. Its cost is linear in the
number of brick--row incidences, apart from the lookup used to detect
interleaved rotations, which can be implemented with a sorted per-row list of
rotation columns.

 \subsection{Alignment Algorithm Depth Bound}\label{sec:practical_alignment_depth_bound}

 Theorem~\ref{thm:abstract_bw_bound_2n} gives the abstract factor-two alignment
bound used in the scaling analysis. Here we summarize the constructive alignment
procedure used by the transpiler and relate it to the detailed proofs in
Appendix~\ref{proof:bw_alignment_gaps}. The implemented procedure follows a
local earliest-feasible placement policy: for a fixed dependency-ordered
instruction matrix, each $\mathrm{CX}$ brick is placed as early as possible
subject to the brickwork parity constraint, per-row ordering, and the reservation
of a free, unoccupied column whenever an interleaved single-qubit rotation must later be
inserted. 
\paragraph{Locality policy.}
The instruction matrix is scanned in non-decreasing order of the bricks'
original columns, with stable tie-breaking by prioritizing placement of top bricks over bottom bricks. When the current
brick \(B\) is placed, earlier placements are not revised. The placement of
\(B\) is constrained only by the most recently placed brick on each row touched
by \(B\), together with the parity constraint imposed by the brickwork graph.
Empty gaps are allowed to close; an additional column is reserved only
when the original instruction matrix contains a single-qubit rotation strictly
between two consecutive $\mathrm{CX}$ bricks on the same row.
\paragraph{Identifiers and order.}
Each two-qubit \(\mathrm{CX}\) gate appears on two rows and carries a global
identifier \texttt{gid}, tying the two endpoints together and fixing the
control/target roles. We write \(\tau\in\mathbb{N}\) for the original column of
the brick and \(\top(B)\) for its top row. Bricks are processed in
non-decreasing \(\tau\), with stable tie-breaking. The statements below
summarize the formal results proved in Appendix~\ref{proof:bw_alignment_gaps}.
\paragraph{Locality of shifts.}
When placing a brick \(B\), the selected column \(\place(B)\) depends only on
the already placed left-neighbour bricks that share a row with \(B\). Bricks on
disjoint rows are independent, and no placement is ever revised.
\emph{Intuition.} The placement floor for \(B\) is computed row-wise based on the
nearest already placed brick on each row touched by \(B\), together with whether
a single-qubit rotation brick must be inserted between the two bricks on that row.
\paragraph{Earliest-feasible placement.}
For each brick \(B\), the algorithm chooses the earliest (that is, leftmost) column that
simultaneously (i) satisfies the brickwork parity constraint, (ii) preserves
per-row order with respect to the immediate left neighbours of \(B\), (iii)
leaves a free column for each interleaved rotation, and (iv) leaves all earlier
placements unchanged. Consequently, among placements obeying this locality
policy, the algorithm minimizes every \(\place(B)\) and hence minimizes the
aligned depth.
\emph{Intuition.} The row-wise constraints induce a necessary lower bound
(``floor''); the algorithm takes the smallest parity-compatible column above
that floor.

% \paragraph{Gap expansion bound.}
% For consecutive bricks \(A,B\) on a row, with
% \(\ocol(A)<\ocol(B)\),
% \begin{equation}
% 1\ \le\ \newgap(A,B)\ \le\ \gap(A,B)+1.
% \end{equation}
% Thus the aligned gap is always positive, so row order is preserved, while the
% gap can exceed the original gap by at most one column.
% \emph{Intuition.} Empty original gaps may close, but parity rounding can add at
% most one column. If a rotation lies in the original gap, the floor reserves an
% additional column for it.
% \paragraph{Row-depth blow-up.}
% For every row \(\rho\),
% \begin{equation}
% \depth_\rho(M')\ \le\ 2\,\depth_\rho(M)-1,
% \end{equation}
% where \(M'\) is the aligned matrix returned by the algorithm.
% \emph{Intuition.} Each original gap contributes at most itself plus one additional gap (which will need to be filled in with identity bricks), and
% there are at most \(\depth_\rho(M)-1\) gaps on row \(\rho\).
% \paragraph{Global depth bound.}
% If \(M\) is row-isolated and \(\depth(M)=n\), then the aligned matrix \(M'\)
% satisfies
% \begin{equation}
% \depth(M')\ \le\ 2n-1\ \le\ 2n.
% \end{equation}
% \emph{Intuition.} Apply the row-depth bound to each row and then maximize over
% rows.

\paragraph{Gap bound.}
For consecutive bricks \(A,B\) on a row \(r\), with
\(\ocol(A)<\ocol(B)\),
\begin{equation}
\newgap(A,B)\ \ge\ 1+\chi_r(B).
\end{equation}
Thus the aligned gap is always positive, so row order is preserved, and a free
column is kept whenever a rotation must be inserted between \(A\) and \(B\).
The aligned gap is not bounded above by \(\gap(A,B)+1\): the placement of \(B\)
is governed by the larger of the floors of its two rows, so on the row with the
smaller floor the gap can grow by more than one column. For bricks on rows
\((0,1)\) at column \(0\), on rows \((2,3)\) at columns \(0\)--\(3\), and on rows
\((1,2)\) at column \(4\), the last brick is placed at column \(7\), so on row
\(1\) the gap grows from \(4\) to \(7\). The depth is bounded instead through the
column bound below.
\emph{Intuition.} Empty original gaps may close, parity rounding adds at most one
column, and delay accumulated on the other row of \(B\) is inherited.
\paragraph{Column bound.}
Every CX brick satisfies
\begin{equation}
\place(B)\ \le\ 2\,\ocol(B)+1,
\end{equation}
and every reinserted single-qubit brick is placed at a column at most \(2d-2\)
(Lemmas~\ref{lem:delay-bound} and~\ref{lem:rotation-insertion}).
\emph{Intuition.} By induction over the processing order: each left neighbor
\(L_r(B)\) has \(\ocol(L_r(B))\le\ocol(B)-1-\chi_r(B)\), so the floor is at most
\(2\,\ocol(B)\); parity rounding adds at most one column.
\paragraph{Global depth bound.}
If \(M\) is row-isolated and \(\depth(M)=n\), then the aligned matrix \(M'\)
satisfies
\begin{equation}
\depth(M')\ \le\ 2n,
\end{equation}
and for every row \(\rho\), \(\depth_\rho(M')\le 2\,\depth_\rho(M)\). The bound
is attained (Remark~\ref{rem:tight}).
\emph{Intuition.} Apply the column bound to the last brick of each row and then
maximize over rows.

\paragraph{Semantic correctness.}
The aligned matrix preserves per-row program order, introduces no
\(\mathrm{CX}\)-brick collisions, respects the brickwork parity constraint, and
leaves a free column between consecutive bricks whenever a rotation is
interleaved.
\emph{Intuition.} Consecutive new gaps on any row are at least one, so there is
no reorder or collision. Parity is enforced by construction, and the
rotation-aware floor reserves the required intermediate column.
\paragraph{Depth bounded.}
Combining the global depth bound with the earliest-feasible placement result
shows that, for a fixed row-isolated instruction matrix and under the stated
locality policy, the transpiler enforces the alignment depth bound by
construction:
\begin{equation}
\depth(M')\ \le\ 2\,\depth(M).
\end{equation}
The complete implementation is available on our Github page:
\url{https://github.com/Nonsens24/brickwork_transpiler}.

\subsection{UBQC Graph Encryption}\label{app:encryption}
This subsection describes the transpiler's implementation of angle randomization in the BFK protocol~\cite{ubqc}, distinguishing its static representation from the adaptive execution required for blindness. The resource-scaling results in this work depend only on the
dimensions of the brickwork graph and on the measurement pattern prior to the UBQC
encryption layer. Blinding in the BFK protocol randomizes measurement angles without
altering the graph or the structure of the pattern, so the depth and qubit-count
bounds of Section~\ref{sec:results} are unchanged by its application.
To this end, the transpiler provides a function \texttt{bfk\_encoder}, which takes as
input a brickwork \texttt{Pattern} and returns both a new brickwork
\texttt{Pattern} carrying randomized displayed angles and a log of Alice's
randomization data. An optional argument allows the removal of byproduct commands
from the returned \texttt{Pattern}, thereby better reflecting the opaque nature of the
protocol.
The measurement basis angle $\delta_{x,y}$ is defined as
\begin{equation}
\delta_{x,y} = \phi'_{x,y} + \theta_{x,y} + \pi r_{x,y},
\end{equation}
where
\begin{equation}
\phi'_{x,y} = (-1)^{s_{x,y}^X} \phi_{x,y} + s^Z_{x,y}\pi,
\end{equation}
with $s^X$ and $s^Z$ denoting the dependency terms on prior measurement outcomes:
\begin{equation}
s^X_{x,y} = \bigoplus_{i \in D_{x,y}} s_i, 
\qquad 
s^Z_{x,y} = \bigoplus_{i \in D'_{x,y}} s_i.
\end{equation}
Here $\phi_{x,y}$ is the unencrypted adaptive measurement angle, $\theta_{x,y}$ is
Alice's random preparation angle drawn from $\{0,\pi/4,\dots,7\pi/4\}$, and
$r_{x,y}\in\{0,1\}$ is Alice's random bit. The dependency corrections $s^X_{x,y}$
and $s^Z_{x,y}$ are computed dynamically during execution from prior measurement
outcomes.
The usage of the function is illustrated below:
\begin{lstlisting}
encoded_pattern, log_alice = bfk_encoder(bw_pattern, remove_dependencies=False)
\end{lstlisting}
The list \texttt{log\_alice} contains Alice-side randomization data and dependency
information in a structured format, for example:
\begin{lstlisting}
{'node': (1, 1), 'phi': 0.5, 'theta': 0.75, 'r': 0, 'x_byproducts': {(1, 2)},
 'z_byproducts': {(0, 2), (1, 3)}, 's_domain': set(), 't_domain': set()},
\end{lstlisting}
The dependency domains are inherited from the underlying Graphix pattern and need
not be empty for general brickwork patterns. The current static-pattern utility
stores randomized displayed angles and retains the Graphix dependency domains. In
the full BFK protocol, however, the quantity $\delta_{x,y}$ must be evaluated
adaptively at measurement time after the relevant prior outcomes are known; the
randomization terms $\theta_{x,y}$ and $\pi r_{x,y}$ should therefore be treated as
Alice-side classical data rather than as a complete replacement for adaptive
protocol execution. Accordingly, the encryption layer is used in this work only for
illustration, and the format above records the information needed to inspect the
generated encrypted schedule.
In addition, a visualization scheme was implemented to represent encoded graphs.
Unlike the noise visualization (Appendix~\ref{app:noise_elab}, Fig.~\ref{fig:noise}), this scheme removes color indicators corresponding
to gate types and instead labels the nodes with either symbolic angles
$\theta_{x,y}$ or their explicit numerical values. An example is provided in
Fig.~\ref{fig:encoded_delta_graph}.
\begin{figure}[ht]
\centering
\includegraphics[width=1\linewidth]{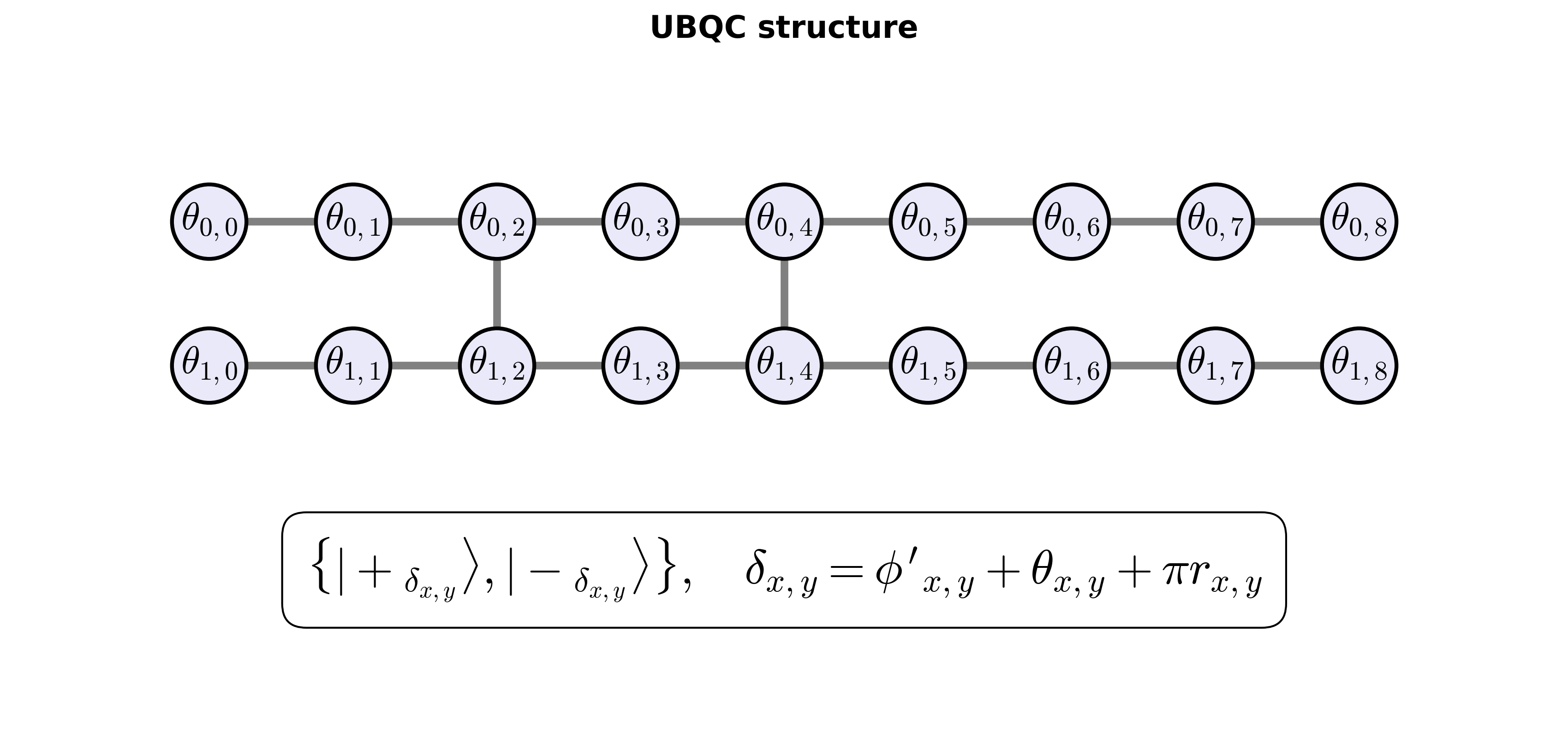}
\caption{An encrypted brickwork graph visualized without color-coding to conceal gate information.}
\label{fig:encoded_delta_graph}
\end{figure}

\subsection{Noise Injection}\label{app:noise_elab}
Although a comprehensive analysis of noise processes is beyond the scope of this
paper, the transpiler is designed to interface with noise models supported by the
underlying Graphix simulation stack. In particular, depolarizing-noise simulation
is supported through Graphix's noise-model interface rather than through a
MasonQ-specific physical noise channel. The MasonQ noise module re-exports the
relevant Graphix noise classes and adds helpers for noisy simulation,
density-matrix extraction, and fidelity evaluation.
For visualization and exploratory studies, one may also consider phenomenological
perturbations of measurement angles. Such perturbations are useful for debugging
and for illustrating sensitivity to measurement-basis errors, but they should
not be identified with a physical depolarizing channel. A depolarizing channel
acts on the quantum state, whereas a measurement-angle perturbation changes the
measurement basis.
The transpiler additionally provides a rendering routine that color-codes graph
nodes by their measurement angle on a continuous heat map, making it easy to see
where angles deviate from their nominal values without inspecting the underlying
parameters. Figure~\ref{fig:noise} illustrates this mode. The routine is a
renderer: it displays whatever angles the supplied pattern carries and does not
itself generate noise. This functionality is intended as a diagnostic and
extensibility feature rather than as a noise model entering the complexity
analysis.
\begin{figure}[ht]
\centering
\includegraphics[width=1\linewidth]{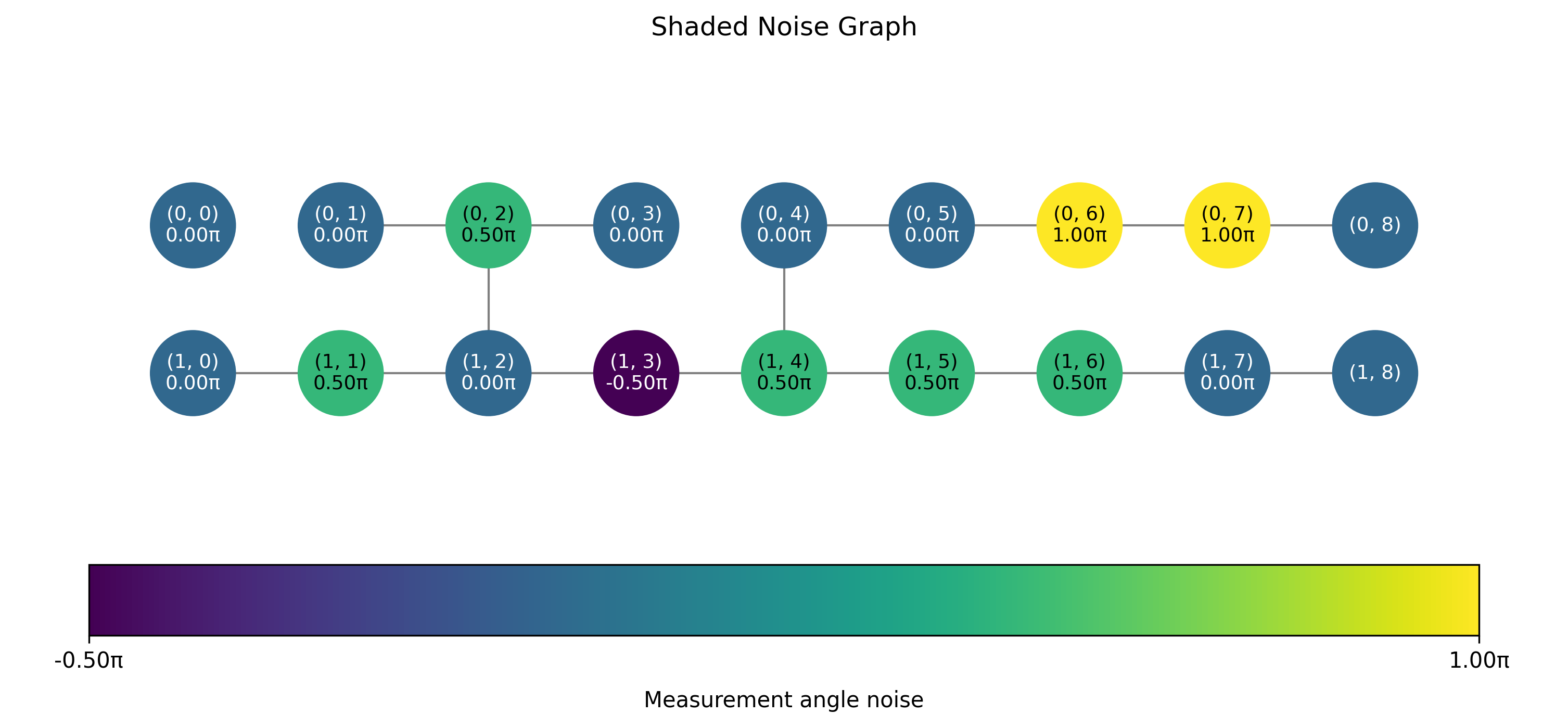}
\caption{A brickwork graph visualized with a color-coded heat map based on the
measurement angles carried by the supplied pattern. The renderer displays angle
values for diagnostic visualization and does not itself define a physical noise
model.}
\label{fig:noise}
\end{figure}

\section{Experiment supplementary material}\label{app:experiment}

Figure~\ref{fig:four_plots_results_experiment} shows the simulated recommendation counts for all four two-bit user feature vectors, using the database entries \(00\) and \(01\). Figure~\ref{fig:four_plots_from_statevector_gx} shows the corresponding recommendation counts after SABRE layout and routing.

\begin{figure*}[!ht]
    \centering

    \begin{minipage}[t]{0.48\textwidth}
        \centering
        \includegraphics[width=\linewidth]{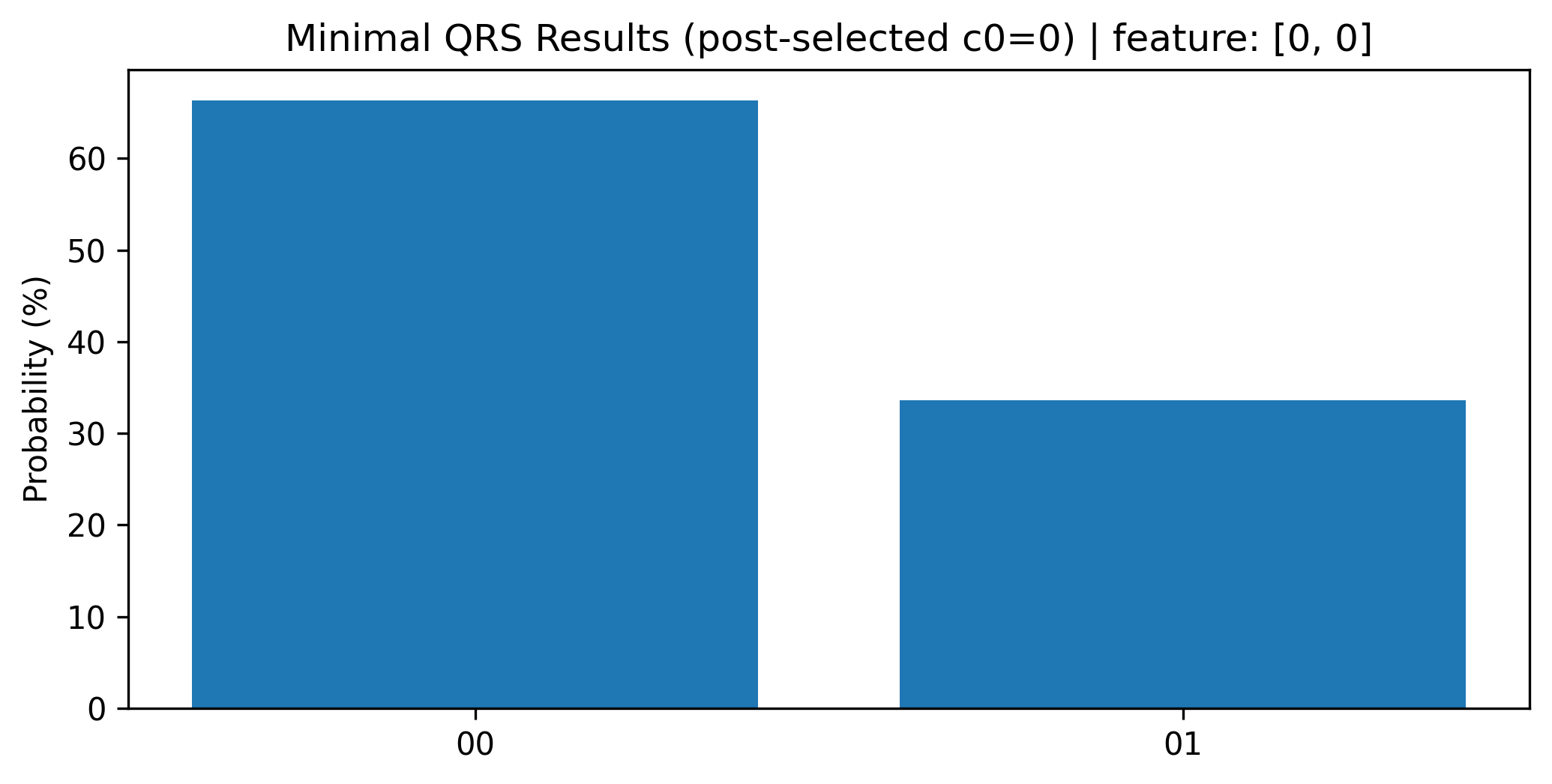}\\[-0.2em]
        \small (a)
    \end{minipage}
    \hfill
    \begin{minipage}[t]{0.48\textwidth}
        \centering
        \includegraphics[width=\linewidth]{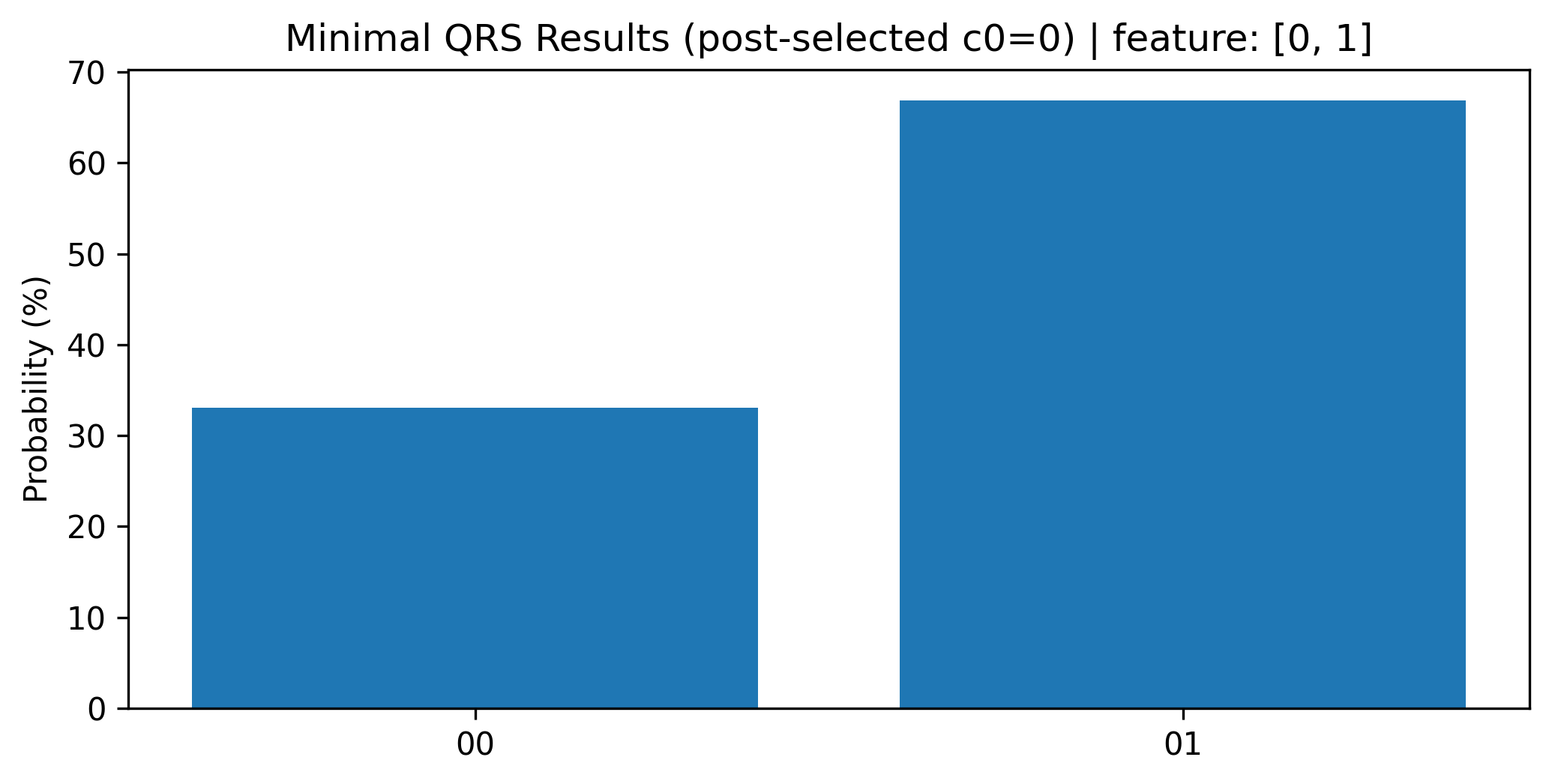}\\[-0.2em]
        \small (b)
    \end{minipage}

    \vskip0.8em

    \begin{minipage}[t]{0.48\textwidth}
        \centering
        \includegraphics[width=\linewidth]{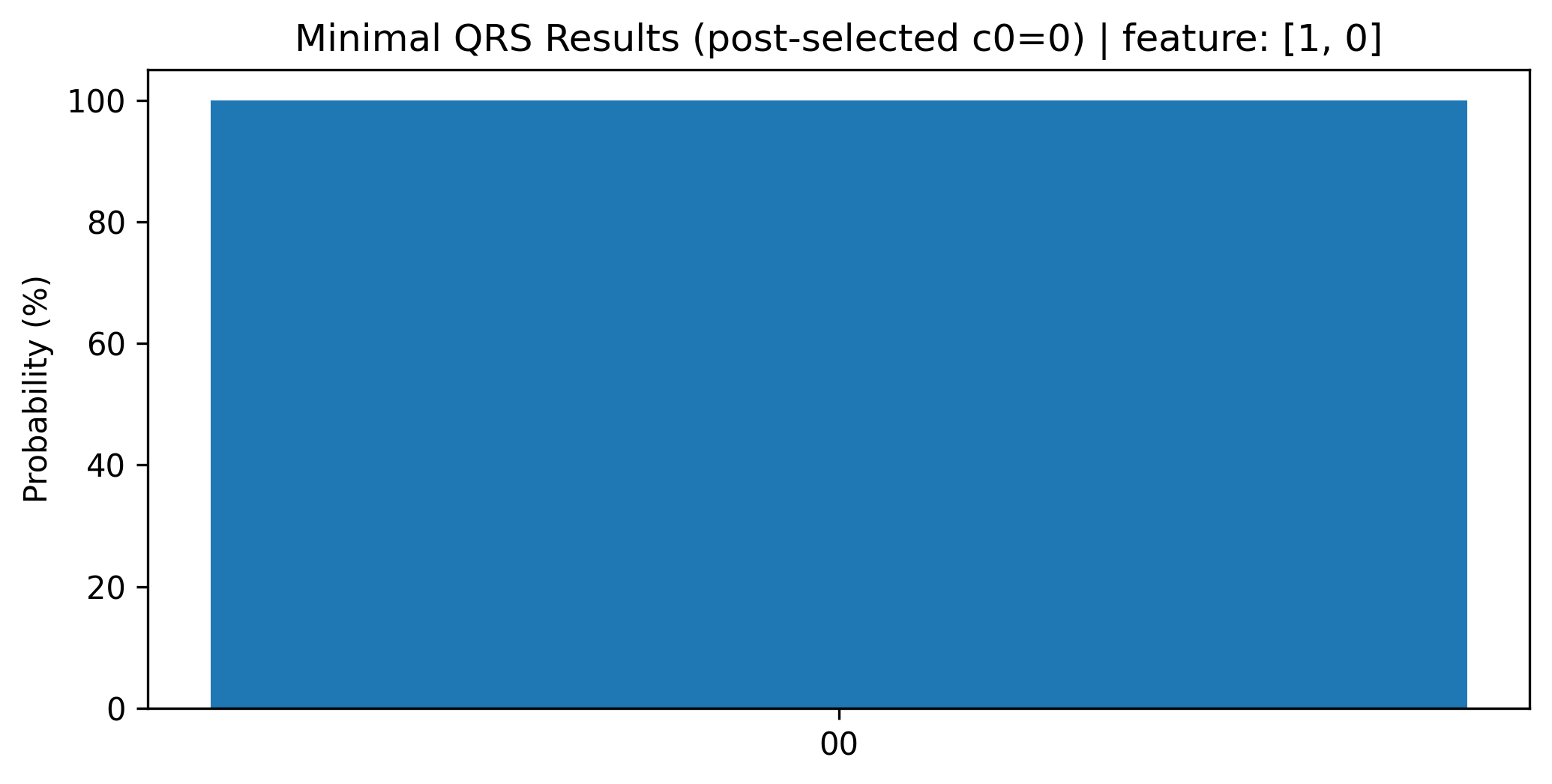}\\[-0.2em]
        \small (c)
    \end{minipage}
    \hfill
    \begin{minipage}[t]{0.48\textwidth}
        \centering
        \includegraphics[width=\linewidth]{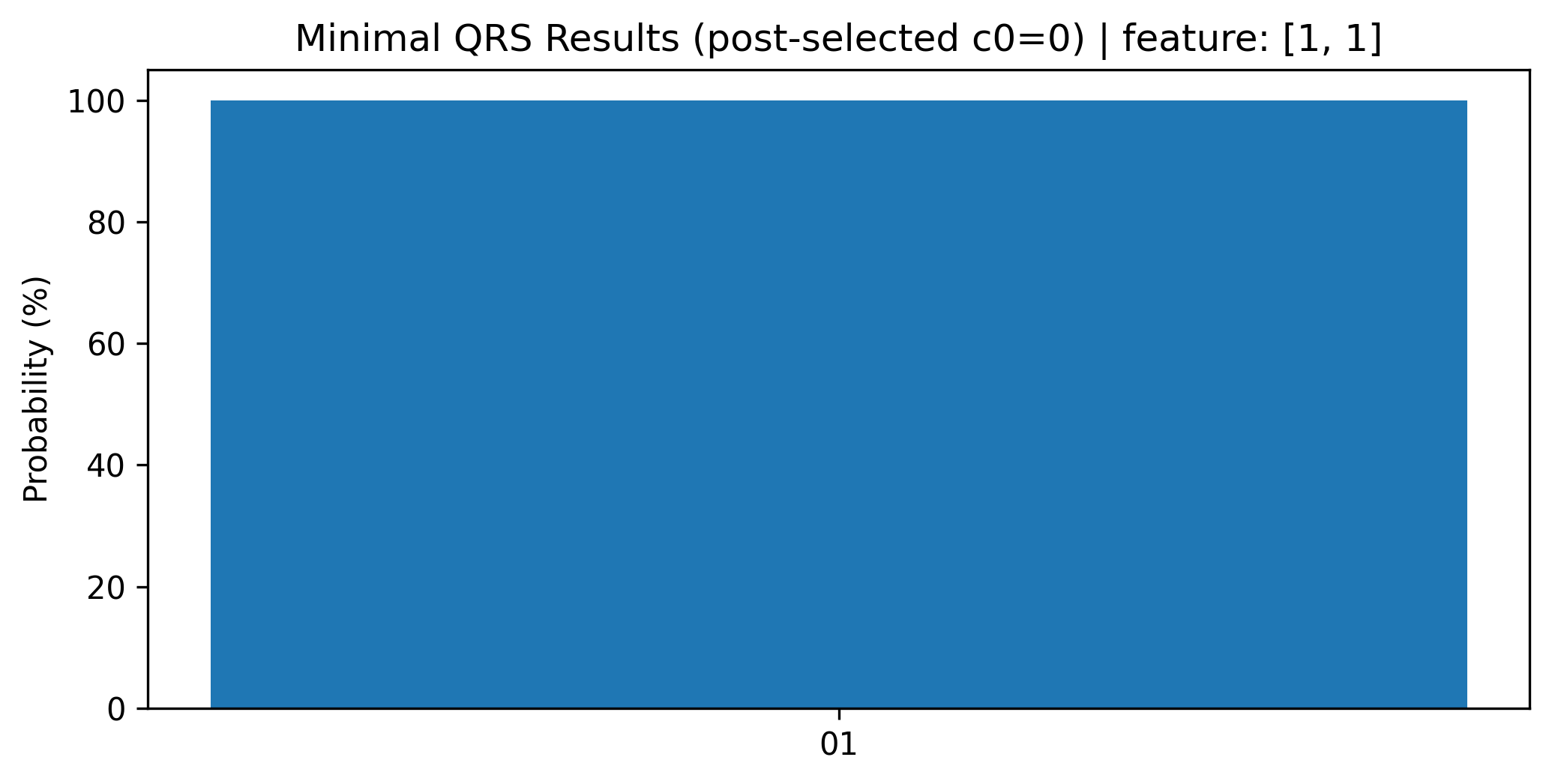}\\[-0.2em]
        \small (d)
    \end{minipage}

    \caption{
    Output counts for four feature-vector recommendations based on the same database encodings, \(00\) and \(01\).
    \textbf{(a)} Feature vector \(\langle 0,0\rangle\).
    \textbf{(b)} Feature vector \(\langle 0,1\rangle\).
    \textbf{(c)} Feature vector \(\langle 1,0\rangle\).
    \textbf{(d)} Feature vector \(\langle 1,1\rangle\).
    }
    \label{fig:four_plots_results_experiment}
\end{figure*}

% \begin{figure*}[htbp]
%     \centering

%     \begin{minipage}[t]{\textwidth}
%         \centering
%         \includegraphics[width=\linewidth,height=5cm,keepaspectratio]{figures/experiment/minimal_recommendation_circuit.png}\\[-0.2em]
%         \small (a)
%     \end{minipage}

%     \vskip0.8em

%     \begin{minipage}[t]{\textwidth}
%         \centering
%         \includegraphics[width=\linewidth]{figures/experiment/Decomposed_minimal_recommendation_circuit.png}\\[-0.2em]
%         \small (b)
%     \end{minipage}

%     \caption{
%     The proposed minimal quantum recommendation system, containing only the QKNN and database-encoding components, based on the quantum recommendation system presented by Sawerwain and Wróblewski \cite{Sawerwain_Wróblewski_2019}.
%     \textbf{(a)} The non-decomposed version of the minimal QRS circuit.
%     \textbf{(b)} The decomposed circuit obeying the gate constraints imposed by the brickwork graph.
%     }
%     \label{fig:experiment_circs}
% \end{figure*}
\begin{figure*}[t]
    \centering

    \begin{minipage}[t]{0.48\textwidth}
        \centering
        \includegraphics[width=\linewidth]{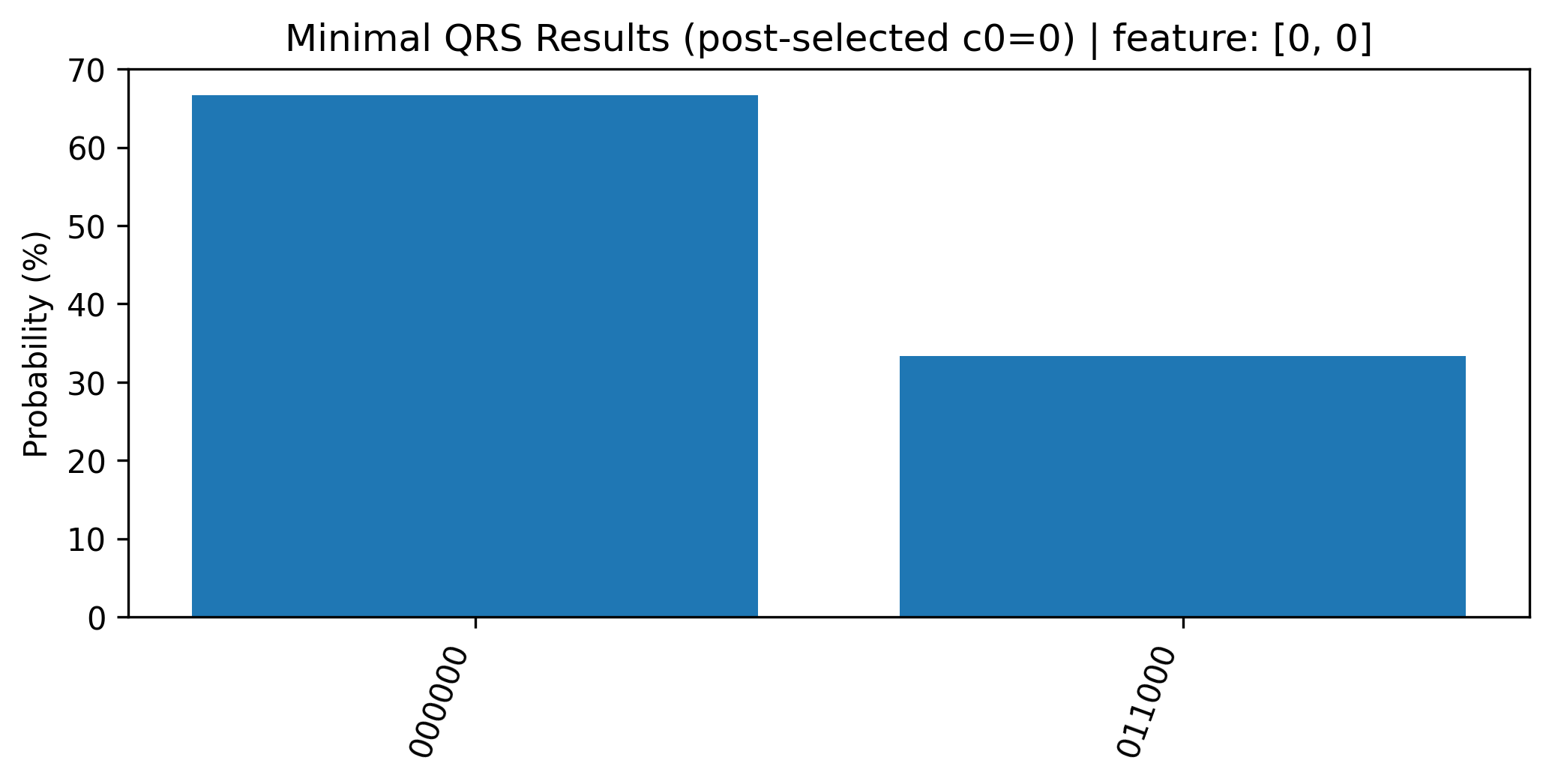}\\[-0.2em]
        \small (a)
    \end{minipage}
    \hfill
    \begin{minipage}[t]{0.48\textwidth}
        \centering
        \includegraphics[width=\linewidth]{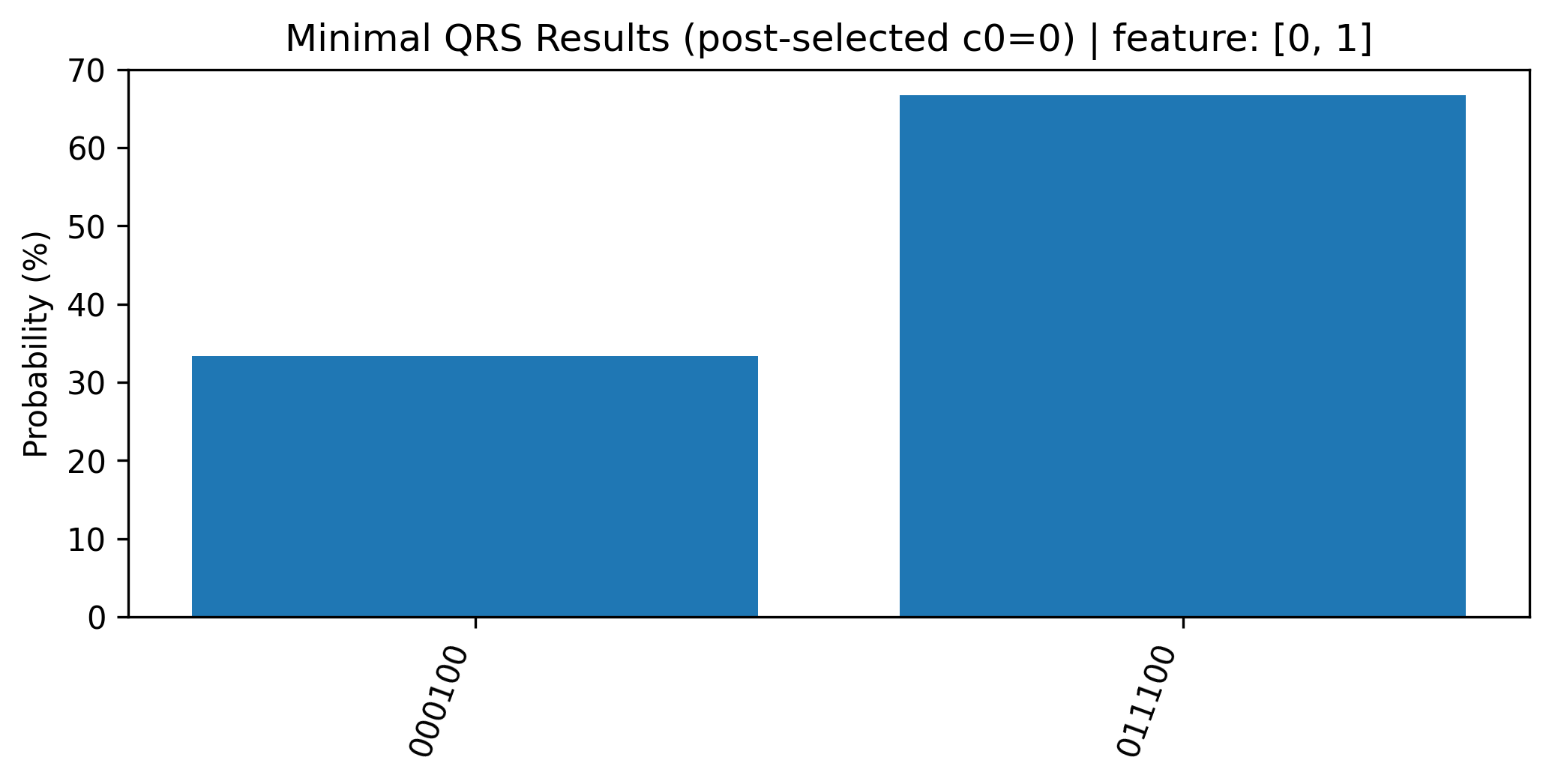}\\[-0.2em]
        \small (b)
    \end{minipage}

    \vskip0.8em

    \begin{minipage}[t]{0.48\textwidth}
        \centering
        \includegraphics[width=\linewidth]{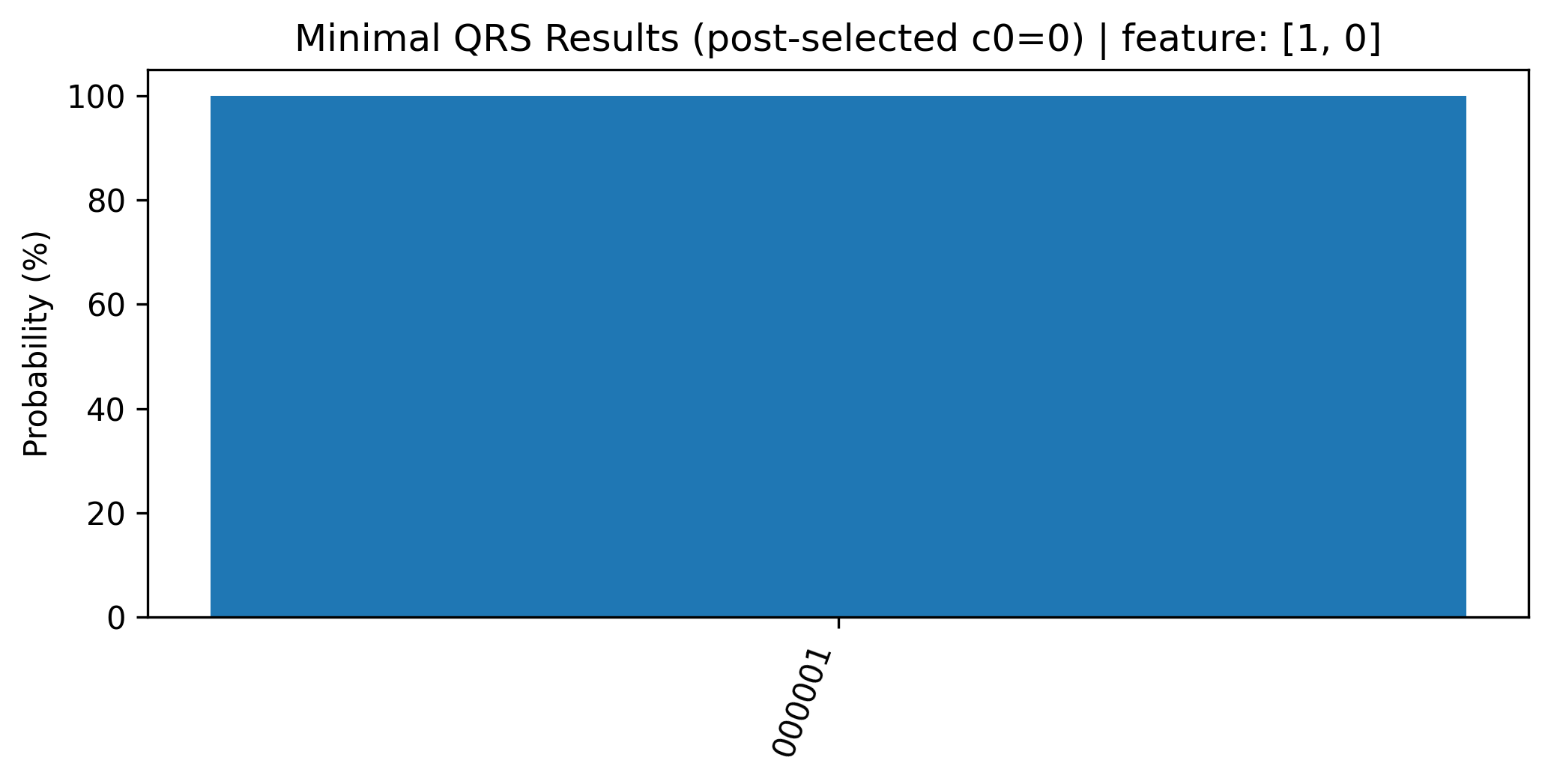}\\[-0.2em]
        \small (c)
    \end{minipage}
    \hfill
    \begin{minipage}[t]{0.48\textwidth}
        \centering
        \includegraphics[width=\linewidth]{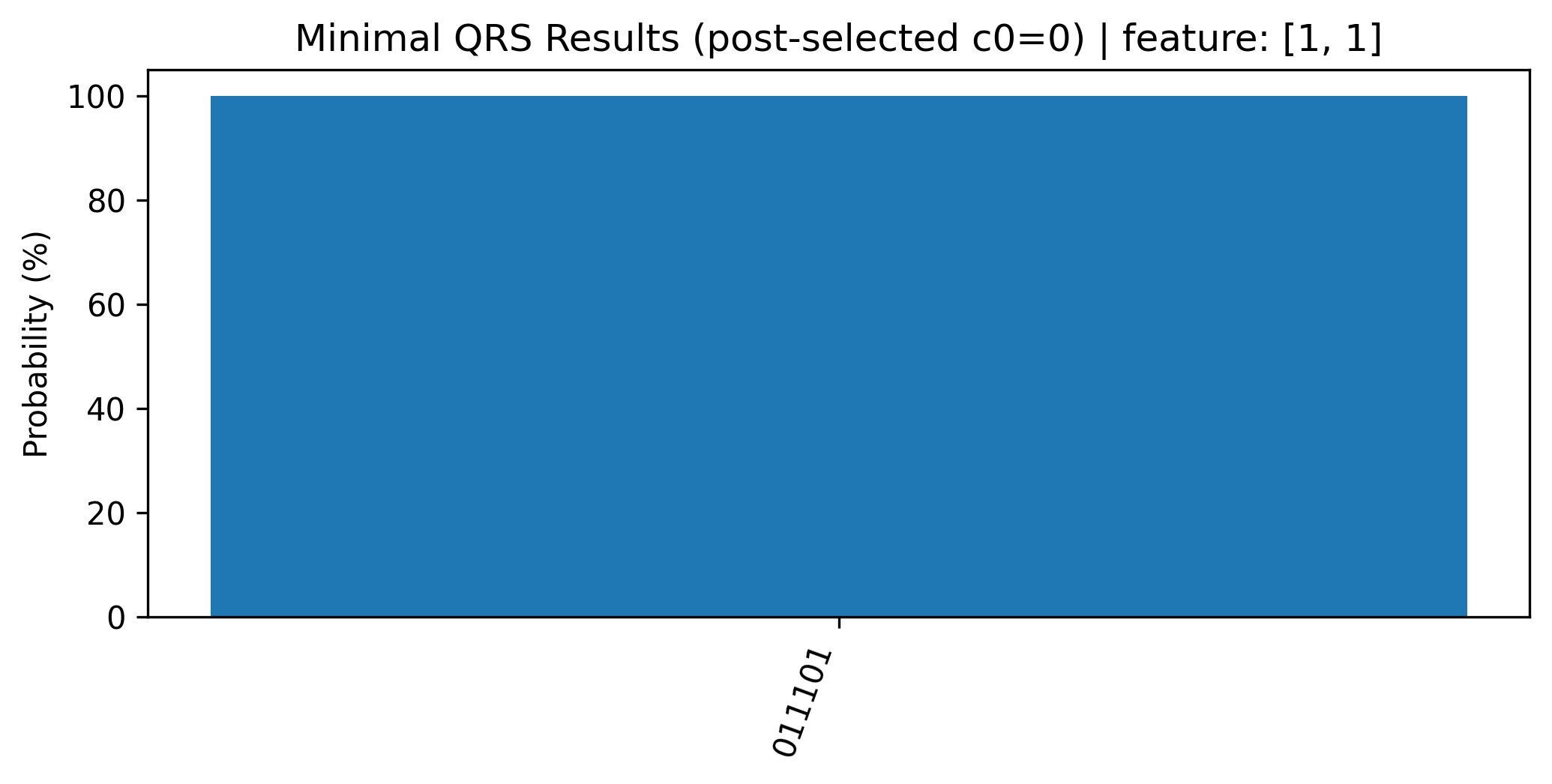}\\[-0.2em]
        \small (d)
    \end{minipage}

    \caption{
    Output counts for four feature-vector recommendations based on the same database encodings, \(00\) and \(01\), optimized with SABRE layout and routing.
    \textbf{(a)} Feature vector \(\langle 0,0\rangle\).
    \textbf{(b)} Feature vector \(\langle 0,1\rangle\).
    \textbf{(c)} Feature vector \(\langle 1,0\rangle\).
    \textbf{(d)} Feature vector \(\langle 1,1\rangle\).
    }
    \label{fig:four_plots_from_statevector_gx}
\end{figure*}

\section{Scope of classical--quantum blind complexity comparisons}
\label{app:comparison_scope}

A central question of this work is under what assumptions and at
what cost blind quantum learning can offer an advantage over
classical privacy-preserving approaches: is it possible to simultaneously achieve a quantum advantage in privacy and speed? To address this question,
we distinguish the additional cost of making an existing computation private from the
cost of implementing a complete private computation. We then consider the
inference-only B-QRS as a case study, as the blinded classical analog is well-studied~\cite{PriRec}.

\paragraph{Criterion for a dual quantum advantage.}
Let $Q(L)$ and $C(L)$ denote the costs of specified unblinded
quantum and classical algorithms, respectively, for a problem
of size $L$. For a fixed resource measure, define the privacy
overhead factors $h_Q(L)$ and $h_C(L)$ through
\begin{equation}
    B_Q(L)=h_Q(L)Q(L),
    \qquad
    B_C(L)=h_C(L)C(L),
\end{equation}
where $B_Q$ and $B_C$ are the costs of the corresponding private
implementations. These comparisons assume the same task and
output guarantee, together with compatible data-access and
privacy assumptions.

Two cases are worth distinguishing. First, the quantum
implementation has a strictly smaller asymptotic privacy
overhead when
\begin{equation}
    \lim_{L\to\infty}\frac{h_Q(L)}{h_C(L)}=0.
\end{equation}
Second, the complete private quantum implementation has
strictly better asymptotic scaling when
\begin{equation}\label{dualquantumadvantage}
    \lim_{L\to\infty}\frac{B_Q(L)}{B_C(L)}
    =
    \lim_{L\to\infty}
    \frac{h_Q(L)Q(L)}{h_C(L)C(L)}
    =0.
\end{equation}
The first case does not imply the second, since a smaller
privacy overhead may accompany a more expensive underlying
algorithm; a quantum speed quantum advantage is not guaranteed for all possible problems. Conversely, an algorithmic advantage may survive
even if the quantum privacy overhead is larger due to an underlying quantum speed advantage. Relations
expressed only as $\mathcal{O}$ bounds allow equal asymptotic
scaling and therefore do not, on their own, establish either
strict improvement. Emphatically, for a blind quantum algorithm to provably satisfy Eq. \ref{dualquantumadvantage} implies a dual quantum advantage in both privacy due to the information-theoretic security of UBQC and in speed due to the improved asymptotic scaling.  

\paragraph{Resource measures.}
To apply this criterion, the costs being compared must first
be specified. Let $n$ denote the gate count over the gate
families covered by our brickwork-depth bound, and let $\omega$
denote the logical circuit width. Section~\ref{sec:upp_bound_bw}
establishes
\begin{equation}
    D(B)=\mathcal{O}(\omega n).
\end{equation}
This bounds brickwork depth. It does not identify $\omega$ as
the actual multiplicative overhead in every resource measure.

For a full $\omega$-row brickwork graph, each brick column
expands into four measured graph columns, with one final
output column appended. The resource-state size is therefore
\begin{equation}
    S(B)=\omega[4D(B)+1]
        =\mathcal{O}(\omega^2 n+\omega).
\end{equation}
There are $4\omega D(B)$ measured qubits, and the bounded graph
degree gives an $\mathcal{O}(S(B))$ bound on the entangling
operations. Streaming preparation can reduce the number of
qubits stored simultaneously without removing these operations.

With sufficient hardware, operations across the width can be
performed in parallel once their dependencies are satisfied.
Assuming bounded-duration local operations, bounded-time
classical feed-forward per measurement layer, and sufficient
parallel capacity for resource preparation and delivery,
execution time is upper bounded by $\mathcal{O}(D(B))$.
Total quantum work instead scales as
$\mathcal{O}(\omega D(B)+\omega)$. Thus, parallel runtime and
total work give different comparisons. A runtime comparison
must also specify the parallel resources available to the
classical implementation. These statements concern logical
operations and exclude fault-tolerance costs.

\paragraph{Case 1: comparing privacy overheads.}
Classical secure computation offers several ways of evaluating
algorithms while protecting private inputs. Their overheads
depend on the protocol, security parameters, and trust
assumptions~\cite{SMPC_Theory_apps_overview}. Consequently,
there is no single classical overhead factor against which
the width-dependent brickwork bound can be compared.

Fully homomorphic encryption provides a relevant example,
since it allows a server to compute on encrypted client
data~\cite{FHE_overhead}. A comparison with a specified FHE
implementation should include its evaluation, encryption,
and decryption costs, together with the security parameter
and numerical precision. Standard FHE does not automatically
hide the evaluated function or circuit structure, so the
permitted leakage must also be specified.

UBQC protects the delegated inputs, outputs, and computation
under its stated assumptions, while allowing disclosure of
the chosen resource dimensions~\cite{ubqc}. Its security
setting is therefore not interchangeable with that of every
classical private-learning protocol. Similarly, concealing
query-dependent memory accesses does not require forbidding
them: oblivious-access methods may be relevant when the task
admits suitable indexing~\cite{ORAM_Barenghi2019}. Blindness
alone does not establish that a classical algorithm must
scan the complete database.

These differences make comparisons of fully specified
implementations more informative than comparisons of
overhead factors in isolation.

\paragraph{Case 2: comparing complete private computations.}
An interesting case is the B-QRS without database encoding
considered in Appendix~\ref{sec:qrs_no_db}. Under the assumptions
and gate accounting given there, constant feature width gives
$\omega=\mathcal{O}(\log L)$ and
\begin{equation}
    D_{\mathrm{inf}}(L)=\mathcal{O}(\sqrt{L}\log L).
\end{equation}
For the full brickwork graph, this implies
\begin{equation}
    S_{\mathrm{inf}}(L)
      =\mathcal{O}\!\left(\sqrt{L}(\log L)^2\right).
\end{equation}
The corresponding logical quantum work is therefore
asymptotically sublinear in $L$, subject to the stated
inference accounting.

A classical baseline that explicitly scores all $L$ candidates
has linear dependence on $L$ for fixed per-candidate cost.
For example, exhaustive dense factorization-machine scoring
takes $\Theta(LKd)$ operations for rank $K$ and feature
dimension $d$. PriRec provides a practical privacy-preserving
example using factorization-machine scoring and a
candidate-matching stage~\cite{PriRec}. However, this is not
automatically the same task as the distance-based recommendation
routine underlying B-QRS~\cite{Sawerwain_Wróblewski_2019}.
The scoring rule, candidate selection, output size, accuracy,
and privacy guarantees must be matched before comparing them.

If such a comparison establishes total quantum inference work
$B_Q(L)=\mathcal{O}(\sqrt{L}(\log L)^2)$ and classical work
$B_C(L)=\Omega(L)$, then
\begin{equation}
    0\leq\frac{B_Q(L)}{B_C(L)}
    \leq c\,\frac{(\log L)^2}{\sqrt{L}}
    \longrightarrow 0
\end{equation}
for some constant $c>0$ and sufficiently large $L$.
This would establish an asymptotic inference-work advantage
over the specified classical baseline. Establishing an advantage
against a class of classical algorithms requires the lower
bound to hold throughout that class; the cost of an exhaustive
implementation alone does not provide such a lower bound.

\paragraph{Database preparation and the scope of inference.}
The distinction between inference and preparation is essential
to this comparison. UBQC with quantum inputs permits delegation
of computations on supplied quantum states~\cite{ubqc}.
The inference complexity derived in
Appendix~\ref{sec:qrs_no_db} includes the amplification stage.
Amplification is therefore already accounted for in the stated
bound and is not an additional overhead introduced by the
present comparison. Initial database preparation remains
excluded from the inference-only analysis.

The inference-only analysis assumes a supplied encoded database
state. Initial preparation is excluded, while all database-access
operations required during inference must be included. No
assumption is made that the supplied state can be reused for
subsequent recommendations. A comparison with classical inference
must likewise specify which input preparation and preprocessing
costs are excluded.
In addition, delegating a supplied database state does not
by itself establish database privacy against the client;
no-cloning alone is insufficient to prove mutual privacy.

The comparison therefore depends on where the cost of
preparing and accessing the data resides. Our results provide
bounds on specified blind quantum implementations and identify
conditions under which favorable inference scaling can survive
transpilation. Whether these conditions yield an advantage
over a classical private implementation is then determined
by the matched costs and the criterion above.

\section{Analysis of Blind Quantum Machine Learning Algorithms}\label{app:qml_analysis}

\subsection{B-HHL}\label{sec:res_hhl}

\newcommand{\bigO}{\mathcal{O}}
\newcommand{\polylog}{\operatorname{polylog}}

Since HHL encodes the solution of a linear system of equations into a quantum register, it can use the extension of the BFK protocol supporting quantum outputs \cite{ubqc}, returning the solution state rather than a measurement sample. To benchmark blind HHL, we considered several matrix families, given the parameters entering its computational complexity:
\[
\mathrm{poly}\!\bigl(\kappa,\, s,\, 1/\varepsilon\bigr)
\cdot \mathrm{polylog}(N).
\]
Here, \(N\) is the dimension of the \(N\times N\) system matrix, \(\kappa\) its condition number, \(s\) the maximum number of nonzero entries per row, and \(\varepsilon\) the target error in the prepared quantum solution state. Our interest here is how the blind HHL brickwork depth scales with matrix dimension \(N\) across different matrix families, rather than an in-depth performance analysis. To that end, we used tridiagonal Toeplitz matrices (inspired by \cite{HHL_Toeplitz}), normalized Laplacians (shifted to control \(\kappa\)), and sparse Hermitian matrices with controlled spectral properties.

\paragraph{Tridiagonal Toeplitz matrix}
Let $A_N(a,b)\in\mathbb{R}^{N\times N}$ be the tridiagonal Toeplitz matrix
\begin{equation}\label{eq:toeplitz_symbolic}
A_N(a,b)=
\begin{pmatrix}
a & b \\
b & a & b \\
  & \ddots & \ddots & \ddots \\
  &        & b & a & b \\
  &        &   & b & a
\end{pmatrix}, \quad s=3.
\end{equation}
In the benchmark case we set $a=1$ and $b=-\tfrac16$, i.e.\ $A_N:=A_N\!\left(1,-\tfrac16\right)$.

% \begin{equation}\label{eq:toeplitz_numeric}
% a=1,\qquad b=-\tfrac{1}{6},\qquad
% A_N \;=\;
% \begin{pmatrix}
%  1 & -\tfrac{1}{6} & 0 & 0 & \cdots & 0 \\
%  -\tfrac{1}{6} & 1 & -\tfrac{1}{6} & 0 & \ddots & \vdots \\
%  0 & -\tfrac{1}{6} & 1 & -\tfrac{1}{6} & \ddots & 0 \\
%  0 & 0 & -\tfrac{1}{6} & 1 & \ddots & 0 \\
%  \vdots & \ddots & \ddots & \ddots & \ddots & -\tfrac{1}{6} \\
%  0 & \cdots & 0 & 0 & -\tfrac{1}{6} & 1
% \end{pmatrix}.
% \end{equation}
From \cite{toeplitz_ev}, the eigenvalues of a tridiagonal Toeplitz matrix where the diagonal is $a$ and the off diagonal is $b$, \(A_N(a,b) \; = \; aI + b\tilde{A}\) are
\begin{equation}\label{eq:toeplitz_eigs}
\lambda_j \;=\; a + 2b \cos\!\Bigl(\tfrac{j\pi}{N+1}\Bigr),\qquad j=1,\dots,N.
\end{equation}
With $A_N\!\left(1,-\tfrac16\right)$, one has \(\lambda_{\max}\to \tfrac{4}{3}\) and \(\lambda_{\min}\to \tfrac{2}{3}\) as \(N\to\infty\), so
\begin{equation}
\kappa(A_N)\;=\;\frac{\lambda_{\max}}{\lambda_{\min}} \; = \; \frac{\frac{4}{3}}{\frac{2}{3}} \; = \; 2,
\end{equation}
i.e., the condition number remains constant across system sizes.

Let \(n=\lceil \log_2 N\rceil\). The circuit width decomposes as
% \begin{equation}
% \omega \;=\; \underbrace{n}_{\text{system}} \;+\; \underbrace{\Theta\!\bigl(\log(\kappa/\varepsilon)\bigr)}_{\text{phase register}} \;+\; \underbrace{\mathcal{O}(1)}_{\text{ancillae}},
% \label{eq:toeplitz_widthplus}
% \end{equation}
% which simplifies to \(\omega=\Theta(\log N)\) as \(\kappa=\Theta(1)\) and \(\varepsilon\) is fixed. 

\begin{equation}
  \omega \;=\; \underbrace{n}_{\text{system}}
  \;+\; \underbrace{\Theta\!\big(\log(\kappa/\varepsilon)\big)}_{\text{phase register}}
  \;+\; \underbrace{\mathcal{O}(n)}_{\text{ancillae}}
  \;+\; \underbrace{1}_{\text{flag}},
  \label{eq:hhl_width}
\end{equation}
where the ancilla term is implementation dependent: Qiskit's
\texttt{TridiagonalToeplitz} Hamiltonian simulation uses $n-1$ ancillae, while an
explicitly given matrix uses none. In either case $\omega=\Theta(\log N)$ as
$\kappa=\Theta(1)$ and $\varepsilon$ is fixed. For Fig.~\ref{fig:res_qft}b we
therefore keep the phase register at five qubits for every $N$, in line with
Eq.~\eqref{eq:hhl_width}, using Qiskit's HHL with exact reciprocal,
$\varepsilon = 10^{-2}$ and $|b\rangle = H^{\otimes n}|0\rangle$; the resulting
width is $\omega = 2n+5$. For $N \le 64$, where the circuit can be simulated, the
post-selected output has fidelity at least $0.998$ with the normalized solution
$A_N^{-1}|b\rangle$. The gate count $G$ grows as $G \approx 1.4\times 10^{3}\, n^2$
up to $N = 2^{12}$, so the circuit is polylogarithmic in $N$, and
Theorem~\ref{thm:bw_general_alignment} gives
$D(\mathcal{B}) = \mathcal{O}(\omega G) = \mathcal{O}(\log^3 N)$, consistent with
the measured depths.

\paragraph{Normalized Laplacian matrix}
A closely related family to the tridiagonal Toeplitz matrices, but not restricted to 1D banded couplings, is the normalized Laplacian
\(\mathcal{L} = I - D^{-1/2} A D^{-1/2}\) of a graph (with adjacency matrix \(A\) and degree matrix \(D\)). It is symmetric positive semidefinite and its eigenvalues lie in the interval \([0,2]\) for any simple connected graph; in particular, \(0=\lambda_1 \leq \lambda_2 \leq \dotsi\leq \lambda_n \leq 2\) on connected graphs \cite{ev_normalised_laplacian, Chung:1997_spectralgraphtheory}. To control the condition number for HHL, we add a positive multiple of the identity to the normalized Laplacian:
\begin{equation}
\mathcal{L}_\mu = \mathcal{L} + \mu I,\qquad \mu>0.
\end{equation}
This shifts every eigenvalue by \(\mu\), making \(\mathcal{L}_\mu\) positive definite with eigenvalues in \([\mu,2+\mu]\). Its condition number therefore satisfies
\begin{equation}
\kappa(\mathcal{L}_\mu)
= \frac{\lambda_{\max}(\mathcal{L})+\mu}{\mu}
\leq \frac{2+\mu}{\mu}.
\end{equation}
For example, choosing \(\mu=1\) gives \(\kappa\leq3\). Keeping \(\mu\) fixed thus bounds the condition number independently of the matrix dimension.
% To control the condition number for HHL, we shift by a positive multiple of the identity:
% \begin{equation}
% \mathcal{L}_\mu \;=\; \mathcal{L} + \mu I,\qquad \mu>0,
% \end{equation}
% so that \(\sigma(\mathcal{L}_\mu)\subseteq[\mu,\,2+\mu]\) and \(\kappa(\mathcal{L}_\mu)=\frac{2+\mu}{\mu}\) (e.g.\ \(\mu=1\Rightarrow \kappa=3\)), preventing \(\kappa\) from interfering with the algorithmic complexity analysis.

The sparsity of \(\mathcal{L}\) is governed by the graph degree. For graphs of maximum degree \(\Delta\), each row has at most \(\Delta+1\) nonzero values (neighbors plus the diagonal), so \(\mathcal{L}\) is \(s\)-sparse with \(s\le \Delta+1=\mathcal{O}(1)\). In the special case \(\Delta=2\) (paths/cycles), each row has at most three nonzero values, matching the Toeplitz tri\-diagonal sparsity; for general bounded-degree graphs the nonzero values are not confined to a tri\-diagonal band but are dispersed according to the neighbor indices, serving as an excellent comparison to the transpiler's behavior given a tridiagonal Toeplitz matrix.

Finally, with \(n=\lceil \log_2 N\rceil\) system qubits, a constant number of ancillae, and a constant-sized phase register when \(\kappa\) is bounded and \(\varepsilon\) is fixed, the circuit width scales as
\begin{equation}
w \;=\; n + \mathcal{O}(1) \;=\; \Theta(\log N).
\end{equation}

\paragraph{Sparse Hermitian matrix}
Let \(A \in \mathbb{R}^{N\times N}\) be a sparse matrix with density \(0.05\). We form
\begin{equation}
\tilde{A} \;=\; \frac{A + A^{\top}}{2},
\end{equation}
which satisfies \(\tilde{A}^{\top}=\tilde{A}\) and is therefore symmetric; for real-valued matrices, symmetry implies Hermitianity (\(A^\dagger=A^{\top}\)). For numerical stability, we add a diagonal shift while preserving symmetric and Hermitian structure:
\begin{equation}
\tilde{A}_{\mu} \;=\; \tilde{A} + \mu I,
\end{equation}
with \(\mu=1\) (e.g., chosen so that \(\lambda_{\min}(\tilde{A}_\mu)>0\)). Two independent samples of this low-density Hermitian matrix are transpiled without circuit optimization. %In practice, one may also rescale so that \(\|\tilde{A}_\mu\|_2 \le 1\) before applying HHL. 
The results of all matrices are plotted in Fig. \ref{fig:res_hhl}. Fig.~\ref{fig:res_hhl} uses Qiskit's default phase register of $n+2$ qubits, which grows with $N$. For the Toeplitz matrix, the controlled power $U^{2^j}$ then uses $\lceil 2^{j/2} k \rceil$ Trotter steps, with $k$ the base number of steps, and the exact reciprocal acts on all $2^{n+1}$ register values, so the circuit carries a $\sqrt{N}$ factor and, for large $N$, a linear term; this is why panel (a) grows faster than Fig.~\ref{fig:res_qft}b. For the Laplacian and sparse Hermitian matrices, the matrix is compiled as a dense unitary and repeated $2^j$ times, so the circuit grows roughly as $N^3$. These costs come from the circuit, not from blinding: in every family the brickwork depth follows the circuit and stays within the bound of Theorem~\ref{thm:bw_general_alignment}.

\begin{figure}[ht!] 
\centering
\includegraphics[width=\linewidth]{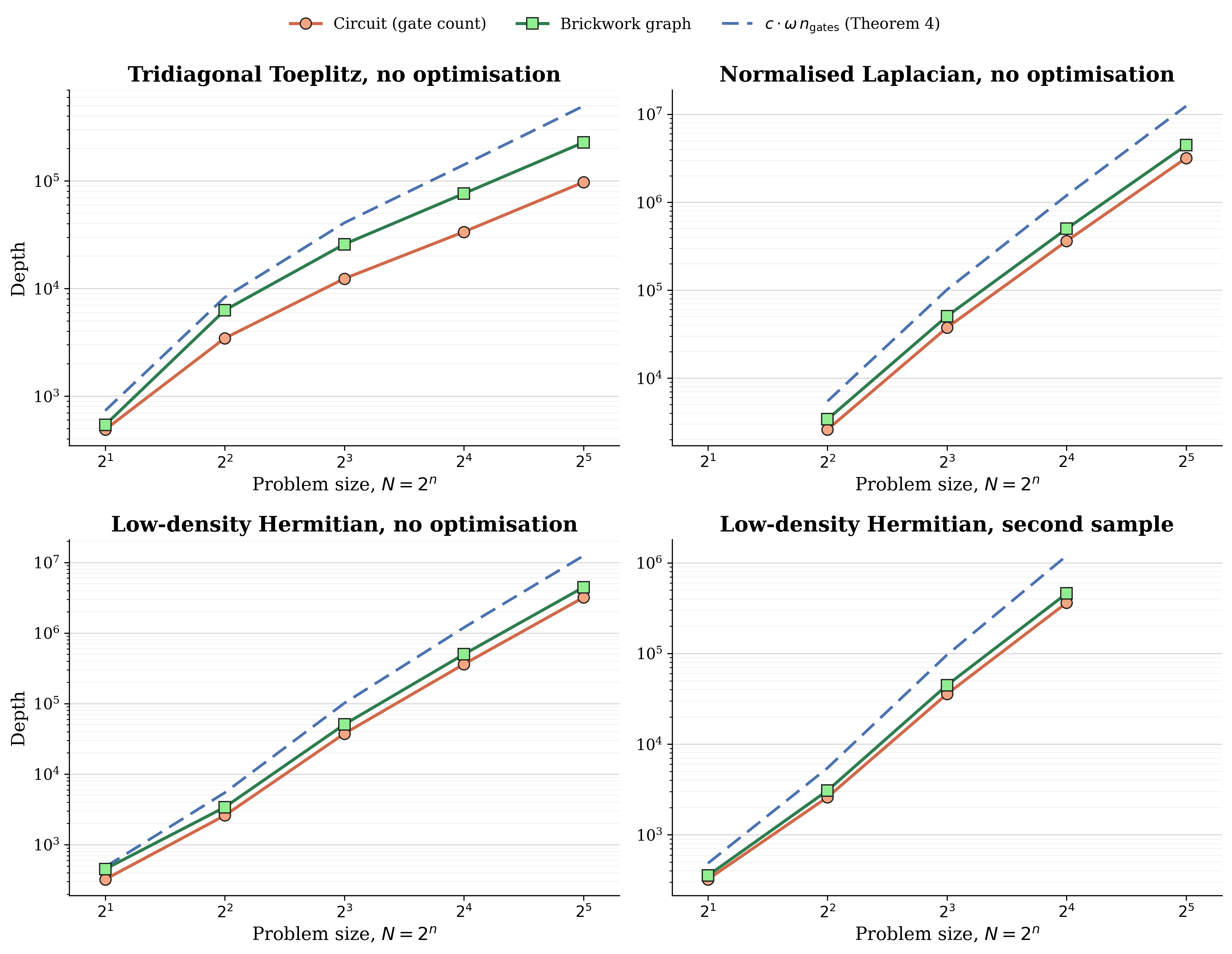}
% \caption{Resulting brickwork graphs of HHL given the circuit (a) Tridiagonal Toeplitz, (b) Normalized Laplacian, (c) sparse Hermitian matrix ($d=0.05$), and (d) sparse Hermitian matrix ($d=0.05$) with different optimization parameters.}
\caption{Brickwork depth $D(\mathcal{B})$ (green), circuit gate count $G$ (orange)
and the Theorem~\ref{thm:bw_general_alignment} reference $c\,\omega G$ (dashed) for HHL with Qiskit's
default phase register of $n+2$ qubits, as a function of the matrix dimension $N=2^n$:
(a) tridiagonal Toeplitz $A_N(1,-\tfrac16)$, (b) normalized Laplacian
$\mathcal{L}+I$ of a random $2$-regular graph, (c) and (d) two independently drawn
sparse Hermitian matrices (density $0.05$) $\tilde A + I$. All circuits use the exact
reciprocal and $\varepsilon=10^{-2}$; no circuit optimization is applied. Because
Qiskit compiles an explicit matrix as a dense unitary, the gate count is essentially
independent of the matrix entries: (b), (c) and (d) coincide except at $N=8$, where
the sample in (d) happens to synthesize to fewer gates. The brickwork depth stays
below the Theorem~\ref{thm:bw_general_alignment} reference in every panel (constant $c=0.3$ chosen
empirically).}
\label{fig:res_hhl}
\end{figure}

\subsection{B-QRS} \label{app:bqrs_with_db}

All gates used in the quantum recommendation system from \cite{Sawerwain_Wróblewski_2019} fall into the four gate families assumed in the derivation of Theorem \ref{thm:bw_general_alignment}. Hence, by Theorem~\ref{thm:bw_general_alignment}, the brickwork depth obeys
\begin{equation}
D(\mathcal{B}) \;=\; \mathcal{O}(\omega n),
\end{equation}
where \(\omega\) is the original circuit width in logical qubits and \(n\) its gate complexity before decomposition, with each multi-controlled gate counted once. We reserve one additional clean ancilla and write
\(W=\omega+1\) for the implemented width. Throughout
this section, \(a_{\mathrm{MCU}}\) denotes the constant
defined in Theorem~\ref{thm:bw_general_alignment}.
An MCU therefore contributes at most
\(a_{\mathrm{MCU}}W\) brick columns before alignment.
The global alignment factor is included explicitly
in the bounds below.

Under a binary encoding of database identifiers, a database with \(L=2^q\) entries requires \(q\) identifier qubits. With constant feature width \(l\) and a fixed number of auxiliary qubits, the circuit width therefore satisfies
\begin{equation}
\omega \;=\; \mathcal{O}(\log L).
\end{equation}
In the absence of Grover amplification, the gate complexity is dominated by database encoding. Following Section~3.3 of \cite{Sawerwain_Wróblewski_2019}, the initialization cost satisfies
\begin{equation}
O_1(l,N) \;=\; 2l \;+\; \frac{N(N-1)}{2},
\end{equation}
with $N=2^q$ the number of database entries\footnote{Our notation identifies $N$ in \cite{Sawerwain_Wróblewski_2019} with $L$ used here. The $2l$ term accounts for feature/user preparation; the quadratic term is the controlled permutation implementing data loading \cite{Li_DECOMPOSITION_OF_UNITARY_MATRICES_AND_QUANTUM_GATES}.}.
With \(l\) constant and \(N=L\), this gives the
upper bound \(n=\mathcal{O}(L^2)\), and therefore
\begin{equation}
D(\mathcal{B}) \;=\; \mathcal{O}\!\big(L^2 \log L\big).
\end{equation}

\subsubsection{Bounding B-QRS With Grover Amplification}
Even though all gates used in the recommendation system are covered by the assumptions leading to Theorem \ref{thm:bw_general_alignment}, the part of the algorithm responsible for amplitude amplification introduces additional overhead. Let \(r\) denote the integer number of amplification
rounds. Under the amplification assumptions used here,
\begin{equation}
    r=\mathcal{O}\!\left(\sqrt{\frac{L}{g}}\right),
\end{equation}
where \(g\geq1\) denotes the number of ``good" elements subject to amplification.

Retaining the circuit accounting of at most \(6\omega\)
native single-qubit rotations and two MCUs per
amplification round, the pre-alignment cost of \(r\)
rounds is bounded by
\[
    r\bigl(6\omega+2a_{\mathrm{MCU}}W\bigr).
\]
Applying the alignment factor \(2\), the amplification
stage admits a brickwork implementation satisfying
\begin{equation}
    D_{\mathrm{amp}}
    \leq
    2r\bigl(6\omega+2a_{\mathrm{MCU}}W\bigr)
    =
    \bigl(12\omega+4a_{\mathrm{MCU}}W\bigr)r.
    % Retain the existing label of F38 here.
\end{equation}

Since \(W=\omega+1=\mathcal{O}(\log L)\) for constant
feature width, combining this contribution with the
database-encoding and non-amplification costs gives
\begin{equation}
    D(\mathcal{B})
    =
    \mathcal{O}\!\left(
       L^2\log L
       +\sqrt{\frac{L}{g}}\log L
    \right).
    % Retain the existing label of F39 here.
\end{equation}
Since \(g\geq1\), this yields
\begin{equation}
    D(\mathcal{B})
    =
    \mathcal{O}\!\left(
       L^2\log L+\sqrt{L}\log L
    \right),
    % Retain the existing label of F40 here.
\end{equation}
and hence
\begin{equation}
    D(\mathcal{B})=\mathcal{O}(L^2\log L).
    % Retain the existing label of F41 here.
\end{equation}

Experimental results supporting this analysis are presented in Fig.~\ref{fig:res_qrs_full_plots}.

%PRE MCU
% Specifically, the overhead induced by the amplification step is determined by the optimal number of amplification rounds 
% \[
% \frac{\pi}{4}\sqrt{\frac{L}{g}},
% \]
% where $g$ denotes the number of ``good'' elements subject to amplification. Consequently, the total depth includes this contribution in terms of gates used to implement Grover's algorithm.
% \begin{equation}\label{eq:grover_explicit}
% 6\omega + 2 (32(6\omega+2)) \cdot \frac{\pi}{4}\sqrt{\frac{L}{g}},
% \end{equation}
% leading to the overall asymptotic complexity
% \begin{equation}
% \mathcal{O}\!\left(L^2 \log L + \sqrt{\tfrac{L}{g}} \log L\right).
% \end{equation}

% Since $g \geq 1$, this expression can be further simplified to
% \begin{equation}\label{eq:big_qrs_L_eq}
% \mathcal{O}\!\left(L^2 \log L + \sqrt{L}\,\log L\right).
% \end{equation}

% Ultimately, the circuit depth is asymptotically bounded by
% \begin{equation}\label{eq:qrs_with_db}
%    D(\mathcal{B}) = \mathcal{O}(L^2 \log L).
% \end{equation}

% Experimental results supporting this analysis are presented in Fig.~\ref{fig:res_qrs_full_plots}.

\begin{figure}[ht!]
\centering
\includegraphics[width=\linewidth]{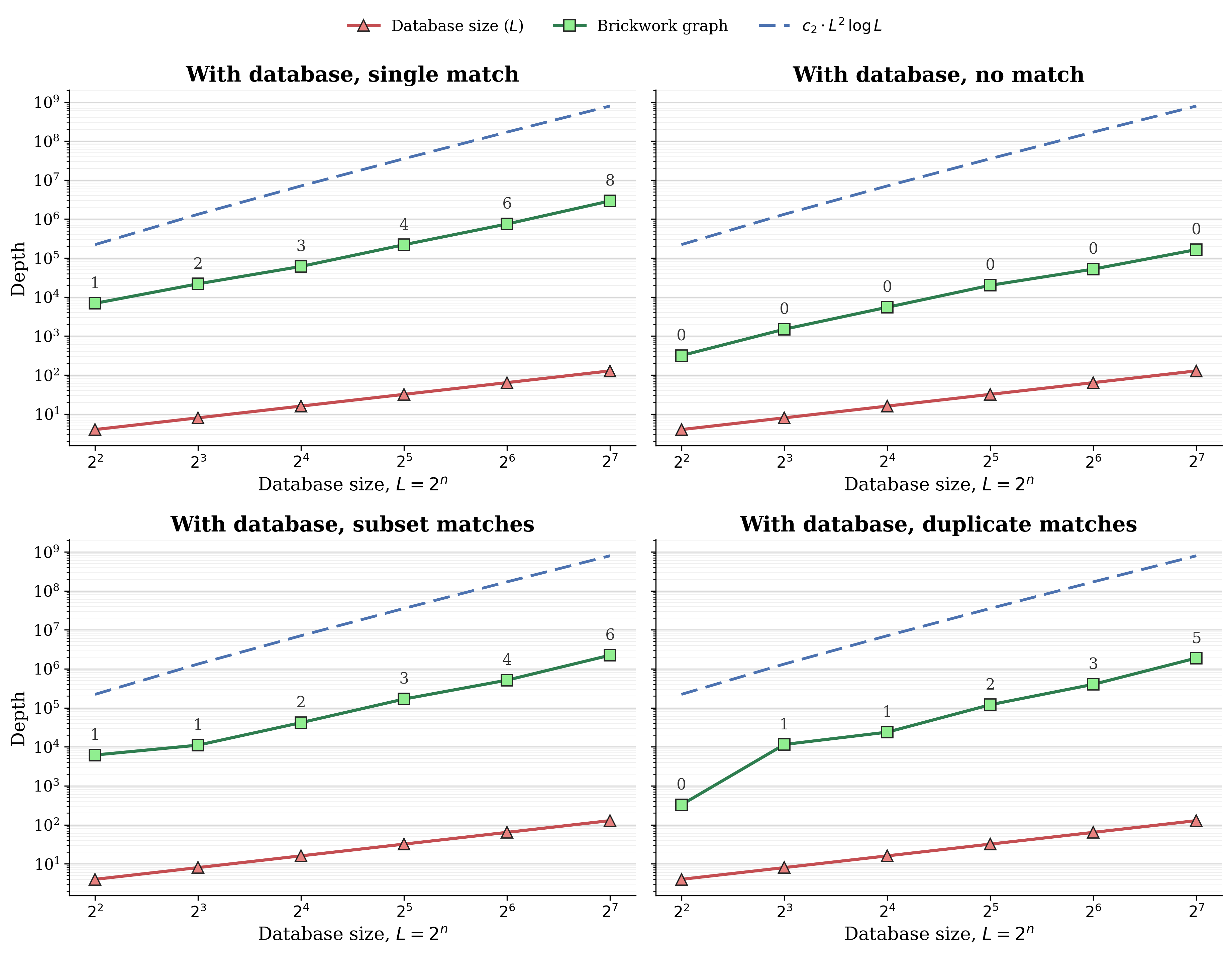}
\caption{Scaling behavior of brickwork graph depths obtained after transpiling the quantum recommendation system circuit, including the database encoding stage. Each point represents the resulting brickwork graph depth for a given input size, illustrating how transpilation cost grows with system scale, corroborating $\mathcal{O}(L^2 \log L)$.}
\label{fig:res_qrs_full_plots}
\end{figure}

\subsubsection{Bounding B-QRS Complexity Without Database Encoding}\label{sec:qrs_no_db}

A not unreasonable scenario is that an organization provides a dataset of size $L$, with feature vectors $f_f$, already encoded in a Hilbert space spanned by $\log L + l$ qubits. In this setting, the user can get a recommendation directly by applying QKNN and Grover amplification to the prepared database qubits. This division of tasks is particularly relevant given the dominating cost of database encoding and the support for quantum inputs in the generalized BFK protocol\cite{ubqc} and is therefore further explored in this section.

It is important to note that the number of gates involved in QKNN depends as much on the \emph{width} of the circuit as its depth. To see this, consider the QKNN construction presented in~\cite{Sawerwain_Wróblewski_2019}. In this design, each additional qubit contributes one additional $\mathrm{CX}$ gate, i.e., the gate count scales linearly with the circuit width. Moreover, the number of subcircuits used remains constant across problem sizes. These subcircuits perform three operations: computing the Hamming distance, summing the distances, and decoding the distances\footnote{This extension is not explicitly described in~\cite{Sawerwain_Wróblewski_2019}, but is required in our construction to ensure that amplitude amplification is applied directly to the database entries rather than on phase-encoded distances.}. 

%PRE MCU$
% The resulting gate cost for QKNN therefore scales as
% \begin{equation}
% 3\omega \cdot 32(6\omega+2) + \omega,
% \end{equation}
% where $\omega$ denotes the circuit's width. As established in Eq.~\eqref{eq:grover_explicit}, the cost of transpiling Grover's algorithm amounts to
% \begin{equation}
% 6\omega + 2 (32(6\omega+2)) \cdot \frac{\pi}{4}\sqrt{\frac{L}{g}},
% % \left(2 \cdot 32(6\omega+2) + 6\omega\right) \cdot \frac{\pi}{4}\sqrt{\frac{L}{g}},
% \end{equation}
% where $\sqrt{\tfrac{L}{g}}$ corresponds to the optimal number of Grover iterations. Combining the two contributions and simplifying leads to
% \begin{equation}
% 576\omega^2 + 193\omega + \left(390\omega + 128\right)\frac{\pi}{4}\sqrt{\tfrac{L}{g}}.
% \end{equation}
% Following the same simplification steps that lead to Eq.~\eqref{eq:big_qrs_L_eq}, this expression reduces asymptotically to
% \begin{equation}
% \mathcal{O}\!\left(\omega^2 + \omega\sqrt{L}\right).
% \end{equation}

% Assuming a constant feature size, the circuit width $m$ grows logarithmically with $L$, i.e., $m \propto \log L$, as each doubling of $L$ increases $m$ by one. Substituting this relation yields
% \begin{equation}
% \mathcal{O}\!\left(\log^2 L + \sqrt{L}\,\log L\right).
% \end{equation}

Retaining the gate accounting represented by the three
QKNN subcircuits, we conservatively assign at most
\(3\omega\) operations the MCU cost bound and count
a further \(\omega\) native rotations. This also bounds
CX operations, since
\(a_{\mathrm{MCU}}W\geq2(6W+1)\).
Including the alignment factor, the QKNN stage admits
a brickwork implementation satisfying
\begin{equation}
    D_{\mathrm{QKNN}}
    \leq
    2\bigl(3\omega a_{\mathrm{MCU}}W+\omega\bigr).
    % Retain the existing label of F42 here.
\end{equation}

For \(r\) amplification rounds, the corresponding bound is
\begin{equation}
    D_{\mathrm{amp}}
    \leq
    2r\bigl(6\omega+2a_{\mathrm{MCU}}W\bigr),
    % Retain the existing label of F43 here.
\end{equation}
where \(r=\mathcal{O}(\sqrt{L/g})\).

The total pre-alignment cost is the sum of these two
stages before their alignment factors are applied.
Applying the factor \(2\) once to that sum gives
\begin{equation}
\begin{aligned}
    D(\mathcal{B})
    &\leq
    2\Bigl(
       3\omega a_{\mathrm{MCU}}W+\omega
       +r(6\omega+2a_{\mathrm{MCU}}W)
    \Bigr)
    \\
    &=
    6a_{\mathrm{MCU}}\omega^2
    +(6a_{\mathrm{MCU}}+2)\omega
    \\
    &\quad+
    \Bigl((12+4a_{\mathrm{MCU}})\omega
          +4a_{\mathrm{MCU}}\Bigr)r,
\end{aligned}
    % Retain the existing label of F44 here.
\end{equation}
where \(W=\omega+1\).

Since \(a_{\mathrm{MCU}}\) is independent of the
problem size and \(g\geq1\), this gives
\begin{equation}
    D(\mathcal{B})
    =
    \mathcal{O}\!\left(
       \omega^2+\omega\sqrt{\frac{L}{g}}
    \right)
    =
    \mathcal{O}\!\left(\omega^2+\omega\sqrt{L}\right).
    % Retain the existing label of F45 here.
\end{equation}

For constant feature width,
\(\omega=\mathcal{O}(\log L)\), and therefore
\begin{equation}
    D(\mathcal{B})
    =
    \mathcal{O}\!\left(
       (\log L)^2+\sqrt{L}\log L
    \right).
    % Retain the existing label of F46 here.
\end{equation}
Consequently, the inference-only depth satisfies
\begin{equation}\label{eq:no_db}
    D(\mathcal{B})=\mathcal{O}(\sqrt{L}\log L).
    % Retain the existing label of F47 here.
\end{equation}

% This results in the overall upper bound
% \begin{equation}\label{eq:no_db}
%    D(\mathcal{B}) = \mathcal{O}(\sqrt{L}\,\log L).
% \end{equation}

To determine whether the complexity in Eq.~\eqref{eq:no_db} is sublinear, we analyze the asymptotic behavior of the ratio between Eq.~\eqref{eq:no_db} and the linear term $L$ as $L \to \infty$:

% \begin{equation}\label{eq:sublinear_L_lhopital}
% \lim_{L\to\infty} \frac{\sqrt{L}\log L}{L}
% = \lim_{L\to\infty} \frac{\log L}{\sqrt{L}}
% = 0 .
% \end{equation}

\begin{multline}\label{eq:sublinear_L_lhopital}
\lim_{L\to\infty} \frac{\sqrt{L}\log L}{L}
= \lim_{L\to\infty} \frac{\log L}{\sqrt{L}}
\overset{\text{L'H\^{o}pital}}{=} %\\
\lim_{L\to\infty} \frac{\tfrac{1}{L}}{\tfrac{1}{2\sqrt{L}}}
= \lim_{L\to\infty} \frac{2}{\sqrt{L}} = 0 .
\end{multline}

% \begin{equation}\label{eq:sublinear_L_lhopital}
% \lim_{L\to\infty} \frac{\sqrt{L}\log L}{L} 
%     = \lim_{L\to\infty} \frac{\log L}{\sqrt{L}} 
%     \;\overset{\text{L'H\^{o}pital}}{=}\; 
%     \lim_{L\to\infty} \frac{\tfrac{1}{L}}{\tfrac{1}{2\sqrt{L}}} 
%     = \lim_{L\to\infty} \frac{2}{\sqrt{L}} 
%     = 0.
% \end{equation}

This establishes that the depth of the brickwork graph for this algorithm is asymptotically sublinear in $L$. In particular, the circuit depth grows strictly slower than the cost of iteratively searching through the database by classical comparison of entries. Hence, this sub-algorithm constitutes a meaningful addition to the flexibility and applicability of the quantum recommendation system introduced in~\cite{Sawerwain_Wróblewski_2019}. Results of our numerical experiments are shown in Fig. \ref{fig:res_qrs_no_db_plots}.

\begin{figure}[ht!]
\centering
\includegraphics[width=\linewidth]{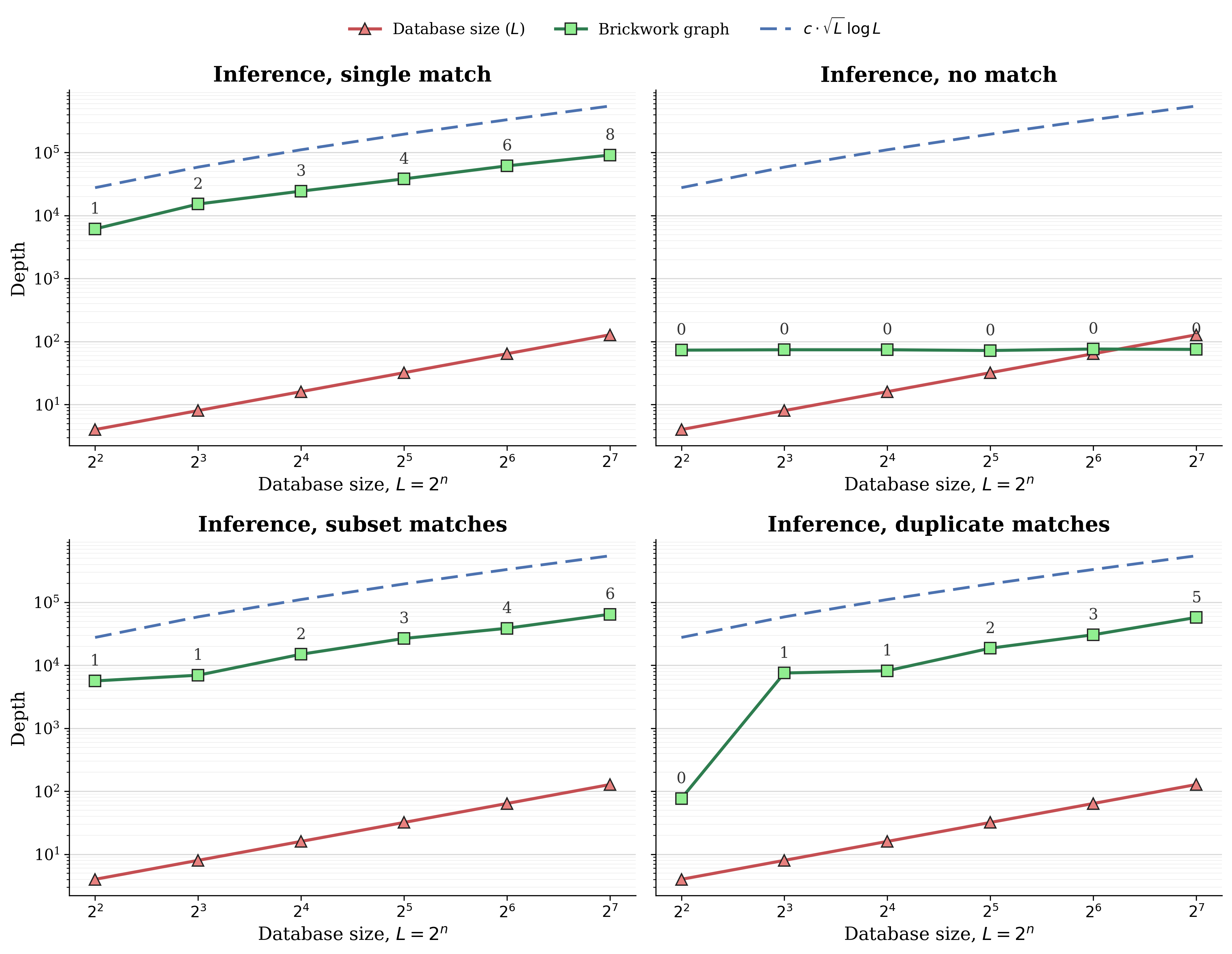}
\caption{Scaling behavior of brickwork graph depths obtained after transpiling the quantum recommendation system circuit, excluding the database encoding stage. Each point represents the resulting brickwork graph depth for a given input size, illustrating how transpilation cost grows with system scale, corroborating $\mathcal{O}(\sqrt{L}\,\log L)$. Notably, the number of Grover amplification rounds depends on the structure of the data; when there are more matches (i.e. clusters of low Hamming distance), fewer rounds are required, since $r=\mathcal{O}(\sqrt{L/g})$.}
\label{fig:res_qrs_no_db_plots}
\end{figure}

\subsection{Blind Quantum Generative Transformers}

Blind quantum machine learning is motivated by the imperative to ensure that emerging computational paradigms align with societal value hierarchies, particularly the growing emphasis on privacy. This motivation extends naturally to the design of privacy-preserving, or ``blind'', large language models (LLMs). Given the societal impact of classical LLMs, it is of significant theoretical interest to investigate quantum analogues that could, in principle, offer both computational advantage and privacy guarantees. Although the direct simulation of large-scale quantum LLMs is presently infeasible due to the prohibitive number of qubits required, it remains instructive to analyze their asymptotic resource requirements. In particular, using the brickwork depth bounds for blind algorithms presented in this paper can provide insight into the practical feasibility of privacy-preserving quantum generative models in NISQ regimes~\cite{Cherrat_2024}. This section will only establish bounds analytically using Theorem~\ref{thm:bw_general_alignment}. To verify that the theorem applies, we analyze one foreign gate\footnote{A gate that was not taken into account for the proof of Theorem~\ref{thm:bw_general_alignment}.}, namely $RBS(\theta)$. According to~\cite{Cherrat_2024}, $RBS(\theta)$ can be decomposed efficiently into four $\mathrm{H}$, two $R_y$, and two $\mathrm{CZ}$ gates---all covered in Theorem~\ref{thm:bw_general_alignment}. Therefore, the theorem applies.

In~\cite{Cherrat_2024}, three quantum transformer architectures are introduced: the \emph{Orthogonal Patch-wise Transformer}, the \emph{Quantum Orthogonal Transformer}, and the \emph{Compound Transformer}. Each is designed to implement the self-attention mechanism with logarithmic or polylogarithmic circuit depth. The central innovations are efficient quantum matrix-data loaders and trainable orthogonal layers, both of which contribute to reductions in runtime complexity and the number of trainable parameters relative to the classical Vision Transformer (ViT). For more background information, refer to ~\cite{Cherrat_2024}.

While all three architectures share similar design principles, they differ in their degree of quantum integration. The Orthogonal Patch-wise and Quantum Orthogonal Transformers operate on block-diagonal subspaces, corresponding to partially quantum realizations of the attention mechanism. In contrast, the \emph{Compound Transformer} merges all attention operations into a single coherent quantum process acting jointly on the patch and feature registers. As a result, it provides the most expressive and fully quantum variant among the proposed designs.

Let $n$ denote the number of input patches and $d$ the feature dimension per patch. Following the constructions in~\cite{Cherrat_2024}, we summarize the resource requirements of blind quantum transformer architectures using gate count as the measure of computational complexity.

% \paragraph{Circuit width.}
% Each quantum layer operates on a joint state representing both patch indices and patch features. Under unary amplitude encoding, the circuit width scales linearly as
% \begin{equation}
%     \omega(n, d) = \mathcal{O}(n + d),
% \end{equation}
% where $n$ is the size of the index register and $d$ the size of the feature register. This represents the maximum number of logical qubits required to implement any of the proposed transformer layers.

\paragraph{Circuit width.}
The quantum circuit implementing a Compound Transformer layer acts on a joint register representing patch indices and patch features. Under unary amplitude encoding, \(n\) input patches and \(d\) features per patch require \(n\) and \(d\) qubits, respectively, giving
\begin{equation}
    \omega(n,d)=\mathcal{O}(n+d).
\end{equation}

\paragraph{Complexity.}
\cite{Cherrat_2024} characterizes the quantum architectures by their circuit depths and the number of distinct circuits required per attention layer, which together determine runtime under the chosen execution model. The Compound Transformer requires only one distinct circuit per attention layer, with circuit depth
\begin{equation}
  T \;=\; \mathcal{O}\big(\log n + n\log d + \log(n+d)\big).
  \label{eq:qct_depth}
\end{equation}
The Compound Transformer acts on the full $(n+d)$-qubit register and subsumes the
other two architectures as special cases: the Orthogonal Patch-wise and Quantum
Orthogonal Transformers correspond to block-diagonal restrictions of this fully
entangled layer. Consequently, the Compound Transformer provides a natural asymptotic
upper bound for the complexity of all proposed quantum transformer variants.

\paragraph{Blind complexity.}
Theorem~\ref{thm:bw_general_alignment} bounds the brickwork depth by the circuit
width times the gate count. Since \cite{Cherrat_2024} reports circuit depths,
we use $G \le \omega T$ (at most $\omega$ gates per layer), which gives
\begin{equation}
  D(\mathcal{B}) \;=\; \mathcal{O}(\omega G)
  \;=\; \mathcal{O}\big(\omega^{2} T\big)
  \;=\; \mathcal{O}\big((n+d)^{2}\,(\log n + n\log d + \log(n+d))\big).
  \label{eq:qct_blind}
\end{equation}
This is an asymptotic upper bound on the brickwork depth of the Compound Transformer
of \cite{Cherrat_2024}, loose by up to a factor $\omega$ when layers contain
fewer than $\omega$ gates, and it bounds the entire family of blind quantum transformer
designs adapted from \cite{Cherrat_2024}.

\end{appendices}

%%===========================================================================================%%
%% If you are submitting to one of the Nature Portfolio journals, using the eJP submission   %%
%% system, please include the references within the manuscript file itself. You may do this  %%
%% by copying the reference list from your .bbl file, paste it into the main manuscript .tex %%
%% file, and delete the associated \verb+\bibliography+ commands.                            %%
%%===========================================================================================%%

% \bibliography{sn-bibliography}% common bib file
%% if required, the content of .bbl file can be included here once bbl is generated
%%\input sn-article.bbl

\end{document}